\documentclass[11pt]{article}

\usepackage{fullpage}
\usepackage{amsmath,amsfonts,amsthm,mathrsfs,xspace,hyperref,graphicx,bbm}
\usepackage{endnotes}
\usepackage{bm}
\usepackage{enumitem}
\usepackage{amssymb,latexsym}
\usepackage{relsize}
\usepackage{cleveref}
\usepackage{stmaryrd}
\usepackage{xcolor}
\usepackage{comment}
\usepackage{tikz}
\usepackage{mdframed}
\usepackage{tabularx}
\usepackage{booktabs}

\mdfdefinestyle{figstyle}{linecolor=black!7, backgroundcolor=black!7, innertopmargin=10pt, innerleftmargin=25pt, innerrightmargin=25pt, innerbottommargin=10pt }

\usepackage[utf8]{inputenc}
\usepackage[backend=bibtex,  % Makes biblatex work correctly
style=alphabetic,      % Usual citation style
maxbibnames=99,      % To not have et al. in references
minalphanames=3,      % To get [BBC+01] instead of [Bea+01]
backref=true
style=plain
sorting=nyt]{biblatex}

\definecolor{weborange}{rgb}{.8,.3,.3}
\definecolor{webblue}{rgb}{0,0,.8}
\definecolor{internallinkcolor}{rgb}{0,.5,0}
\definecolor{externallinkcolor}{rgb}{0,0,.5}

\definecolor{DarkBlue}{rgb}{0,0,0.8}  % For better printing.
\definecolor{DarkOrange}{rgb}{0.8,0.4,0}  % For better printing.
\def\mylinkcolor{DarkBlue}

\hypersetup{%pagebackref, 
	colorlinks=true, urlcolor=\mylinkcolor,
  linkcolor=\mylinkcolor, citecolor=\mylinkcolor}

\newtheorem{theorem}{Theorem}[section]
\newtheorem{proposition}[theorem]{Proposition}
\newtheorem{lemma}[theorem]{Lemma}

\newtheorem{corollary}[theorem]{Corollary}
\newtheorem{definition}[theorem]{Definition}
\theoremstyle{definition}

\newtheorem{example}[theorem]{Example}

\DeclareMathOperator*{\Ex}{\mathbb{E}}

\newcommand{\ket}[1]{{|#1\rangle}}
\newcommand{\bra}[1]{{\langle#1|}}

\newcommand{\ketbra}[2]{|#1\rangle\! \langle #2|}
\newcommand{\Tr}{\mbox{\rm Tr}}
\newcommand{\Id}{\mathbbm{1}}

\newcommand{\inner}[1]{\langle #1 \rangle}

\newcommand{\Z}{\ensuremath{\mathbb{Z}}}

\newcommand{\Group}{\mathfrak{G}}
\newcommand{\Automor}{\boldsymbol{\alpha}}
\newcommand{\ModAuto}{\boldsymbol{\sigma}}
\newcommand{\Algrep}{\boldsymbol{\pi}}
\newcommand{\Grouprep}{\boldsymbol{\lambda}}

\newcommand{\eps}{\varepsilon}
\newcommand{\poly}{\mathrm{poly}}
\newcommand{\polylog}{\mathrm{polylog}}

\newcommand{\strategy}{\mathscr{S}}

\newcommand{\supp}{\mathrm{supp}}

\newcommand{\bZ}{\mathbb{Z}}

\newcommand{\cl}{\ell}

\newcommand{\class}[1]{\ensuremath{\mathsf{#1}}\xspace}
\newcommand{\MIP}{\class{MIP}}
\newcommand{\MIPco}{\class{MIP}^{\mathrm{co}}} 
\newcommand{\RE}{\class{RE}}
\newcommand{\coRE}{\class{coRE}}
\newcommand{\NEXP}{\class{NEXP}}

\newenvironment{gamespec}{
	\begin{mdframed}[style=figstyle]}{
\end{mdframed}}

\newcommand{\jqnote}[1]{\textcolor{red}{\small {\textbf{(Junqiao:} #1\textbf{)
  }}}}

\newcommand{\bofh}{\mathcal{B}(\mathcal{H})}

\newcommand{\alicealg}{\mathscr{A}}
\newcommand{\bobalg}{\mathscr{B}}

\usepackage{braket}
\usepackage{comment}

\newcommand{\tensor}{\otimes}

\newcommand\bC{\mathbb C}

\newcommand\bF{\mathbb F}

\newcommand\bN{\mathbb N}

\newcommand\bR{\mathbb R}

\newcommand\cA{\mathcal A}
\newcommand\cB{\mathcal B}
\newcommand\cC{\mathcal C}

\newcommand\cE{\mathcal E}
\newcommand\cF{\mathcal F}
\newcommand\cG{\mathcal G}
\newcommand\cH{\mathcal H}
\newcommand\cI{\mathcal I}

\newcommand\cL{\mathcal L}
\newcommand\cM{\mathcal M}

\newcommand\cR{\mathcal R}

\newcommand\cU{\mathcal U}

\newcommand\cX{\mathcal X}
\newcommand\cY{\mathcal Y}

\newcommand{\BofH}{\mathcal{B}(\mathcal{H})}

\newcommand{\deltasync}{\delta_{\text{sync}}}

\newcommand{\tensormax}{\tensor_{\text{max}}}

\newcommand{\eq}[1]{\hyperref[eq:#1]{(\ref*{eq:#1})}}
\renewcommand{\sec}[1]{\hyperref[sec:#1]{Section~\ref*{sec:#1}}}
\newcommand{\thm}[1]{\hyperref[thm:#1]{Theorem~\ref*{thm:#1}}}
\newcommand{\lem}[1]{\hyperref[lem:#1]{Lemma~\ref*{lem:#1}}}
\newcommand{\defn}[1]{\hyperref[def:#1]{Definition~\ref*{def:#1}}}
\newcommand{\prop}[1]{\hyperref[prop:#1]{Proposition~\ref*{prop:#1}}}
\newcommand{\cor}[1]{\hyperref[cor:#1]{Corollary~\ref*{cor:#1}}}
\newcommand{\fig}[1]{\hyperref[fig:#1]{Figure~\ref*{fig:#1}}}
\newcommand{\tab}[1]{\hyperref[tab:#1]{Table~\ref*{tab:#1}}}
\newcommand{\app}[1]{\hyperref[app:#1]{Appendix~\ref*{app:#1}}}
\newcommand{\chap}[1]{\hyperref[chap:#1]{Chapter~\ref*{chap:#1}}}
\newcommand{\factlink}[1]{\hyperref[fact:#1]{Fact~\ref*{fact:#1}}}
\newcommand{\proto}[1]{\hyperref[proto:#1]{Protocol~\ref*{fact:#1}}}
\newcommand{\examp}[1]{\hyperref[examp:#1]{Example~\ref*{fact:#1}}}

\begin{document}

\title{Tracially embeddable strategies:\\ Lifting $\MIP^*$ tricks to $\MIP^{co}$}
% Tracial embeddable stratagies for Quantum commuting model
% How to convert a MIP* protocol to a MIPco protocol
\author{
Junqiao Lin\thanks{Junqiao.Lin@cwi.nl} \\ CWI \& QuSoft, Amsterdam, The Netherlands
}

\maketitle

\begin{abstract}
	We prove that any two-party correlation in the commuting operator model can be approximated using a \textit{tracially embeddable strategy}, a class of strategies defined on a finite tracial von Neumann algebra, which we define in this paper. Using this characterization, we show that any approximately synchronous correlation can be approximated to the average of a collection of synchronous correlations in the commuting operator model. This generalizes the result from Vidick [JMP 2022] which only applies to finite-dimensional quantum correlations. As a corollary, we show that the quantum tensor code test from Ji et al. [FOCS 2022] follows the soundness property even under the general commuting operator model. 

	Furthermore, we extend the state-dependent norm variant of the Gowers-Hatami theorem to finite von Neumann algebras. Combined with the aforementioned characterization, this enables us to lift many known results about robust self-testing for non-local games to the commuting operator model, including a sample efficient finite-dimensional EPR testing for the commuting operator strategies. We believe that, in addition to the contribution from this paper, this class of strategies can be helpful for further understanding non-local games in the infinite-dimensional setting. 
\end{abstract}

\setcounter{tocdepth}{2}
\newpage
\section{Introduction}
A two-player non-local game $\cG$ is played between a referee and two cooperating but socially-distanced players, Alice and Bob, who can only communicate with the referee in the process of the game. A non-local game is specified by a tuple $(\cX,\cA,\mu,D)$, where $\cX$ is a finite set of questions, $\cA$ is a finite set of answers, $\mu$ is a probability distribution over $\cX \times \cX$ specifying the distribution over questions for the pair of players and $V : \cX^2 \times \cA^2 \to \{\text{win},\text{lose}\}$ is a function which decides whether if the player wins on the answer pairs $(a,b)$ given question pair $(x,y)$. In this game, the referee first samples a pair of questions $(x,y)$ according to $\mu$ and sends it to Alice and Bob, respectively. Then, the referee decides whether the players win or lose based on the function $V$ and the response $(a,b)$ given by the players. Famously, \cite{bellEinsteinPodolskyRosen1964a} showed that if the two players employ the laws of quantum mechanics in their strategy, certain scenarios can be won with higher probability. This implies that the players can achieve a stronger family of \textit{correlations} or bipartite probability distributions between them if quantum entanglement is involved in their strategy.

There are several ways to model quantum strategies. The most studied model is the tensor product strategy, where a strategy consists of two Hilbert spaces, $\cH_A$ and $\cH_B$, a joint entangled state $\ket{\psi} \in \cH_A \tensor \cH_B$, and two sets of POVM $\{ A_a^x\}_{a \in \cA} \subseteq \cB(\cH_A)$ and $\{B_b^y\}_{b \in \cA} \subseteq \cB(\cH_B)$ for questions $x,y \in \cX$. Interestingly, for certain classes of games, it is possible to deduce the structure of the tensor product strategy given how well it performs in the game. One such example is the CHSH game~\cite{mckagueRobustSelftestingSinglet2012}, where even strategies which have a success rate close to the optimal \textbf{must} use a state that is ``close" to an EPR pair shared between both players. This is a property known as ``robustness." Other games which also have such properties include the magic square game~\cite{wuDeviceindependentParallelSelftesting2016b} and some classes of linear constant system games~\cite{coladangeloRobustSelftestingLinear2019a}. This property is a useful building block in many computational tasks, such as device-independent cryptography~\cite{bowlesSelftestingPauliObservables2018, jainParallelDeviceIndependentQuantum2020}, zero-knowledge proof systems~\cite{broadbentQuantumDelegationOfftheshelf2024}, and plays a key part in showing the breakthrough result $\MIP^*=\RE$~\cite{jiMIPRE2022}. 

The commuting operator model is another natural choice to model quantum strategy, which is used in, e.g., the Haag-Kastler axioms in algebraic quantum field theory~\cite{haagAlgebraicApproachQuantum1964}. In this model, strategies consist of one (potentially infinite-dimensional) Hilbert space $\cH$, a vector state $\ket{\psi} \in \cH$, and two sets of POVMs $\{ A_a^x\}_{a \in \cA}$ and $\{B_a^x\}_{a \in \cA}$ on $\bofh$ indexed by $x \in \cX$ which are pairwise commuting (i.e. $A_a^x B_b^y = B_b^y A_a^x$ for all $(x,y,a,b) \in \cX^2 \times \cA^2$). Since two operators defined in different tensor factors commute with each other, the set of commuting operator strategies contains the set of tensor product strategies. Surprisingly, as proven in~\cite[Theorem 12.10]{jiMIPRE2022}, there exists a non-local game in which the optimal success rate overall commuting operator strategies is strictly greater than the optimal success rate overall tensor product strategies. This fact solves several fundamental problems, including Tsirelson's problem in quantum information~\cite{tsirelsonResultsProblemsQuantum1993} and Connes' embedding problem in functional analysis~\cite{connesClassificationInjectiveFactors1976}.

While we know that such an advantage exists mathematically for a non-local game, we currently do not know if we can realize this advantage experimentally. In light of this, a natural question is, under what scenarios does what we currently understand in the tensor product model continue to hold in the commuting operator model? Working in the tensor product model allows one to work with finite-dimensional strategies, which provide many useful tools, such as trace and density matrices. Many of these tools do not have a clear analogue in the commuting operator setting. One notable exception is when the players are limited to using \textit{synchronous strategies} since it can be studied through the language of tracial $C$*-algebras~\cite{paulsenEstimatingQuantumChromatic2016}. This gives a natural generalization to the trace in the infinite-dimensional Hilbert space. This fact allows these robustness statements to be proven the same way as their finite-dimensional counterpart (see, for example,~\cite [Theorem 3.1]{mousaviNonlocalGamesCompression2022}). 
\subsection{Our Results}
In this paper, we consider a subclass of the commuting operator strategy known as \textit{tracially embeddable strategies}. Informally, this is a class of strategies in which Alice's observable is realizable in the standard representation of a tracial von Neumann algebra $\alicealg$, Bob's observable is realizable in the commutant $\alicealg'$, and the joint state between the two players can be written as $\sigma \ket{\tau}$, where sigma is a positive element within $\alicealg$, and $\ket{\tau}$ is the resulting vector state for the trace of $\alicealg$ which arises from standard form. 

Readers who are not familiar with von Neumann algebras may find the finite-dimensional setting of the above definition easier to understand. Let $\alicealg$ be the algebra of $n \times n$ matrices, or $\alicealg = \cM_n(\bC)$. In this example $\alicealg$ admits a unique tracial state, $\tau(X) = \frac{1}{n} \Tr(X)$, or the normalized trace. The standard form for $\cM_n(\bC)$ is represented as $\cM_n(\bC) \tensor \cI_n$, where the finite-dimensional matrices are embedded within the first register in this system. In this example, the commutant of $\alicealg$, $\alicealg'$ is $\cI_n \tensor \cM_n(\bC)$, or the second register of the given system. In this case, the tracial state can be represented as the maximally entangled state, or $\ket{ME} = \frac{1}{\sqrt{n}} \sum_{i=1}^n \ket{i}\ket{i}$. To make this connection more concrete, we see that
\begin{equation}
	\tau(X) = \braket{ME| X \tensor \cI |ME}
\end{equation}
for any matrices $X \in  \cM_n(\bC)$. For any unit vector $\ket{\psi} \in \bC^n \tensor \bC^n$, we can always find some positive matrices $\sigma \in \cM_n(\bC)$, namely the canonical square root\footnote{Where recall in quantum information, for a positive element $\sum_{i = 0}^n \alpha \ketbra{\psi_i}{\psi_i}$ in finite dimension where $\{ \ket{\psi_i} \}_{i \in [n]}$ is a set of orthogonal basis, the canonical square root refers to the following matrix $\sum_{i = 0}^n \sqrt{\alpha} \ketbra{\psi_i}{\psi_i}$.} of the density matrix on the first register, with $\ket{\psi} = \left(\sigma \tensor \cI\right) \ket{ME}$. From the quantum information point of view, this class of strategies has many more similarities to a finite-dimensional tensor product strategy, and we refer to~\Cref{exam:StandardformFD} for a more comprehensive example in the finite-dimensional case. 

As a main theorem of this paper (\Cref{thm:tracialeEmbedding}), we show that the closure of the set of correlations generated by a tracially embeddable strategy (in the $\ell_1$ norm sense) is equal to the set of commuting operator correlations. In other words, for showing theorems with $\eps$-dependency on the quantum correlation (such as robustness), we can assume the underlying strategy is tracially embeddable without a loss of generality. This result can also be seen as a generalization of the synchronous strategies introduced in~\cite{mousaviNonlocalGamesCompression2022} to commuting operator strategies as a whole since many of the simplifications which arise from synchronous strategies also apply to working with tracially embeddable strategies. By working with this class of strategies, we generalize the following results about the tensor product strategies to the commuting operator strategies:

\paragraph{Rounding near synchronous correlations.}The first result we generalize is the ``rounding" theorem. Recall that a correlation is synchronous if the players will always give the same answer given the same question. The rounding theorem states the following:

\begin{theorem}[Rounding for approximately synchronous correlation, informal] \label{thm:MainRoundinginformal}
	Let $\cG = (\cX, \cA, \mu, D)$ be a synchronous game, and let $\{C_{x,y,a,b}\}$ be a $\delta$-approximately commuting operator synchronous correlation for the game $\cG$. Then there exists a collection of commuting synchronous correlations $\{C_{x,y,a,b}^{\lambda}\}$ indexed by $\lambda \in [0, \infty)$ and a distribution $P$ with support over $[0, \infty)$ such that
	\begin{equation*}
		\Ex_{(x,y) \sim \mu} \sum_{a,b}\left|C_{x,y,a,b} - \int_{0}^{\infty} P(\lambda) \cdot C_{x,y,a,b}^{\lambda} d \lambda \right|_1  \leq O(\poly(\delta)).
	\end{equation*}
	Furthermore, if $C_{x,y,a,b}$ is realizable by a tracially embeddable strategy $\strategy$ defined within the tracial von Neumann algebra $\alicealg$. Then $C_{x,y,a,b}^{\lambda}$ can be chosen to be correlations realizable by a set of projective, symmetric and tracially embeddable strategy $\strategy^{\lambda}$ defined within $\alicealg$ which is proportionally close to $\strategy$.
\end{theorem}
%Such results are already known in the case where the players are restricted to the set of finite-dimensional correlations \cite{vidickAlmostSynchronousQuantum2022},
\noindent The formal statement for this theorem is presented in \Cref{thm:MainRounding}. Since a convex combination of any synchronous correlations is still a synchronous correlation, the first part of~\Cref{thm:MainRoundinginformal} shows that any near synchronous correlation can be well-approximated by a synchronous correlation. We remark that the tensor product version of~\Cref{thm:MainRoundinginformal} is already proven in~\cite{vidickAlmostSynchronousQuantum2022}. \cite{paul-paddockRoundingNearoptimalQuantum2022} contains a similar result in the context of representation for game algebras. Both results rely on a well-known joint distribution trick due to Connes \cite{connesClassificationInjectiveFactors1976}. Using the tracially embeddable strategies, one can translate this trick to the commuting operator model. This generalization was obtained independently in~\cite{delasalleALMOSTSYNCHRONOUSCORRELATIONS2023}. Their proof relies on a generalization of Connes' joint distribution to the type III von Neumann algebra case using the Connes-Tomita-Takesaki theory. Interestingly enough, we also make use of the Connes-Tomita-Takesaki theory in the proof of~\Cref{thm:tracialeEmbedding} in order to reduce the type III von Neumann algebra case to the type $\text{II}_{\infty}$ case. In comparison to~\cite{delasalleALMOSTSYNCHRONOUSCORRELATIONS2023}, we consider our proof to be more ``natural" to the quantum information community, as working with a tracially embeddable strategy enables us to run a similar analysis as~\cite{vidickAlmostSynchronousQuantum2022} for the commuting operator model. 

\Cref{thm:MainRoundinginformal} allows us to generalize many statements about the structure of the quantum synchronous strategies to general strategies. One example is the quantum soundness for the Quantum Tensor Code (QTC) test~\cite{jiQuantumSoundnessTesting2022} (\Cref{cor:soundnesstensorcode}). Roughly speaking, quantum soundness is a weaker notion of robustness for quantum strategy, where if a strategy succeeds in a game with high probability, then some ``hidden"  measurement must be made between the players. \cite{jiQuantumSoundnessTesting2022} shows that if the player succeeds on the ``Quantum tensor code test" with high probability using commuting synchronous strategy, then they must perform a ''hidden"  measurement of some set of POVMs which outputs a tensor code. This result is subsequentially generalized to the case where the players are given access to general (non-synchronous) tensor-product strategies using the rounding theorem in~\cite{vidickAlmostSynchronousQuantum2022}. By using~\Cref{thm:MainRoundinginformal}, we make the same generalization to general commuting operator strategies in a similar manner as~\cite{vidickAlmostSynchronousQuantum2022}.

We remark that the quantum low-individual degree test, a specific example of the Quantum tensor code test, plays an important part in quantum complexity theory. Recall that the complexity class $\MIP$ refers to the class of problems decidable by a multiplayer interactive proof system, and this is famously shown to be $\NEXP$, the class of problems decidable by a non-deterministic experiential time Turing machine~\cite{babaiNondeterministicExponentialTime1991a}. The low-individual degree test, a variant of the tensor code test, is an important step in showing $\MIP = \NEXP$. The same protocol is an important component in showing $\NEXP \subseteq \MIP^*$~\cite{vidickThreeplayerEntangledXOR2016, natarajanTwoplayerEntangledGames2018}\footnote{We remark as pointed out by~\cite{vidickErratumThreePlayerEntangled2020}, the original analysis for the low-individual degree test contains a mistake, which is subsequentially fixed in~\cite{jiQuantumSoundnessTesting2022, vidickAlmostSynchronousQuantum2022}.}, where recall $\MIP^{*}$ is the class of problem decidable by a multiplayer interactive proof system where the provers have access to the tensor product model of entanglement. In conjunction with the variant of PCP of proximate presented in~\cite[Chapter 10]{jiMIPRE2022}, quantum soundness of the low-individual degree test immediately implies that $\MIP = \NEXP \subseteq \MIPco$, where $\MIPco$ corresponds to a multiplayer interactive proof system with the commuting operator model of entanglement. This is the first non-trivial lower bound for the complexity class $\MIPco$ in the literature. We discuss the further implications for our work for determining the complexity of $\MIPco$ in~\Cref{sec:open_problems}. 

\paragraph{EPR pair testing in the commuting operator model.}Working with tracially embeddable strategies, we generalize some of the known results on robust EPR testing from the tensor product model to the commuting operator setting. 

There are several protocols which perform EPR pairs verification in the tensor product model (see, for example,~\cite {natarajanQuantumLinearityTest2017a, jiMIPRE2022}). These tests typically require the verifier to perform cross-verification with some subset drawn uniformly from the EPR pair. Our work is based on Mikael de la Salle's elegant simplification of the ``answer reduction" protocol from~\cite{jiMIPRE2022}. In this work, de la Salle observes that if the (anti-)commutation relationship holds between a certain subset of observables in the Pauli Basis test, then these relationships will hold overall observables up to a correction term which depends on the ``spectral gap" of these subsets. We expand on this idea and introduce the $\mu$-dependent Pauli Basis test (\Cref{fig:PauliBasis}), whereby the referee can, instead, pick any arbitrary distribution $\mu$ over the subset of qubits to be used for cross-checking. We show that the robustness of this game has a polynomial dependency on the optimal success rate and the \textit{spectral gap} of $\mu$, a constant (which could potentially be infinite) that depends on the choice of $\mu$ (see \Cref{sec:spectralgap} for more details). Informally, we show the following theorem:
\begin{theorem}[Rigidity of the $\mu$-dependent Pauli basis test, informal] \label{thm:RigPaulibasisinf}
	Let $\mu$ be a distribution over $\{0,1\}^n$, and $\kappa(\mu)$ be the spectral gap of $\mu$. Then there exists a game $\cG(\mu)$ with the following property:
	
	For any tracially embeddable strategies $\{ A_a^x\}, \{ B_b^y\} \subseteq \bofh$, $\ket{\psi} \in \cH$ which succeed $\cG(\mu)$ with success rate $1 - \eps$, there exist two isometries $V_{A}: \cH \rightarrow  \cH \tensor \bC^{2^{2n}}$ and $V_{B}: \cH \rightarrow \cH \tensor \bC^{2^{2n}}$ with $(V_B \tensor \cI_{2^{2n}}) V_{A} = ( V_{A} \tensor \cI_{2^{2n}} )V_B$, and a state $\ket{\text{Aux}} \in \cH  \tensor \cI_{2^{2n}}$ such that
	\begin{equation*}
		\left\| \left(  V_B \tensor \cI_{2^{2n}}  \right) V_{A} \ket{\psi} -\ket{\text{Aux}}  \ket{\text{EPR}}^{\tensor n} \right\|^2 \leq O\left(\poly(\kappa(\mu), \eps )\right),
	\end{equation*}
\end{theorem}
To our knowledge, this is the first robust self-testing result for the commuting operator model in the literature, as de la Salle only focuses on cases where players are restricted in synchronous strategies. We refer to~\Cref{sec:PBremark} for further discussion on how this work compares to existing literature on EPR testing.

As a part of proving~\Cref{thm:RigPaulibasisinf}, we also reprove a version of the Gowers-Hatami theorem by Vidick~\cite{vidickExpositoryNoteQuantum2018} in~\Cref{thm:SGHinTracialVNA}. This is a key tool for many robust self-testing theorems in the literature (see, for example~\cite{coladangeloRobustSelftestingLinear2019a, broadbentQuantumDelegationOfftheshelf2024}). In some sense, this can be seen as an easier version of the Gowers-Hatami theorem from~\cite{dechiffreOperatorAlgebraicApproach2019} since we are only considering an approximation for a finite group in the tracial von Neumann algebra setting. In contrast, the previously mentioned work considers amenable groups represented in a general von Neumann algebra setting. However, our version has a structure similar to Vidick's version. It is thus easier to use when lifting other self-testing results from the tensor product model to the commuting operator model.

\subsection{Open problems} \label{sec:open_problems}
\paragraph{The complexity of $\class{MIP^{co}}$.} It was shown in~\cite{jiMIPRE2022} that $\MIP^*=\RE$, where $\RE$ is the class of problems reducible to the Halting problem. A question left open by their work is to characterize the complexity class $\MIPco$. It is conjectured that $\MIPco=\coRE$, where $\coRE$ denotes the complement of $\RE$.\footnote{Note that the ${co}$ on either side of the equation refer to different things!}. We note that a zero gap variant of the $\MIPco$ has already been proven to be equal to $\coRE$~\cite{slofstraTsirelsonProblemEmbedding2019, mousaviNonlocalGamesCompression2022}. Showing that $\MIPco=\coRE$ would imply that estimating the commuting operator value of a game up to $\frac{1}{2}$ accuracy is computationally equivalently to estimating it up to $\eps > 0$ accuracy for any constant $\eps$ (as both problem would be equal to $\coRE$)!

A plausible approach to proving the $\MIPco=\coRE$ conjecture was outlined by~\cite[Section 1.2.3]{mousaviNonlocalGamesCompression2022} in which a similar gapped compression theorem needs to be developed for the commuting operator model as the tensor product model in $\MIP^*=\RE$. Since constructing the gapped compression theorem is not the main focus of this paper, we highlight the necessary pieces currently unproven in the literature below. 
\begin{itemize}
	\item An efficient, robust EPR tester. This is essential for the ``introspection" step of $\MIP^*=\RE$ as it allows the verifier to reduce sampling complexity within the compression theorem while maintaining a polynomial soundness/completeness gap. 
	\item Quantum soundness for the classical low-degree test, or in other words, giving the provers the power of commuting operator of entanglement, does not give them any noticeable advantage in the classic low-degree test. This is essential in $\MIP^*=\RE$ as it allows the verifier to use probabilistically checkable proofs (PCP) to reduce verification complexity while maintaining a polynomial soundness/completeness gap.
	\item A strong parallel repetition theorem for the commuting operator model. Roughly speaking, a strong parallel repetition theorem states that if the optimal success rate for the game is less than 1, then forcing the players to play $n$ instances of the game in parallel should result in an exponential decay rate relative to $n$. This is used by the verifier to amplify the soundness/completeness gap back to a constant after the two transformations above. 
\end{itemize}
In this paper, we show the existence of points one and two in \Cref{cor:PauliBasis} and \Cref{thm:soundnesstensorcode}, respectively, and point three will be shown in an upcoming manuscript. This is a huge step in proving $ \MIPco=\coRE $, and we expect the tracially embeddable strategy to be a key part of proving this conjecture.

\paragraph{Exact tracially embeddable strategies. }In~\Cref{thm:tracialeEmbedding}, we only show that the \textit{closure} of the set of correlations generated by a tracially embeddable strategy $\cC_{qc}^{Tr}$ is equal to the set of commuting operator correlations. Thus, a natural question would be whether the correlation set $\cC_{qc}^{Tr}$ is closed, in the $\ell_1$ sense, at all. We suspect the answer is no since to get the positive $\sigma$ in the proof of~\Cref{thm:tracialeEmbedding}, an approximation theorem must be applied to the states in $\cL^{1}(\alicealg, \tau)$ (\Cref{prop:state_on_positive_cone}). Since the set of commuting operator correlation is closed, a negative answer to the above question would imply that there exists some quantum commuting operator correlation which a tracially embeddable strategy cannot realize. 

Another interesting question is whether the same question above is true when we relax the definition of tracially embeddable strategy and allow the state to be any arbitrary vector state $\ket{\psi} \in \cH$ instead, where $\cH$ is the Hilbert space where the standard form of $\alicealg$ is being represented on. This would allow theorems which require the correlations to be exact to assume the
underlying strategy is a tracially embeddable strategy similar to this paper.
%This makes~\Cref{prop:state_on_positive_cone} unnecessary in the proof of~\Cref{thm:tracialeEmbedding}.  

\paragraph{Implementation of a tracial state within the NPA hierarchy.} In \cite{navascuesConvergentHierarchySemidefinite2008b}, a hierarchy of semidefinite programs is used in order to determine the optimal success rate for a non-local game using a hierarchy of semidefinite programs. Roughly speaking, level $n$ of the hierarchy optimizes over the gram matrix $\{ S \ket{\psi}\}$ where $S$ consist of products of observables of length $n$ from Alice or Bob. By using~\Cref{thm:tracialeEmbedding}, it is possible to define a variant of the NPA hierarchy where $\sigma$ is considered to be an additional variable that can be added to each level $n$, as well as an additional constraint to be added using the cyclicity property of trace. Does this variant of the hierarchy require fewer levels for convergence, or can it be potentially more efficient when the hierarchy does converge?

\paragraph{Operator algebraic formulation of self-testing.} Roughly speaking, a self-test refers to the uniqueness of quantum strategies given a specific correlation set or correlation sets close to it in the approximate case. It was recently shown in~\cite{paddockOperatorAlgebraicFormulationSelftesting2024} that having an exact self-test is equivalent to having a unique abstract state acting on the game algebra. However, it is unclear what the definition of robust self-testing should be in the operator algebraic framework.

In this paper, we take advantage of the main theorem and define a notion of robust self-testing of any approximated correlation for the commuting operator model in~\Cref{thm:RigPaulibasis}. Informally, in our definition, a game has the rigidity property if, for all tracially embeddable strategies, there exists a pair of partial isometry which has a commutation property (corresponding to Alice and Bob performing the partial isometry locally), which maps the state used within the approximated strategy to a vector state which is close to the ideal tracially embeddable strategy. 
We remark that this notion of self-testing is similar to the notion of self-testing defined in~\cite{natarajanQuantumLinearityTest2017a,coladangeloRobustSelftestingLinear2019a}.

The proof of the main theorem could also give some potential insight into defining robust self-testing in the operator algebraic framework. The proof mainly relies on Haagerup's reduction~\cite[Theorem 2.1]{haagerupReductionMethodNoncommutative2010}, which, roughly speaking, approximates the crossed product of the game algebra by a sequence of increasing tracial von Neumann algebras. A potential definition for robust self-testing could be: for any other correlation close to the given correlation, the approximate correlation must be realized by a state on a tracial von Neumann, which is close to the unique state composed with the normal conditional expectation. Since this is beyond the scope of this paper, we will instead leave it as an open problem about the proper definition of robust self-testing in the algebraic framework and whether it is equivalent to the definition we have given in~\Cref{thm:RigPaulibasis}.

%A potential definition for robust self-testing in the operator framework could be having a unique state on the game algebra, and any other correlation which is closed to such correlation must use a state which is  (in the $\cL^1$ sense) of unique state on one such increasing tracial von Neumann algebra. 

%and we hope this can give some insight into defining robust self-testing in the operator algebraic framework.

\paragraph{Generalizing other known facts about $\MIP^*$ to $\MIPco$.} In~\cite{natarajanQuantumFreeGames2023a}, Natarajan and Zhang considered the ``compiled" version of a non-local game, whereby a homomorphic encryption scheme is used to reduce a two-player non-local game into one player. They show that even in the compiled version of the CHSH game, the optimal success value for the finite-dimensional strategy is the same as in the compiled transformation. Can we use~\Cref{thm:tracialeEmbedding} to make a similar, meaningful observation for the infinite-dimensional case? Can we generalize other results related to $\MIP^*$ onto $\MIPco$ (for example, \cite{natarajanBoundingQuantumValue2023b}, \cite{dongComputationalAdvantageMIP2023}) in a meaningful way?

\paragraph{Acknowledgments.}  I thank William Slofstra, Henry Yuen, and Hamoon Mousavi for their involvement in many discussions of this project. I thank Narutaka Ozawa for pointing out the possibility of proving \Cref{thm:tracialeEmbedding} and the numerous email exchanges on this topic. I thank the anonymous reviewer from Annales Henri Poincar\'e for pointing out~\cite[Theorem 2.1]{haagerupReductionMethodNoncommutative2010}. I thank Mikael de la Salle for some helpful discussions regarding von Neumann algebras, and I thank Se-Jin Kim for pointing out mistakes in an earlier version of the proof of~\Cref{thm:tracialeEmbedding}. All these discussions are incorporated within~\Cref{sec:Tracialemd} of this paper. I thank the anonymous reviewers from QIP, TQC and Annales Henri Poincar\'e for their helpful comments in the earlier draft of this paper. Furthermore, I thank Yuming Zhao and Alexander Frei for the helpful discussions. I thank Eric Culf, Francisco Escudero Gutiérrez, Harold Nieuwboer, Māris Ozols and Sebastian Zur for going through earlier versions of this paper. Any remaining error within this paper is the fault of my own. Part of this work was completed as J.L. MSc thesis at IQC, University of Waterloo, and this work is supported by ERC STG grant 101040624-ASC-Q. 

\section{Preliminaries}

\subsection{Non-local games, strategies and correlations}
A \textit{two-player non-local game} is described by a tuple $\cG = (\cX^2, \cA^2, \mu, D)$, where $\cX$ is a finite set denoting the questions, and $\cA$ is another finite set denoting the answers, $\mu$ is a distribution over $\cX^2$, and $D$ is known as the evaluation which maps $\cX^2 \times \cA^2 \rightarrow \{0,1\}$. The game is played between two cooperating players, Alice and Bob, and a referee. In this game, the referee first samples $(x,y)$ according to the distribution $\mu$ and sends $x$ to Alice and $y$ to Bob. Upon receiving their questions, Alice (and resp.~Bob) must, without communicating with the other player, respond with answers $a$ (resp.~$b$) in $\cA$ back to the referee, and the players win if and only if $D(x,y,a,b) = 1$. Note that conventionally, non-local games are usually expressed with players sharing different question and answer sets. However, by forcing the probability distribution $\mu$ to be zero on certain questions, these two formulations are equivalent. In this paper, we consider a class of non-local games known as \textit{synchronous games}, where for all $x \in \cX$,  $D(x,x,a,b) = \delta_{a,b}$, or in other words, when given the same question, Alice and Bob must respond with the same answer.

Let $\cH$ be a (potentially infinite-dimensional) Hilbert space, and let $\BofH$ denote the set of the bounded operators on $\cH$. If $\cA$ is a finite set, then a \textit{positive operator-valued measure (POVM)} on $\cH$ with outcome set $\cA$ is a collection of positive operators $\{ A_a \}_{a \in \cA}$ in $\BofH$, such that $\sum_{a \in \cA} A_a = \Id_{\cH}$. A \textit{projection-valued measure (PVM)}, or \textit{projective measurement}, is a POVM where each of the operators $A_a$ is a projection operator (i.e.~$A_a^2 = A_a$). A \textit{tensor product strategy} for a non-local game $\cG = (\cX^2, \cA^2, \mu, D)$ is a tuple $\strategy = (\cH_{A} \tensor \cH_B, \ket{\psi},\{A_a^x\}_{a \in \cA} , \{B_b^x\}_{b \in \cA})$, where $\cH$ is a Hilbert space, $\ket{\psi} \in \cH_A \tensor \cH_B$ is a unit vector, for each $x \in \cX$, the set $\{A^x_a\}_{a \in \cA}$ (resp.~$\{B_b^x\}_{b \in \cA}$) is set of POVM in $\cB(\cH_{A})$(resp.~$\cB(\cH_{B})$). A tensor product strategy is said to be \textit{finite-dimensional} if both $\cH_A$ and $\cH_B$ are finite-dimensional. 

In contrast, a \textit{commuting operator strategy} for a non-local game $\cG = (\cX^2, \cA^2, \mu, D)$ is a tuple $\strategy = (\cH, \ket{\psi},\{A_a^x\}_{a \in \cA} , \{B_b^x\}_{b \in \cA})$, where $\cH$ is a Hilbert space, $\ket{\psi} \in \cH$ is a unit vector, for each $x \in \cX$, both $\{A^x_a\}_{a \in \cA}$ and $\{B_b^x\}_{b \in \cA}$ are sets of POVMs in $\BofH$, such that $[A_a^x, B_b^y] = 0$ for all $x, y \in \cX$ and $a, b \in \cA$. A commuting operator strategy is a generalization of the tensor product strategy. Unless otherwise stated, we assume all quantum strategies in this paper to be in the commuting operator model. A strategy is said to be \textit{projective} if $A$ and $B$ are collections of PVMs. The \textit{winning probability} of the strategy $\strategy$ is
\begin{equation*}
	\omega^{co}(\cG, \strategy) := \sum_{(x,y) \in \cX^2} \mu(x,y) \sum_{(a, b) \in \cA^2} D(x, y, a, b) \braket{\psi|A_a^x B_b^y |\psi}
\end{equation*}
and the \textit{commuting operator value} of the game $\cG$ is
\begin{equation*}
	\omega^{co}(\cG) := \sup_{\strategy} \omega^{co}(\cG,\strategy),
\end{equation*}
where the supremum is over all commuting operator strategies for game $\cG$. Each commuting operator strategy naturally induces a \textit{correlation set} $C$, which is a collection of real numbers defined by  
\begin{equation*}
	C = \left\{C_{x,y,a,b} = \braket{\psi|A_a^x B_b^y|\psi} \right\}_{(x,y, a, b) \in \cX^2 \times  \cA^2} \in  [0,1]^{|\cX|^2 \cdot  |\cA|^2}.
\end{equation*}
In a correlation set, $C_{x,y,a,b}$ corresponds to the probability in which Alice and Bob output the answer pair $(a,b)$ given the question pair $(x,y)$. The set of all correlations generated by a commuting operator strategy is denoted as $C_{qc}(\cX, \cA) \subseteq [0,1]^{|\cX|^2 |\cA|^2}$, and we use $C_{qc}$ to denote the union of $C_{qc}(\cX, \cA)$ for all possible $\cX$ and $\cA$. Synchronous games give a natural set of correlations known as \textit{synchronous correlations}, where a correlation set $\{C_{x,y,a,b} \}$ is synchronous if and only if for all $x\in \cX$, the probability value $C_{x,x,a,b} = 0$ whenever $a \neq b$, and we write $C_{qc}^{s}(\cX, \cA) \subseteq C_{qc}(\cX, \cA)$ (and correspondingly $C_{qc}^{s}$) to denote the set of synchronous commuting operator correlations. 

In this paper, we also consider the set of correlations that are approximately synchronous. Given a synchronous game $\cG$ and a correlation set $C$, we denote the \textit{synchronicity} of the correlation set to be
\begin{equation} \label{eq:syncgame}
	\deltasync(\cG, C) := \max\{\Ex_{x \sim \mu_x} \sum_{a \neq b} C_{x,x,a,b} , \Ex_{y \sim \mu_y} \sum_{a \neq b} C_{y,y,a,b}\}, 
\end{equation}
where $\mu_x$ (resp. $\mu_y$) denotes the marginal distribution of $x$  (resp. $y$) over $\mu$. This quantity, intuitively, corresponds to the probability in which the players answer the synchronous question incorrectly. We define a correlation $C$ to be \textit{$\delta$-synchronous} if $ \deltasync(\cG, C)  \leq \delta$, and we use $C_{qc}^{\delta}(\cX, \cA)$ to denote this set of correlations. A commuting operator strategy is defined to be a \textit{synchronous strategy} (resp.~$\delta$-synchronous strategy) if it realizes a synchronous correlation (and resp.~$\delta$-synchronous correlation) for the game $\cG$, and we use $\deltasync(\cG, \strategy) $ to denote the synchronicity of the correlation set generated by a strategy $\strategy$. For the remainder of this paper when discussing $\delta$-synchronous correlations with respective to some distribution, we assume that $\Ex_{y \sim \mu_y} \sum_{a \neq b} C_{y,y,a,b} \leq \Ex_{x \sim \mu_x} \sum_{a \neq b} C_{x,x,a,b}$ (i.e. $\deltasync(\cG, C) = \Ex_{x \sim \mu_x} \sum_{a \neq b} C_{x,x,a,b}$ ) for simplicity.
\subsection{String and finite fields and Pauli matrices} \label{sec:PrestringandPauli}
%Let $q = 2^n$; we recall that $\bF_2$ is a subfield $\bF_{2}^n$ and forms a vector space of dimension n over $\bF$. This connection gives a bijection between elements in $\bF_{2}^n$ to $\{0,1\}^{n}$ such that group addition is preserved, and we will use this representation to express elements within $\bF_{2}^n$. 

For a bit string $t, a, b \in \{0,1\}^n$, we use $|a|$ to denote the hamming weight of $a$ and $a \cdot b$ to denote the inner product between a and b, or $\sum a_i b_i \text{ mod } 2$. Furthermore, we use $t|_a \in \{0,1\}^{|a|}$ to denote the substring of $t$ indexed by $a_i = 1$.  Given a distribution $\mu$, we use $\Ex_{\mu}$ to denote the expectation over the distribution $\mu$ and for a set $S$, we use $\Ex_{x \in S}$ to denote the expectation over the set $S$. We use $\bF_{2} = \{0,1\}$ to denote the finite field of two elements in this paper. 

Given a finite group $G$, we recall the left regular representation for $G$ to be the map $L_{G}: G \rightarrow \cB(\bC^{|G|})$, $L_{G}(g) \ket{h} = \ket{gh}$. The left regular representation is also a unitary representation for the group $G$. We use $\ket{EPR}$ to denote the state $\frac{\ket{00} + \ket{11}}{\sqrt{2}}$ in $\bC^2 \tensor \bC^2$ in this paper. We also use $\rho_X$ and $\rho_Z$ to denote the following Pauli operators 
\begin{equation*}
	\rho_X = \begin{pmatrix}
		0 & 1 \\ 
		1 & 0
	\end{pmatrix}, \qquad 	\rho_Z = \begin{pmatrix}
		1 & 0 \\ 
		0 & -1
	\end{pmatrix},
\end{equation*}
and we see that $\rho_X^2= \rho_Z^2 = \cI_2$. Furthermore, we define the eigenspace of $\rho_X$ and $\rho_Z$ to be
\begin{equation*}
	\rho_0^X =  \frac{1}{2} \begin{pmatrix}
		1 & 1 \\ 
		1 & 1
	\end{pmatrix}, \qquad	\rho_1^X = \frac{1}{2} \begin{pmatrix}
		1 & -1 \\ 
		-1 & 1
	\end{pmatrix}, \qquad
	\rho_0^Z = \frac{1}{2} \begin{pmatrix}
		1 & 0 \\ 
		0 & 0
	\end{pmatrix}, \qquad\rho_1^Z =\frac{1}{2} \begin{pmatrix}
		0 & 0 \\ 
		0 & 1
	\end{pmatrix}
\end{equation*}
and we see that $\rho_X = \rho_0^X - \rho_1^X$ and $\rho_1^Z = \rho_0^Z - \rho_1^Z$. For $W \in \{X,Z\}$ and $a \in \{0,1\}^n$, we define $\rho_a^W = \bigotimes_{i = 1}^n  \rho_{a_i}^W$ to be the eigenspace for the matrices $\rho_W^{\tensor n}$, and we use $\rho_{W}^{\tensor n}$ measurement to denote the PVM measurement $\{ \rho_a^W \}_{a \in \{0,1\}^n}$. Furthermore, for $a \in \bF_{2}^n$, we define $\rho_W(a)$ to be $\sum_{b \in \bF_{2}^n} (-1)^{a \cdot b} \rho_{b}^{W}$, where, to abuse notation, $a$ and $b$ are defined as both as an element in $\bF_{2}^n$ and $\{0,1\}^n$ via the bijection mentioned earlier. We see that for both $W \in \{X, Z\}$, the set $\{\rho_W(a)\}_{a \in \{0,1\}^n}$ are unitary representations for the group formed by the addition structure $\bF_{2}^n$. For $a, b \in \{0,1\}^n$, we see that $\rho_X(a) \rho_Z(b) = (-1)^{a \cdot b} \rho_Z(b) \rho_X(a)$, and hence $\{ \rho_X^a \}_{a \in \{0,1\}^n} \cup \{ \rho_Z^b \}_{b \in \{0,1\}^n}$ generates a unitary representation for the well-known order n Weyl-Heisenberg group $H^{(n)}$. Finally, we recall that the left regular representation for the group $H^{(n)}$ is
\begin{equation*}
	L_{\rho^X(a) \rho^Z(b)} = \rho^X(a) \rho^Z(b) \tensor \cI_{2^n},
\end{equation*}
with $\ket{e} = \ket{EPR}^{\tensor n}$. 

\subsection{Von Neumann algebras} \label{sec:VNAintro}

%We further assume all $C$*-algebras are \textit{unital} in this paper, meaning the algebra admits an identity element $\cI_{\bofh}$.

Let $\cH$ be a Hilbert space, a (concrete, unital) $C$*-algebra $\alicealg \subseteq \bofh$ is a normed $*$-algebra with $\cI_{\cH} = \cI_{\alicealg}$ and closed in the norm topology. We use $\alicealg^{+}$ to denote the set of positive elements within $\alicealg$, (i.e.~elements of the form $s^*s$ for $s \in \alicealg$), and for $a, b \in \alicealg^{+}$, we say $a \leq b$ if $b - a \in \alicealg^{+}$.

A \textit{state} on a $C$*-algebra is a linear function $\psi: \alicealg \rightarrow \bC$, which is \textit{positive}, meaning that $\psi(a) \geq 0$ for all $a \in \alicealg^{+}$ and satisfies $\psi(\cI) = 1$. A state $\psi$ on $\alicealg$ is said to be \textit{faithful} if, for all $a \in \alicealg^{+}$, we have $\psi(a) = 0$ if and only if $a = 0$. Furthermore, we refer to the state to be a \textit{tracial state} if $\psi(st) = \psi(ts)$ for all $s,t \in \alicealg$. The famed GNS representation theorem states that every state $\psi$ on a $C$*-algebra $\alicealg$ induces a \textit{representation} (a $*$-homomorphism to some $\BofH$) $\pi_{\psi}$ onto $\cB(\cH_{\psi})$, and a unit vector $\ket{\psi} \in \cH_{\psi}$ such that $\psi(z) = \braket{\psi|\pi_{\psi}(z)|\psi}$ for all $z \in \alicealg$, and $\overline{\alicealg \ket{\psi}} = \cH$ (we refer to \cite[Theorem 4.5.2]{kadisonFundamentalsTheoryOperator1997} for more details about the GNS representation). This representation is specified with the triplet $(\pi_{\psi}, \cH_{\psi}, \ket{\psi})$. 

In this paper, we use the language of tracial von Neumann algebra to work with non-local games.Given $S \subseteq \BofH$, the weak-operator topology consist of neighbourhood of the form
\begin{equation*}
	\text{N}(a, \{ \ket{\psi_i}\}_{i \in [n]}, \cdots \ket{\phi_1}, \cdots \ket{\phi_n}, \eps) = \{ b : b \in \BofH, |\braket{\psi_i| (b-a)|\phi_i}| < \eps \text{ for all } i \in [n] \}
\end{equation*}
for $a \in S$,$n \in \bN$, $\ket{\psi_i}, \ket{\phi_i} \in \cH$ and $\eps >  0$. It is well known that the weak-operator topology is a finer topology than the norm topology.  A $C$*-algebra $\alicealg \subseteq \BofH$ is said to be a \textit{von Neumann algebra} if $\alicealg$ is closed under the weak-operator topology. Since the weak $*$-topology is more coarse than the norm topology, not every $C$*-algebra is a von Neumann algebra. For $X \subseteq \BofH$, the \textit{commutant} $X'$ of $X$ is defined to be the set of all elements which commute with $X$, or $X' := \{z \in \BofH : zw = wz \text{ for all } w \in X\}$. By the von Neumann bicommutant theorem an equivalent definition for von Neumann algebra  $\alicealg$ is being equal to it's double commutant $\alicealg''$. Given $X \subseteq \BofH$, we say $\alicealg$ is the von Neumann algebra generated by $X$ if $\alicealg = X''$, or $\alicealg$ is equal to the double commutant of $X$. 

In this paper, we work with \textit{separable} von Neumann algebra, meaning the underlying Hilbert space $\cH$ is separable. By~\cite[Proposition 3.19]{takesakiTheoryOperatorAlgebras2001} and the remark after~\cite[Proposition 13.1.4]{kadisonFundamentalsTheoryOperator1997}, a separable von Neumann algebra always admits a faithful normal state, where a state $\psi$ on $\alicealg$ is said to be \textit{normal} $\psi$ is continuos with respect to the weak operator topology. We use $\cU(\alicealg)$ to denote the set of unitary elements ($A^* A = A^* A= \cI$) in $\alicealg$.

%if for all bounded sequence of operators $\{A_{n}\}_n  \subseteq \alicealg^{+}$ with $A = \lim_{n \rightarrow \infty} A_n$, we have $\psi(A) =  \lim_{n \rightarrow \infty} \psi(A_{n})$. 

\begin{comment}
	cFor an element $a \in \alicealg$, we define the \textit{spectrum} of $a$ to be the closure of the set
	\begin{equation*}
		sp(A) := \{\lambda \in \bC | A - \lambda \cI \text{ is not invertible in } \alicealg \}. 
	\end{equation*}
	and we note that this definition is a generalization for eigenvalues of $\cM_n(\bC)$ in the finite-dimensional case.
\end{comment}

%increasing nets $\{A_{\lambda}\} \subseteq \alicealg^{+}$ with $A = \sup_{\lambda} \{ A_{\lambda} \}$, we have $\psi(A) = \lim \psi(A_{\lambda})$.

% and if $\psi$ is a normal state on a von Neumann algebra $\alicealg$, then the GNS representation $\phi: \alicealg \to\BofH$ is $\sigma$-weakly continuous

We referred to a von Neumann algebra to be \textit{tracial} if it admits a faithful normal tracial state $\tau$, and we use $(\alicealg, \tau)$ to emphasize the existence of $\tau$. The faithful trace $\tau$ naturally gives the notion of a $l_2$-norm on $\alicealg$, defined to be
\begin{equation*}
	||A||_2 := \sqrt{\tau(A^* A)}.
\end{equation*}
%By the Radon-Nikodym theorem (\cite[Theorem 7.4.9]{anantharaman_introduction_2010}),  for any normal state $\psi$, there exist an operator $x$ such that $\tau(|x|) < \infty$ such that $\psi(a) = \tau(xa)$. 
Recall that the \textit{standard form} for a tracial von Neumann algebra $(\alicealg, \tau)$ is the GNS representation triplet $(\pi_{\tau}, \cL^2(\alicealg, \tau), \ket{\tau})$ of $\alicealg$ for the tracial state $\tau$, where $\cL^2(\alicealg, \tau)$ denotes the Hilbert space for the representation. Note that the standard form for a tracial von Neumann algebra is unique up to canonical isomorphism \cite[Proposition 7.5.1]{anantharamanIntroductionII1Factors2010}. For simplicity of notation, if $(\alicealg, \tau)$ is in standard form, for each $A \in \alicealg$, we use $a$ to denote $\pi_{\tau}(A)$ as the operator defined within $\cB(\cL^2(\alicealg, \tau))$. 
For the remainder of this paper, unless otherwise specified, when discussing a tracial von Neumann algebra $\alicealg$, we assume $(\alicealg, \tau)$ is represented under the standard representation. 

Under the standard representation, the vector $\ket{\tau}$ is \textit{cyclic}, meaning that $\overline{\alicealg \ket{\tau}} = \cL^2(\alicealg, \tau)$, and \textit{separating}, meaning that for all $A \in \alicealg$, we have $A \ket{\tau} = 0$ if and only if $A = 0$. The separating condition also implies that each $A \in \alicealg$ corresponds to a unique vector $A \ket{\tau} \in \cL^2(\alicealg, \tau)$. Intuitively, one can view $\alicealg$ as a~\textit{left regular representation} on the Hilbert space $\cL^2(\alicealg, \tau)$. To be more precise, $A \in \alicealg$ and $(A \ket{\tau}) \in \cL^2(\alicealg, \tau)$, the operator $B \in \alicealg \subseteq \cL^2(\alicealg, \tau) $ performs the following operator on the Hilbert space $\cL^2(\alicealg, \tau)$
\begin{equation*}
B  (A \ket{\tau}) = (BA \ket{\tau}),
\end{equation*}
where one should interpret $(BA \ket{\tau})$ in the above equation as a vector in $\cL^2(\alicealg, \tau)$. 

Recall, given a von Neumann algebra $\alicealg$, the \textit{opposite algebra} $\alicealg^{op} := \{a^{op}: a \in \alicealg \}$ is a von Neumann algebra which has the same linearity as $\alicealg$, but has the opposite multiplication structure, or more precisely $(ab)^{op} = (b)^{op} (a)^{op}$. The algebra $\alicealg^{op}$ can also be faithfully embeddable onto $\cB(\cL^2(\alicealg, \tau))$ by
\begin{equation} \label{eq:oppemb}
	\pi_{\tau}^{op}(a^{op}) (\sigma  \ket{\tau}) = (\sigma a) \ket{\tau},
\end{equation}
for all $\sigma \in \alicealg$. This is known as the right regular representation for $\alicealg$. Clearly, $\pi_{\tau}^{op}(\alicealg)  \subseteq \alicealg'$, and in fact, $\alicealg^{'} = \alicealg^{op}$ \cite[Theorem 7.1.1]{anantharamanIntroductionII1Factors2010}. For simplicity of notation, we use $a^{op}$ to denote $\pi_{\tau}^{op}(a)$ in this paper. The map $op: a \rightarrow a^{op}$ forms a $*$-anti-isomorphism from $\alicealg \rightarrow \alicealg^{op}$. To see this, for all $a,b \in \alicealg$, $\lambda \in \bC$, we have 
\begin{equation*}
	(\lambda a + b)^{op} = \lambda a^{op} + b^{op}, \; \quad (ab)^{op} = b^{op} a^{op}.
\end{equation*}
To see that $op$ preserves the $*$ map, let $\sigma, \rho, b \in \alicealg$, we observe that
\begin{align*}
	\braket{\tau|\sigma^* (b^*)^{op} \rho|\tau } = \braket{\tau|\rho \sigma^* b^*|\tau } =  \braket{\tau| b^* \rho \sigma^*|\tau } =  \braket{\tau| (b^{op})^* \rho \sigma^*|\tau } = \inner{\tau| \sigma^* (b^{op})^* \rho |\tau }.
\end{align*} 
Since $\ket{\tau}$ is a cyclic vector for the algebra $\alicealg$ on the Hilbert space $\cL^2(\alicealg, \tau)$, by continuity $\braket{\psi|(b^*)^{op} |\phi} = \braket{\psi| (b^{op})^* |\phi}$ for all $\ket{\psi}, \ket{\phi} \in \cL^2(\alicealg, \tau)$,  which implies that $(b^{op})^*  = (b^*)^{op}$. As a consequence, the $op$ map preserves properties such as an element being unitary, projectors or positive from $\alicealg$ to $\alicealg'$. Since $\cI^{op} = \cI$, for a POVM (PVM) $\{A_a\} \subseteq \alicealg$, the set $\{(A_a)^{op}\} \subseteq \alicealg'$ will also be a POVM (PVM). 

For $\sigma \in \alicealg^{+}$ such that  $\tau(\sigma) = 1$, we define the $\sigma$-seminorm on $\alicealg$ as
\begin{equation}
	\| A \|_{\sigma} = \sqrt{\tau(A^* A \rho)}. 
\end{equation}
This can be seen as a finite von Neumann algebra analogue of the state-dependent norm defined within \cite[Section 1.3]{vidickExpositoryNoteQuantum2018}. By a quick calculation, we see that $\|UA\|_{\rho} = \|A\|_{\rho}$ for all $U\in \cU(A)$. 

We consider representations and near representations of finite group in this paper. For any finite group $G$ and map $\phi: G \rightarrow \cU(\alicealg)$ be a map, we say $\phi$ is a \textit{unitary representation} if $\phi(g) \phi(h) = \phi(gh)$ for all $g, h \in \cG$.  We say the map $\phi$ is an $(\eps, \sigma)-$representation for some $\sigma \in \alicealg^{+}$ with $\tau(\sigma) = 1$ if $\phi(g^{-1}) = \phi(g)^*$ and $\Ex_{g,h \in G}  \| \phi(g)\phi(h) - \phi(gh) \|_{\sigma} \leq \eps$. 

Since the proof of the main theorem requires more advanced techniques on von Neumann algebras, some of which might be inaccessible to the quantum information community, we separate some of the technical tools which are only used in the main theorem in~\Cref{App:Proofmain}. For a more comprehensive introduction, we refer the reader to \cite{blackadarOperatorAlgebras2006}.

\section{Tracially embeddable strategies} \label{sec:Tracialemd}
We start this section by introducing a special class of commuting operator strategies known as \textit{tracially embeddable strategies}, motivated by the standard form of a tracial von Neumann algebra.  
\begin{definition}[Tracially embeddable strategy] \label{def:Tracialemd}
	Let $\cG = (\cX^2, \cA^2, \mu, D)$ be a non-local game. A quantum commuting strategy $\strategy = (\cH, \ket{\psi},\{A_a^x \}_{a \in \cA} , \{B_b^y\}_{b \in \cA})$ for game $\cG$ is called \textbf{tracially embeddable} if there exists a tracial von Neumann algebra $(\alicealg, \tau)$ with standard form $(\pi_{\tau}, \cL^2(\alicealg, \tau), \ket{\tau})$ and $\sigma \in \alicealg^{+}$ such that $\cH = \cL^2(\alicealg, \tau)$, $\ket{\psi} = \sigma \ket{\tau}$, $\{A_a^x \}_{a \in \cA} \subseteq \alicealg$ and $ \{B_b^y\}_{b \in \cA} \subseteq \alicealg'$.
\end{definition}

%We note in the previous version of this paper $\sigma$ is not assumed to be positive within the definition of tracially embeddable strategies, and we refer to the remark at the end of section \Cref{App:tracialembproof} for more details.

In this paper, we use $\strategy = (\cL^2(\alicealg, \tau),\sigma \ket{\tau}, \{A_a^x \}, \{(B_b^y)^{op}\})$ to denote a tracially embeddable strategy on the tracial von Neumann algebra $(\alicealg, \tau)$\footnote{Note that the POVM $\{B_b^y\}$ is defined in $\alicealg'$ under this formulation.}. Furthermore, we call the set of possible correlations generated by a tracially embeddable strategy as \textit{tracially embeddable correlations}, and we use $\cC_{qc}^{Tr}(\cX, \cA) \subseteq \cC_{qc} (\cX, \cA)$ to denote the set of tracially embeddable correlation with input set $\cX$ and output set $\cA$. As the main theorem of this paper, we see that every correlation set $C_{qc}$ can be approximated by a tracially embeddable strategy through the following theorem. 
\begin{theorem}  \label{thm:tracialeEmbedding}
	Let $\cX$ and $\cA$ be arbitrary finite sets, then
	\begin{equation*}
		\overline{\cC_{qc}^{Tr}(\cX, \cA)} = \cC_{qc}(\cX, \cA).
	\end{equation*}
	where the closure above is in the $l_1$ norm sense. 
\end{theorem}
\noindent We give a proof for~\Cref{thm:tracialeEmbedding} in the next section. Intuitively, tracially embeddable strategies have a similar structure as finite-dimensional tensor product strategies. Let $\strategy = (\cL^2(\alicealg, \tau),\sigma \ket{\tau}, \{A_a^x \}$, $\{(B_b^y)^{op}\})$ be a tracially embeddable strategy on the tracial von Neumann algebra $(\alicealg, \tau)$. We can think of $\alicealg$(resp.~$\alicealg'$) to be Alice's (resp.~Bob's) algebra, similar to the local Hilbert space structure in the tensor product model. Since $1 = \braket{\tau| \sigma^2 |\tau} = \| \sigma\|_2^2$ and $\sigma \in \alicealg^{+}$, we can intuitively think of $ \sigma^2$ as the reduced density matrix on Alice's side. Furthermore, we can use the opposite algebra map introduced in the previous section to perform ``observable switching", or to move Bob's measurement operator to Alice's algebra using the vector state $\ket{\tau}$. To make the discussion more concrete, we consider the following example:
\begin{comment}
Since we are mainly considering correlation sets, a common calculation for $\strategy$ is
\begin{equation} \label{eq:Examplecalculation}
	\braket{\tau|\sigma A_a^x(B_b^y)^{op}\sigma|\tau} = 	\braket{\tau|\sigma A_a^x \sigma (B_b^y)^{op}|\tau}  = \braket{\tau|\sigma A_a^x \sigma B_b^y|\tau} \leq \sqrt{\braket{\tau|\sigma (A_a^x)^2\sigma|\tau}}\sqrt{\braket{\tau|\sigma (B_y^b)^2\sigma|\tau}},
\end{equation}
where the first line follows from $(B_b^y)^{op} \subseteq \alicealg'$, second from the definition of the opposite map, and the third inequality follows from Cauchy-Schwarz's inequality. This is similar to many of the calculations in the finite-dimensional case. 
\end{comment}
\begin{example} \label{exam:StandardformFD}
	Let $\alicealg = \mathbf{M}_n(\bC)$, recall, the GNS representation using the trace $\Tr(\cdot)$ maps $\alicealg$ to $\alicealg \tensor \cI_n \in \mathbf{M}_{n^2}(\bC)$, with the linear functional $\Tr(\cdot)$ gets mapped to the maximally entangled state $\ket{\tau} = \frac{1}{\sqrt{n}} \sum_{i=0}^n \ket{ii}$. As a sanity check, we see that $\Tr(A) = \braket{\tau|(A \tensor \cI_n)|\tau}$ for all $A \in \mathbf{M}_n(\bC)$. We note that $\ket{\tau}$ is a tracial state for $\mathbf{M}_n(\bC) \tensor \cI_n$ by definition, and for all $\ket{\psi} \in \bC^{n^2}$, there exists some $\sigma \in \mathbf{M}_n(\bC)$ such that $(\sigma \tensor \cI_n) \ket{\tau} = \ket{\psi}$, showing that $\ket{\tau}$ is also cyclic. In this case, $\sigma \sigma$ is also the reduced density matrix for $\ket{\psi}$ on the first register. For all $\sigma \tensor \Id_n \in  \mathbf{M}_n(\bC)$, the operator $\Id_n \tensor \cA^{T} \in \alicealg^{'}$ maps 
	\begin{equation*}
		(\Id_n \tensor A^{T}) \left((\sigma \tensor \Id_n )  \ket{\tau}\right) = (\sigma A \tensor \Id_n)\ket {\tau}.
	\end{equation*}
	Hence, we can define the opposite algebra in a similar fashion as \Cref{eq:oppemb}, where 
	\begin{equation*}
		(A \tensor \Id_n)^{op} = \Id_n \tensor \cA^{T}.
	\end{equation*} 
	For a non-local game $\cG = (\cX^2, \cA^2, \mu, D)$ and a finite dimensional strategy $(\{ A^x_a \tensor \Id_n \}, \{ \Id_n \tensor B_b^y\}, \ket{\psi})$ define on $\mathbb{C}^n \tensor \mathbb{C}^n$. The correlation generated by $(\{ A^x_a \tensor \Id_n \}, \{ \Id_n \tensor B_b^y\}, \ket{\psi})$ can be rewritten as
	\begin{equation*}
		C_{x,y,a,b} = \braket{\tau|(\sigma \tensor \Id_n) (A^x_a \tensor \Id_n) \left((B_b^y)^{T}\tensor \Id_n \right)^{op} (\sigma \tensor \Id_n) | \tau},
	\end{equation*}
	for some $\sigma \in \mathbf{M}_n(\bC)$. Although $\sigma$ does not necessarily have to be positive, by the polar decomposition of $\sigma = \sigma U$ for some unitary element $U$, by replacing Bob's observable by $(U^{T})^* B_b^y U^T$, we can, without a loss of a generality, assume that $\sigma$ is positive regardless. 
	%We remark that this is the finite dimensional equivalent of the left most expression from \eqref{eq:Examplecalculation}, which is how tracially embeddable strategy is expressed in this paper.
\end{example}
%In this paper, we mainly consider approximations correlations. As a consequences of~\Cref{thm:tracialeEmbedding}, we can, without a loss of generality, assume every quantum strategies which arises from $\cC_{qc}$ uses a tracially embeddable strategies for the remainder of this paper. 
Using the structure of tracially embeddable strategies, we introduce a natural generalization of symmetric strategy~\cite[Definition 2.3]{vidickAlmostSynchronousQuantum2022} for the commuting operator model.
\begin{definition}[Symmetric strategy] \label{def:symstrategy}
	Let $\cG = (\cX^2, \cA^2, \mu, D)$ be a synchronous game and let \\ $( \cL^2(\alicealg, \tau), \sigma \ket{\tau},\{ A_{a}^x \},  \{ (B_{b}^y)^{op} \} )$ be a tracially embeddable strategy. We call this strategy to be \textit{symmetric} if $A_{a}^x = B_{a}^x$ for all $x \in \cX$ and $a \in \cA$. 
\end{definition}
Symmetric strategies will be written as $( \cL^2(\alicealg, \tau), \sigma \ket{\tau},\{ A_{a}^x \})$ in this paper. If a symmetric strategy is also projective (i.e.~$\{A_x^a\}$ are projective) and $\sigma = \cI$, then the given strategy is also synchronous. By using \cite[Theorem 5.5]{paulsenEstimatingQuantumChromatic2016}, and taking the double commutant of the corresponding $C$*-algebra for Alice's measurement operator, we can always write any synchronous strategies as a projective, symmetric strategy of the form $( \cL^2(\alicealg, \tau), \ket{\tau},\{ A_{a}^x \})$. We prove some additional lemma about symmetric strategy on~\Cref{sec:symstragproof}. 

\subsection{Proof of the main theorem} \label{App:Proofmain}
Before proving the main theorem, we first introduce additional tools required to prove the main theorem in this subsection. We remark that understanding the rest of the paper, which is dedicated to potential applications of the main theorem, does not depend on understanding the technical details from this subsection.

\subsubsection{The crossed product construction and additional theorems on vNA}

% under the discrete topology (i.e. for every $x \in \bR$, the set $\{x\}$ is an open set under this topology) Recall, given a set $S$, the discrete topology of $S$ consist of open sets of the form

Recall, for a set $S$, the discrete topology is generated by open sets of the form $\{s\}$ for $s \in S$. Let $\Group \subseteq \bR$ be a countable subgroup of $(\bR,+)$ equipped with the discrete topology. Let $\Automor$ be a \textit{continuous automorphic representation} of $\Group$ on $\alicealg$, or a continuous map $\Automor: \Group \rightarrow \Automor_g$ where each $\alpha_g$ is a $*$-automorphism acting on $\alicealg$ indexed by $g \in \Group$. Let $\cl^2(\Group,\bC)$ denote the set of square integrable functions from $\Group$ to $\bC$. Since $\cG$ is a subgroup of $\bR$, for any function $\psi \in \cl^2(\Group,\bC)$, we have $\sum_{g \in \Group} \| \psi(g) \|^2 < \infty$. This means that any functions in $\cl^2(\Group,\bC)$ can be represented as elements in the Hilbert space $\oplus_{g \in \Group} \bC_g$ by $\psi \rightarrow \oplus_{g \in \Group} \psi(g)$. 

%Any operators $A \in \cB(\cl^2(\Group,\bC))$ acting on $\cl^2(\Group,\bC)$ can be represented by an infinite dimensional matrix indexed by $\Group$ over $\bC$ , or $(A(g,h))_{(g,h) \in \Group^2}$ where $A(g,h) \in \bC$ denotes the matrix entries index by $g,h \in \Group$. 
%, both potentially being unbounded,

%

We say a von Neumann algebra $\alicealg \subseteq \BofH$ is in standard form if it admits a cyclic and separating vector in $\cH$ for the algebra $\alicealg$ (see~\cite{arakiPropertiesModularConjugation1974a} for more details). We remark that in the case where $\alicealg$ is tracial, $\alicealg$ being in standard form is equivalent to $\alicealg$ represented under GNS representation of the trace function, which is consistent with the definition given in \Cref{sec:VNAintro}. Recall from~\cite{takesakiTomitaTheoryModular1970a} that any von Neumann algebra in standard form admits a modular operator $\Delta$ and a modular conjugation operator $J$ acting on $\cH$ such that 
\begin{equation*}
	\Delta^{it} \alicealg \Delta^{-it} = \alicealg, \quad J \alicealg J = \alicealg', \quad t \in \bR.
\end{equation*}
The above relation can be used to define a \textit{modular automorphism group} for any subgroup $\Group$ of $\bR$. More specifically, let $\text{Aut}(\alicealg)$ define the set of automorphism on $\alicealg$, or the set of map which maps $\alicealg$ to itself. For each $t \in \Group$, we define the function $\ModAuto_t \in \text{Aut}(\alicealg)$ as $\ModAuto_t(A) = \Delta^{it} A \Delta^{-it}$. The modular automorphism group is defined to be the continuous map $\ModAuto: \Group \rightarrow \text{Aut}(\alicealg)$ given by 
\begin{equation*}
	\ModAuto(t) = \ModAuto_t,
\end{equation*}
for all $t \in \Group$. We remark that since $\Group$ is assumed to be equipped with the discrete topology in this paper, $\ModAuto$ is continuous by default. 

%or a continuous automorphic representation $\ModAuto$ which maps elements of the group $\Group$ onto the set automorphism on $\alicealg$ defined by $\ModAuto_t(A) = \Delta^{it} A \Delta^{-it}$.

Recall from the literature that the \textit{crossed product} of $\alicealg$ by $\Group$ with respect to $\Automor$, denoted by $\cR = \alicealg \rtimes_{\Automor} \Group$ is a von Neumann algebra acting on $\cl^2(\Group,\cH)$ generated by $\{\Algrep_{\Automor} (A) \}_{A \in \alicealg}$ and $\{\Grouprep_{\Automor} (h) \}_{h \in \Group}$, where
\begin{equation*} 
	(\Algrep_{\Automor} (A) \zeta)(g) = \ModAuto_g^{-1}(A) \zeta(g) \quad  (\Grouprep_{\Automor} (h) \zeta)(g) = \zeta(g - h),\quad \text{for} \quad \zeta \in \cl^2(\Group,\cH), \quad g \in \Group.
\end{equation*}
Since we assume $\Group$ to be a countable subgroup of $\bR$ in this paper, if $\alicealg$ is a separable von Neumann algebra, then the von Neumann algebra $\alicealg \rtimes_{\Automor} \Group$ is also separable. Under this assumption, the generator for $\alicealg \rtimes_{\ModAuto} \Group$ can be specified by the operator $\left\{ 	\Algrep_{\ModAuto}(A)\right\}_{A \in \alicealg}$ and  $\left\{  \Grouprep_{\ModAuto} (h) \right\}_{h \in \Group}$, where 
\begin{equation*}
		\Algrep_{\ModAuto}(A)  = \sum_{t \in \Group} \Delta^{it} A \Delta^{-it} \tensor \ketbra{t}{t}, \quad  \Grouprep_{\ModAuto} (h) = \cI_{\alicealg} \otimes \sum_{t \in \Group} \ketbra{t-h}{t},
\end{equation*}
for $A \in \alicealg$ and $h \in \Group$. In this case, the map $\boldsymbol{\theta}: \alicealg \rightarrow \alicealg \rtimes_{\ModAuto} \Group$ given by
\begin{equation} \label{eq:isomapcrossedproduct}
	\boldsymbol{\theta} (A) = \Algrep_{\ModAuto}(A)
\end{equation}
gives $*$-isomorphism from $\alicealg$ to its crossed product construction with $\Group$. The following proposition extends a normal state from $\alicealg$ to $\alicealg \rtimes_{\ModAuto} \Group$. 
\begin{proposition}[Proposition 13.1.4 of~\cite{kadisonFundamentalsTheoryOperator1997a}] \label{prop:normalstatecrossed}
	Let $\Group$ be a subgroup of $\bR$, $\alicealg$ be a von Neumann algebra and $\ModAuto$ be the modular automorphism group over $\Group$. For every normal state $\rho$ acting on $\alicealg$, there is a normal state $\omega$ acting on $\alicealg \rtimes_{\ModAuto} \Group$ such that 
	\begin{equation} \label{eq:normalstatecrossed}
		\omega(A) = \rho\left( (\cI_{\alicealg} \otimes \bra{0_{\Group}}) A  (\cI_{\alicealg} \otimes \ket{0_{\Group}}) \right) 
	\end{equation}
	for all $A \in \alicealg \rtimes_{\ModAuto} \Group$, where $\boldsymbol{\theta}$ is the isomorphic map defined in \eqref{eq:isomapcrossedproduct}, and $0_{\Group}$ should be interpreted as the identity to the group $(\Group, +)$. 
\end{proposition}
Since $\Delta^{i0} A \Delta^{-i0} = A$ for all $A \in \alicealg$,  this implies that $ (\cI_{\alicealg} \otimes \bra{0_{\Group}}) \boldsymbol{\theta}(A)  (\cI_{\alicealg} \otimes \ket{0_{\Group}}) = A$ for all $A \in \alicealg$. This means that we can replace~\eqref{eq:normalstatecrossed} with
\begin{equation*} 
	\omega(A) = \rho\left( (\cI_{\alicealg} \otimes \bra{0_{\Group}}) \boldsymbol{\theta}(A)  (\cI_{\alicealg} \otimes \ket{0_{\Group}}) \right).
\end{equation*}
In other words, the normal state given by~\Cref{prop:normalstatecrossed} also preserves the action given by $\boldsymbol{\theta}$. 

\begin{comment}
	\begin{proof}
		This follows directly from~\cite[Proposition 13.1.4]{kadisonFundamentalsTheoryOperator1997a}, with the additional observation that the idenity
		The proof is identical to  except we have $\Algrep_{\ModAuto}(A) (0,0) = \ModAuto_0(A) = A$, where $\Algrep_{\ModAuto}(A) (0,0)$ $0 \in \Group \subseteq \bR$ in this context, is the identity to the group $(\bR, +)$. 
	\end{proof}
\end{comment}
%Recall from the literature that a von Neumann algebra $\alicealg$ is $\sigma$-finite if it admits a faithful normal state. By~\Cref{label}
 For a von Neumann subalgebra $\bobalg \subseteq \alicealg$, we say a completely positive contraction map $\boldsymbol{\Phi}: \alicealg \rightarrow \bobalg$ is a \textit{conditional expectation} from $\alicealg$ to $\bobalg$ if for all $X \in \alicealg$, $A,B \in \alicealg$, we have
\begin{equation*}
	\boldsymbol{\Phi}(AXB) = A \boldsymbol{\Phi}(X) B.
\end{equation*}
Furthermore, a conditional expectation is \textit{normal} if it preserves the weak topology. We recall the following reduction theorem below, which states that any separable von Neumann algebra under the crossed product of $\Group = \bigcup_{n \geq 1} 2^{-n} \cdot \bZ \subseteq \bR$ with respect to the modular automorphism group can be approximated by a sequence of tracial von Neumann algebras.
\begin{theorem}[Haagerup's reduction, Theorem 2.1 of \cite{haagerupReductionMethodNoncommutative2010}] \label{thm:Haagreduction}
	Let $\alicealg$ be a separable algebra von Neumann algebra in standard form, let $\Group = \bigcup_{n \geq 1} 2^{-n} \cdot \bZ \subseteq \bR$ and let $\ModAuto$ be the corresponding modular automorphism group of $\Group$. There exists an sequence $(\mathscr{R}_n, \tau_n)_{n \geq 1}$ of subalgebras of $\mathscr{R} = \alicealg \rtimes_{\ModAuto} \Group$ such that 
	\begin{itemize}
		\item $\mathscr{R}_n \subseteq \mathscr{R}_{n+1}$ for all $n \in \bN$. 
		\item Each $\mathscr{R}_n$ is a tracial von Neumann algebra with the normal trace being $\tau_n$.
		\item For each $n$, there exists a normal faithful conditional expectation $\boldsymbol{\Phi}_n$ from $\mathscr{R} \rightarrow \mathscr{R}_n$ with $\boldsymbol{\Phi}_n(A) + \boldsymbol{\Phi}_n(B) = \boldsymbol{\Phi}_n(A+B)$;
		\item $\bigcup_{n \geq 1} \mathscr{R}_n$ is $w^*$ dense within $\mathscr{R}$.
	\end{itemize}
\end{theorem}
%We remark that the conditional expectation in the above theorem being linear is not stated explicitly in the original theorem statement of \cite[Theorem 2.1]{haagerupReductionMethodNoncommutative2010}, but it follows directly from the definition of $\boldsymbol{\Phi}_n$ on page 8 of~\cite{haagerupReductionMethodNoncommutative2010}.
%However,recall from~\Cref{sec:VNAintro} that a separable von Neumann algebra is also $\sigma$-finite. 

%A separable von Neumann algebra always admit a normal faithful state, a property know as ``$\sigma$-finiteness" in the standard von Neumann algebra literature (see the remark after~\cite[Definition 5.5.14]{kadisonFundamentalsTheoryOperator1997} and~\cite[Proposition 3.19]{takesakiTheoryOperatorAlgebras2001} for more details).
 We remark that~\Cref{thm:Haagreduction} actually holds for $\sigma$-finite von Neumann algebra according to the original formulation. However, as seen in the remark after~\cite[Definition 5.5.14]{kadisonFundamentalsTheoryOperator1997}, all seperable von Neumann algebra is also $\sigma$-finite. Finally, we recall two theorems from the literature about tracial von Neumann algebra in standard form which we use in the proof of~\Cref{thm:tracialeEmbedding}. The first is an approximation theorem about normal linear function for tracial von Neumann algebra. 

\begin{proposition}[Proposition 7.3.4 and Theorem 7.3.8 of \cite{anantharamanIntroductionII1Factors2010}] \label{prop:state_on_positive_cone}
	For every positive normal linear functional $\psi$ on a tracial von Neumann algebra $\alicealg \subseteq \bofh$ in standard form, there exists some vector $\ket{\psi}\in \cH$ with $\psi(A) = \braket{\psi|A|\psi}$. Furthermore, there exists a sequence of positive operator $A_n \in \alicealg^{+}$ such that 
	\begin{equation*}
		\lim_{n \rightarrow \infty} || \ket{\psi} - A_n \ket{\tau} || = 0
	\end{equation*}
\end{proposition}

Recall, given two linear functionals $\psi_1$ and $\psi_2$ acting on a $C$*-algebra $\alicealg$, we say $\psi_1 \leq \psi_2$ iff $\psi_2 - \psi_1$ is positive. The second theorem is a variant of the Radon-Nikod\'{y}m theorem. 

\begin{theorem}[Proposition 7.3.5. of \cite{kadisonFundamentalsTheoryOperator1997}] \label{thm:RNforBob}
	Let $\alicealg \subseteq \BofH$ be a von Neumann algebra, and let $\omega$ be a positive linear functional acting $\alicealg$. Furthermore, there exists a vector state $\ket{\psi}$ such that $\omega \leq \omega_{\ket{\psi}}$, where $\omega_{\ket{\psi}}$ is the linear functional acting on $\alicealg$ such that $\omega_{\ket{\psi}} (A) = \braket{\psi|A|\psi}$. Then there is a positive operator $H^{'}$ in the unit ball $(\alicealg^{'})_1$ of $\alicealg^{'}$ such that 
	\begin{equation*}
		\omega(A) = \braket{\psi| H^{'}A|\psi}
	\end{equation*}
	for all $A \in \alicealg$.
\end{theorem}

\subsection{The proof}
We prove \Cref{thm:tracialeEmbedding} below. 

\begin{proof}
	Let $C_{x,y, a,b} \in C_{qc}$ be an arbitrary commuting operator correlation, and let 
	\begin{equation*}
		\strategy = (\{ \widetilde A_a^x\}, \{\widetilde B_b^y\}, \ket{\widetilde \psi}, \cH)
	\end{equation*}
	be the commuting operator algebra strategy which realizes $C_{x,y, a,b}$. By Fritz's characterization~\cite[Proposition 3.4]{fritzTsirelsonProblemKirchberg2012} of commuting operator correlation, since both $\cX$ and $\cA$ are finite sets, we can, without a loss of generality, assume $\cH$ is separable. Let $\mathscr{N}$ be the von Neumann algebra generated by $\{\widetilde A_a^x\}$, since $\cH$ is separable, $\mathscr{N}$ is also separable. This also implies that $\mathscr{N}$ admits a faithful normal state, and this gives a $\sigma$-weakly continuous isomorphism which maps $\mathscr{N}$ to its standard form. 
	
	Let $\Group = \cup_{n \geq 1} 2^{-n} \cdot \bZ \subseteq \bR$, and let $\ModAuto$ be the corresponding modular automorphism group for $\mathscr{N}$ and let $\mathscr{R} = \mathscr{N} \rtimes_{\ModAuto} \Group$. For each $(y,b) \in \cX \times \cA$, let $\widetilde\psi_b^y: \mathscr{N} \rightarrow \bC$ be the normal linear functional given by $\widetilde\psi_b^y(A)= \braket{\widetilde\psi|A \cdot \widetilde B_b^y|\widetilde\psi}$ and let $\widetilde\psi: \mathscr{N} \rightarrow \bC$ be the normal linear functional define by $\widetilde\psi(A)= \braket{\widetilde\psi|A|\widetilde\psi}$. Since each $\{\widetilde B^y_b\}_{b \in \cA}$ is a set of POVMs, we have $\sum_b \widetilde\psi_b^y =\widetilde\psi$ for all $y\in \cA$. By \Cref{prop:normalstatecrossed}, each $\widetilde\psi_b^y$ (reps. $\widetilde\psi$) extends into a normal state $\widetilde\psi_b^y$ (reps. $\widetilde\psi$) on $\mathscr{R}$ such that $\widetilde\psi_b^y(A) = \widetilde\psi_b^y(\boldsymbol{\theta}(A))$ for all $A \in \mathscr{N}$, where the map $\boldsymbol{\theta}$ is the map defined in~\eqref{eq:isomapcrossedproduct}. 
	
	Let $(\mathscr{R}_n, \tau_n)_{n \geq 1}$ be the sequence of tracial von Neumann subalgebra of $\mathscr{R}$ guarantee by \Cref{thm:Haagreduction}. For each $(y,b) \in \cX \times \cA$ and $n \in \bN$, define $\{\widehat C^n_{x,y,a,b}\}_{(a,x) \in \cX \times \cA}$ to be 
	\begin{equation}
		\widehat C^n_{x,y,a,b} = \widetilde\psi_b^y  \circ \boldsymbol{\Phi}_n  \circ \boldsymbol{\theta} (\widetilde A^x_a). 
	\end{equation}
	We see that
	\begin{align*}
		\sum_{a,b} \widehat C^n_{x,y,a,b}  &= \sum_{a,b} \widetilde\psi_b^y  \circ \boldsymbol{\Phi}_n  \circ \boldsymbol{\theta} (\widetilde A^x_a) \\
		&= \sum_{b} \widetilde\psi_b^y  \circ \boldsymbol{\Phi}_n  \circ \boldsymbol{\theta} ( \sum_a \widetilde A^x_a) \\
		&= \widetilde\psi  \circ \boldsymbol{\Phi}_n  ( \cI_{\mathscr{R}} ) 
	\end{align*}
	where the second line follows from the functional $\widetilde\psi_b^y$, $ \boldsymbol{\Phi}_n$ and $\boldsymbol{\theta}$ being linear. By property 1 and 4 of~\Cref{thm:Haagreduction}, the sequence $\{\boldsymbol{\Phi}_n  \circ \boldsymbol{\theta} (X) \}_{n \in \bN}$ is an increasing sequence of operators within $\mathscr{R}$ for all $X \in \mathscr{R}$ with $\lim_{n \rightarrow \infty} \boldsymbol{\Phi}_n  \circ \boldsymbol{\theta} (X) = \boldsymbol{\theta} (X)$. In particular, we have 
	\begin{equation*} 
		\lim_n \widetilde \psi  \circ \boldsymbol{\Phi}_n  ( \cI_{\mathscr{R}} )  = \widetilde \psi(\cI_{\mathscr{N}}) = 1.
	\end{equation*}
	Let 
	\begin{equation*}
		C^n_{x,y,a,b} = \frac{1}{\widetilde \psi  \circ \boldsymbol{\Phi}_n  ( \cI_{\mathscr{R}} ) } \widehat C^n_{x,y,a,b}.
	\end{equation*}
	Since $\sum_{a,b} C^n_{x,y,a,b} = 1$ for all $(x,y) \in \cX^2$ and $n \in \bN$, this shows that $\{C^n_{x,y,a,b}\}$ is a correlation set for each $n \in \bN$. Our goal is to first show that $C^n_{x,y,a,b}$ approximates $C_{x,y,a,b}$, then show that each $C^n_{x,y,a,b}$ can be approximated by a sequence of tracially embeddable strategies defined within $\mathscr{R}_n$. Since for each $(x,y,a,b) \in \cX^2 \times \cA^2$, the linear functional $\psi_x^a$ is normal, we have
	\begin{equation*}
		\lim_n  C^n_{x,y,a,b} = \frac{\lim_n \widetilde \psi_x^a \circ \boldsymbol{\Phi}_n  \circ \boldsymbol{\theta} (\widetilde A^x_a)}{	\lim_n \widetilde \psi  \circ \boldsymbol{\Phi}_n  ( \cI_{\mathscr{R}} )} = \widetilde \psi_x^a \circ  \boldsymbol{\theta}  ( \widetilde A^x_a)  = \braket{\widetilde \psi|\widetilde A^x_a \cdot \widetilde B_b^y|\widetilde \psi} = C_{x,y,a,b},
	\end{equation*}
	hence showing the sequence of correlation $C^n_{x,y,a,b}$ approximates $C_{x,y,a,b}$.  Now, we wish to show that each of the  $C^n_{x,y,a,b}$ can be approximated by a sequence of tracially embeddable strategies defined within $\mathscr{R}_n$ to complete the proof. Let $(\pi_{\tau_n}, \cL^2(\mathscr{R}_n, \tau_n), \ket{\tau_n})$ be the standard form for $(\mathscr{R}_n, \tau_n)$. Fix $n \in \bN$, since $\Phi_n$ is normal, $\frac{1}{\widetilde \psi  \circ \boldsymbol{\Phi}_n  ( \cI_{\mathscr{R}} ) }  \widetilde \psi  \circ \boldsymbol{\Phi}_n$ is a normal state on $\mathscr{R}$ (and hence for $\mathscr{R}_n \subseteq \mathscr{R}$). By~\cite[Theorem 6]{arakiPropertiesModularConjugation1974a}, there exists some vector state $\ket{\psi_n} \in \cL^2(\mathscr{R}_n, \tau_n) $ such that for all $X \in \mathscr{R}_n$,
	\begin{equation*}
		\frac{\widetilde \psi_n  \circ \boldsymbol{\Phi}_n(X)}{\widetilde \psi_n  \circ \boldsymbol{\Phi}_n  ( \cI_{\mathscr{R}} ) }   = \braket{\psi_n|X|\psi_n}.
	\end{equation*}
	Let $A^{x,n}_a = \boldsymbol{\Phi}_n \circ \boldsymbol{\theta}(\widetilde A^x_a) \in \mathscr{R}_n$, since $\boldsymbol{\Phi}_n$ is linear
	\begin{equation*}
		\sum_a A^{x,n}_a = \boldsymbol{\Phi}_n \circ \boldsymbol{\theta}(\sum_a \widetilde A^x_a) = 1
	\end{equation*}
	and hence each $\{A^{x,n}_a\}_{a \in \cA}$ is a POVM defined within $(\mathscr{R}_n)$. For each $(y,b) \in \cX \times \cA$, the state $	\frac{1}{\widetilde \psi  \circ \boldsymbol{\Phi}_n  ( \cI_{\mathscr{R}} ) }  \widetilde \psi_b^y  \circ \boldsymbol{\Phi}_n(A)$ is a normal linear functional with $\frac{1}{\widetilde \psi  \circ \boldsymbol{\Phi}_n  ( \cI_{\mathscr{R}} ) }  \widetilde\psi_b^y  \circ \boldsymbol{\Phi}_n \leq \frac{1}{\widetilde\psi  \circ \boldsymbol{\Phi}_n  ( \cI_{\mathscr{R}} ) }  \widetilde\psi  \circ \boldsymbol{\Phi}_n $. Hence, by \Cref{thm:RNforBob}, there exists a positive element $B^{y,n}_b \in  \mathscr{R}_n'$ such that for all $X \in \mathscr{R}$
	\begin{equation*}
		\frac{ \widetilde \psi_b^y  \circ \boldsymbol{\Phi}_n(X)}{\widetilde \psi  \circ \boldsymbol{\Phi}_n  ( \cI_{\mathscr{R}} ) }  = \braket{\psi_n|B^{y,n}_b X|\psi_n}.
	\end{equation*}
	Since for all $y \in \cX$, $X \in \mathscr{N}$, we have 
	\begin{equation*}
		\sum_b 	\frac{\widetilde \psi_b^y  \circ \boldsymbol{\Phi}_n \circ \boldsymbol{\theta}(X)}{\widetilde \psi  \circ \boldsymbol{\Phi}_n  ( \cI_{\mathscr{R}} ) }    = \frac{\widetilde \psi  \circ \boldsymbol{\Phi}_n\circ \boldsymbol{\theta}(X)}{\widetilde \psi  \circ \boldsymbol{\Phi}_n  ( \cI_{\mathscr{R}} ) }  , 
	\end{equation*} 
	this implies that 
	\begin{equation*}
		\sum_b \braket{\psi_n|B^{y,n}_b X|\psi_n} = \braket{\psi_n|X|\psi_n}
	\end{equation*}
	for all $X \in \mathscr{R}$. This implies that $\sum_b B^{y,n}_b \leq \cI_{\mathscr{R}_n}$ with $\braket{\psi_n|\left(\cI_{\mathscr{R}_n} - \sum_b B^{y,n}_b \right)|\psi_n} = 0$. By adding $(\cI_{\mathscr{R}_n} - \sum_b B^{y,n}_b)$ to one of the $B^{y,n}_b$,	we can further assume that each $\{B^{y,n}_b\}_{b \in \cA}$ forms a set of POVM in $(\mathscr{R}_n)'$. Finally, by \Cref{prop:state_on_positive_cone}, there exists a sequence of positive operator $\sigma_m \in \mathscr{R}_m$ such that 
	\begin{equation} \label{eq:stateapprox}
		\lim_{m \rightarrow \infty} \| \ket{\psi_n} - \sigma_m \ket{\tau_n} \| = 0.
	\end{equation}
	Since $\ket{\psi}$ is a vector state, without a loss of generality, we can also assume that $\tau_n(\sigma_m^2)= \|\sigma_m \ket{\tau_n} \| = 1$ by normalization. We see that for each $m$, $(\cL^2(\mathscr{R}_n, \tau_n), \sigma_m \ket{\tau_n}, \{A^{x,n}_a\}, \{B^{y,n}_b\})$ is a tracially embeddable strategy, and by \eqref{eq:stateapprox}, we have
	\begin{align*}
		&\lim_{m \rightarrow \infty} \sum_{x,y,a,b}| C^n_{x,y,a,b} - \braket{\psi_n| \sigma_m A^{x,n}_a B^{y,n}_b \sigma_m  |\psi_n} | \\
		&= \lim_{m \rightarrow \infty} \sum_{x,y,a,b} | \frac{1}{\widetilde \psi  \circ \boldsymbol{\Phi}_n  ( \cI_{\mathscr{R}} ) } \widetilde\psi_b^y  \circ \boldsymbol{\Phi}_n  \circ \boldsymbol{\theta} (\widetilde A^{x,n}_a)  - \braket{\psi_n| \sigma_m A^{x,n}_a B^{y,n}_b \sigma_m  |\psi_n} | \\
		&=  \lim_{m \rightarrow \infty} \sum_{x,y,a,b} |\braket{\psi_n| A^{x,n}_a B^{y,n}_b |\psi_n}  - \braket{\psi_n| \sigma_m A^{x,n}_a B^{y,n}_b \sigma_m  |\psi_n} | =0,
	\end{align*}
	showing that $C^n_{x,y,a,b} \in \overline{\cC_{qc}^{Tr}(\cX, \cA)}$, completing the claim. 

\end{proof}

\section{Rounding in the commuting operator model}
In this section, we show the rounding theorem stated below. Intuitively, it says that any almost synchronous commuting operator correlation can be approximated by a convex combination of synchronous correlations. Furthermore, if the underlying correlation is a tracially embeddable correlation defined within tracial von Neumann algebra $\alicealg$, then these synchronous correlations can be well-approximated by strategies defined within $\alicealg$.
\begin{theorem}[Rounding] \label{thm:MainRounding}
	Let $\cG = (\cX^2, \cA^2, \mu, D)$ be a synchronous game, and let $\{C_{x,y,a,b}\} \in \cC_{qc}(\cX, \cA)$ be a $\delta$-synchronous correlation for $\cG$. Then there exist a collection of commuting operator synchronous correlations $\{  C_{x,y,a,b}^{\lambda} \}_{\lambda \in [0, \infty]} \subseteq \cC_{qc}^{s} (\cX, \cA)$ and a probability distribution $P$ over $[0, \infty]$, such that
	\begin{equation*}
		\Ex_{(x,y) \sim \mu} \sum_{a,b}|C_{x,y,a,b} - \int_{0}^{\infty} P(\lambda) \cdot C_{x,y,a,b}^{\lambda} d \lambda|  \leq O(\delta^{\frac{1}{8}}).
	\end{equation*}
	Furthermore, if $C_{x,y,a,b}$ is realizable via a tracially embeddable strategy $\strategy = ( \cL^2(\alicealg, \tau), \sigma \ket{\tau},\{ A_a^x \}, \{B_y^b\})$. Then $C_{x,y,a,b}^{\lambda}$ can be chosen to be correlations realizable by a set of symmetric strategy $\{\strategy = ( \cL^2(P_{\lambda}\alicealg P_{\lambda}, \frac{1}{\tau}\tau), \sigma \ket{\tau},\{ A_a^{\lambda,x} \}\}$, where $\{P_{\lambda}\}_{\lambda} \in \alicealg$ is an increasing sequence of projectors in $\alicealg$ such that 
	\begin{equation*}
		\int^{\infty}_{0}  P_{\lambda} d \lambda = \sigma^2, 
	\end{equation*}
	and 
	\begin{equation*}
		\int_{0}^{\infty} \Ex_{x \sim \mu} \sum_{a}  || (A_a^x - A_a^{\lambda,x} ) P_{\lambda}||_2^2 d \lambda  \leq O(\delta^{\frac{1}{4}} + \delta).
	\end{equation*}
\end{theorem}
\noindent As seen in~\Cref{exam:StandardformFD}, any finite-dimensional strategy can always be converted into a tracially embeddable strategy. Hence, this theorem can be seen as a combination of Theorem 3.1 and Corollary 3.3 of~\cite{vidickAlmostSynchronousQuantum2022} for commuting operator correlations. We give a proof for \Cref{thm:MainRounding} in \Cref{App:Roundingproof}. Due to the structure of tracially embedded strategies, the proof of~\Cref{thm:MainRounding} follows mostly the same structure as \cite{vidickAlmostSynchronousQuantum2022}. We state some applications for the rounding theorem in the following subsection.

\subsection{Applications}

Similar to \cite[Theorem 3.1]{vidickAlmostSynchronousQuantum2022}, \Cref{thm:MainRounding} can be used to translate ``rigidity" statements related to synchronous games from synchronous strategies to general commuting operator strategies. In this section, we state two corollaries related to the rounding theorem.
\subsubsection{Quantum Soundness}  \label{sec:soundnesstensorcode}
The first application for our rounding theorem concerns quantum soundness. Typically, quantum soundness for a specific class of non-local games refers to a hidden ``global structure'' relating to any near-optimal quantum strategy. We start this section by first showing the specific corollary related to quantum soundness, which arises from the main theorem. Then, we will describe a potential use through the tensor code test \cite{jiQuantumSoundnessTesting2022}. We note that the following corollary is an analogue of \cite[Corollary 4.1]{vidickAlmostSynchronousQuantum2022} for the commuting operator model, and the proof follows a similar structure. 
\begin{corollary} \label{cor:corsoundness}
	Let $\cG = (\cX^2, \cA^2, \mu, D)$ be a synchronous game. Let $\cY$ and $\cB$ be two finite sets, and let $\rho$ be a probability distribution over $\cX \times \cY$, such that for any $\delta$-synchronous correlation $C_{x,y,a,b}$, the following condition about the marginal distribution of $\rho$ on $\cX$ holds:
	\begin{equation} \label{eq:corsoundnessre1}
		\Ex_{x \sim \rho}\sum_{a \neq b} C_{x,y,a,b} \leq O(\delta). 
	\end{equation}
	Furthermore, suppose there exists a set of functions $\{g_{xy} \}_{(x,y) \in \cX \times \cY} : \cA \rightarrow 2^{\cB}$, such that for every pair $(x,y) \in \cX \times \cY$, the sets $\{g_{xy} (a)\}_{a \in \cA}$ are pairwise disjoint. Moreover, suppose every synchronous strategy $\strategy^{\text{sync}} = (\cL^2(\bobalg, \tau_{\bobalg}), \ket{\tau_{\bobalg}},\{ A_{a}^x \})$ on a tracial von Neumann algebra $(\bobalg, \tau_{\bobalg})$ for $\cG$ follows the ``soundness property'':
	\begin{itemize}
		\item[] There exists a family of POVMs $\{ G_b^y \}_{(y,b) \in \cY \times \cB} \subseteq \bobalg$ and a convex, monotone, non-decreasing function $\kappa:[0,1] \rightarrow \bR_{\geq 0}$ such that 
		\begin{equation}\label{eq:corsoundnessre2}
			\Ex_{(x,y) \sim \rho} \sum_{a \in \cA} \braket{ \tau_{\bobalg}|   A_a^{x}  (G_{g_{xy}(a)}^y)^{op}| \tau_{\bobalg}} \geq O\left(\kappa(\omega^{co}(\cG, \strategy^{\text{sync}})) \right),
		\end{equation}
		with $G_{g_{xy}(a)}^y = \sum_{b \in g_{xy}(a)} G_b^y$, where $op$ above is with respect to $\cL^2(\bobalg, \tau_{\bobalg})$.
	\end{itemize}
	Then the ``soundness property'' extends to any arbitrary $\delta$-synchronous strategy $\strategy = (\cL^2(\alicealg, \tau),\sigma \ket{\tau}, \\ \{A_a^x \}, \{ B_{a}^x \})$ on the measurement operator $\{A_a^x \}$, with the $\kappa$ function replaced with 
	\begin{equation*}
		O(\kappa(\omega^{co}(G, \strategy)  - \poly(\delta)) - \delta^{\frac{1}{8}}).
	\end{equation*}
\end{corollary}

We give a proof for \Cref{cor:corsoundness} in \Cref{sec:roundingapp}. We remark that the $\kappa$ function in \Cref{cor:corsoundness} is often a complexity measurement function with respect to some $\eps$. The ``soundness property" above is an example of a ``typical" rigidity statement. In a typical use case for \Cref{cor:corsoundness}, the synchronicity of the strategy can be assumed to have the same complexity measure as the success rate of the strategy by adding a ``synchronicity test". More precisely, for a synchronous game $\cG = (\cX^2, \cA^2, \mu, D)$, a synchronicity test corresponds to the referee making the following change to the game's input distribution: with probability $c$, sample the marginal distribution $x \sim \mu$, replace the original question with the pair $(x,x)$ and send it to both players. This addition to the game guarantees that any $(1-\eps)$-optimal strategy for $\cG$ to be at least $\frac{\eps}{c}$-synchronous. 

Now, we show a typical use case for \Cref{cor:corsoundness}. Tensor code is a generalization of the Reed-Muller code, and it plays an important role in the computer science community, particularly within the literature for probabilistic proof systems due to the so-called \textit{axis-parallel line vs point test} (see, for example, \cite{babaiNondeterministicExponentialTime1991a}, \cite{aroraProbabilisticCheckingProofs1998}). Since the axis-parallel line vs point test is not an essential part of this paper, we refer the reader to \cite[Section 2]{jiQuantumSoundnessTesting2022}, for the definition of tensor code, as well as the tensor code version of the axis-parallel line vs point test. We use the same notation for the ``quantum axis-parallel line vs point test" as in \cite{jiQuantumSoundnessTesting2022} for the remainder of this section.

Similar to \cite[Corollaryg 4.1]{vidickAlmostSynchronousQuantum2022} for the finite-dimensional case, \Cref{cor:corsoundness} can be used to generalize the soundness result from synchronous strategies to general commuting operator strategies. To illustrate this, we recall a simplified version of the main result from \cite{jiMIPRE2022}. 

\begin{theorem}[Quantum soundness of the tensor code test, synchronous case, informal] \label{thm:soundnesstensorcode}
	Let $\cC $ be an interpolatable $[n,k,d]_{\Sigma}$ code and let $\strategy = (\cL^2(\alicealg, \tau), \ket{\tau},\{ A_{a}^u \}_{ \{u \in [n]^m, a \in \sigma\}} \cup \{ B^{l}_{g} \}_{ \{l \subseteq [n]^m, g \in \cC\}} \cup \{ P^{(u,v)}_{(a,b)} \}_{\{l \subseteq [n]^m, g \in \cC\}})$ be a synchronous strategy which succeeds in `` $\cC^{\tensor m}-$tensor code test" with probability $1-\eps$. Then there exists a projective measurement $\{ G_c \}_{c \in \cC^{\tensor m}}$ such that 
	\begin{equation} \label{eq:soundnesstensorcode}
		\Ex_{u \sim \Sigma^{m}} \sum_{c \in \cC^{\tensor m}} \braket{\tau| A_{c(u)}^u G_c |\tau} \leq 1 - \poly(m,t,d) \cdot \poly(\eps, n^{-1}). 
	\end{equation}
\end{theorem}
In order to apply \Cref{cor:corsoundness} to the setting of \Cref{thm:soundnesstensorcode}, let $\cG$ be the ``$\cC^{\tensor m}-$tensor code test", as presented in \cite[Figure 1]{jiQuantumSoundnessTesting2022}. Similar to \cite[section 5.1]{jiMIPRE2022}, a ``synchronicity test" with $c = \frac{1}{2}$ can be added to $\cG$ to ensure any $1-\eps$ optimal strategy for $\cG$ is at least $2\eps$ synchronous. Let $\cY = \{ y \}$ be the singleton set, $\cB$ be the set of all $([n,k,d]_{\Sigma})^{\tensor m}$ tensor codes and let $\rho$ be the uniform distribution over the point space $\Sigma^{m}$. For each $u \in \Sigma^{m}$ and $a \in \cA$ where $a$ is the ``point question'' (i.e.~$a \in \Sigma$), let $g_{uy}(a)$ be the collection of all possible codes $f \in ([n,k,d]_{\Sigma})^{\tensor m}$ such that $f(u) = a$ and otherwise an empty set for $a$ being either the ``line question'' or the ``pair question''. Since the point test appears in $\cG$ with probability $\frac{1}{3}$, the requirement for \eqref{eq:corsoundnessre1} trivially follows. Then \eqref{eq:soundnesstensorcode} corresponds to \eqref{eq:corsoundnessre2} in \Cref{cor:corsoundness} with $\kappa(\eps) = \poly(m,t,d) \cdot \poly(\eps, n^{-1})$. Using the same corollary, we derive the following corollary. 
\begin{corollary}[Quantum soundness for the tensor code test]  \label{cor:soundnesstensorcode}
	Let $\cC $ be an interpolatable $[n,k,d]_{\Sigma}$ code and let $\strategy = (\cL^2(\alicealg, \tau), \sigma \ket{\tau_{\alicealg}}, \{ A_{a}^x \}, \{B_a^y \}_{\{l \subseteq [n]^m, g \in \cC\}})$ be an arbitrary tracially embeddable strategy which succeeds in the `` $\cC^{\tensor m}-$tensor code test with synchronicity test" with probability $1-\eps$. Then there exists a projective measurement $\{ G_c \}_{c \in \cC^{\tensor m}}$ such that 
	\begin{equation*} 
		\Ex_{u \sim \Sigma^{m}} \sum_{c \in \cC^{\tensor m}} \braket{\tau| \sigma A_{c(u)}^u (G_c)^{op} \sigma|\tau} \leq 1 - \poly(m,t,d) \cdot \poly(\eps, n^{-1}) - \poly(\eps),
	\end{equation*}
	where $ A_{c(u)}^u$ denotes the point measurement on Alice's side. 
\end{corollary}
%We remark that using known classical PCP construction, the above corollary immediately implies that $\class{NEXP} \subseteq \MIPco$, which implies that $\MIPco$ is strictly more powerful than classical $\MIP$. Since this is not the focus of this paper, we instead refer to \cite[Chapter 10]{jiMIPRE2022} for more details. 

\subsubsection{Applications toward algebraic relationship}

The second application concerns proving algebraic relations in a ``typical" rigidity result. On a higher level, a lot of ``rigidity" proofs first reduce to showing that the measurement operator forms an approximate representation under a certain norm. This corollary shows that for synchronous games, these assumptions can be made for the general case, given that the same assumptions hold for the synchronous case. We remark that \eqref{eq:coralgrelation2} is equivalent to an approximation in the state-dependent norm. We will show a typical use case for \Cref{cor:coralgrelation} in the proof for \Cref{prop:anticommPauliBasis} during the analysis for the $\mu$-dependent Pauli Basis test. Note that the below corollary is an analogue of \cite[Corollary 4.4]{vidickAlmostSynchronousQuantum2022} for the commuting operator model. 

\begin{corollary} \label{cor:coralgrelation}
	Let $\cG = (\cX^2, \cA^2, \mu, D)$ be a synchronous game with $\{X, Z\}\subseteq \cX$ and $ \Z_q \subseteq \cA$ for some $q = p^n$ for some prime power p. Furthermore, assume $\mu(X, X) = O(1)$ and $\mu(Z, Z) = O(1)$. Let $\chi: \mathbb{Z}_q^2 \rightarrow \mathbb{C}$, $\chi(a,b) = \omega_p^{ab}$, where $\omega_p = e^{\frac{2 \pi i}{p}}$. Suppose for every synchronous symmetric projective strategy $\strategy^{\text{sync}} = (\cL^2(\bobalg, \tau_{\bobalg}), \ket{\tau_{\bobalg}},\{ A_{a}^x \})$ follows the ''algebraic relations" property:
	\begin{itemize}
		\item[] For $X(b) = \sum_{j \in \Z_q} \omega_p^{bj} A_{j}^{X}$, and $Z(b) = \sum_{j \in \Z_q} \omega_p^{bj} A_{j}^{Z}$, where $\omega_p = e^{\frac{2 \pi i}{p}}$, there exist a convex, monotone, non-decreasing function $\kappa:[0,1] \rightarrow\bR_{\geq 0}$ such that
		\begin{equation}\label{eq:coralgrelation}
			\Ex_{(a,b) \in \Z_q \times \Z_q} || X(a)Z(b) - \omega_p^{ab} Z(b)X(a) 
			\ket{\tau_{\bobalg}}||^2 \leq \kappa(\omega^{co}(G, \strategy^{\text{sync}} )).  
		\end{equation}
	\end{itemize}
	Then the ``algebraic relations" property extends to any arbitrary $\delta$-synchronous symmetric, projective strategy $\strategy = (\cL^2(\alicealg, \tau),\sigma \ket{\tau}, \{A_a^x \}, \{ B_{a}^x \})$ on both measurement operator $\{A_a^x \}$ and $\{ B_{a}^x \}$, with \eqref{eq:coralgrelation} replaced with
	\begin{equation}\label{eq:coralgrelation2}
		\Ex_{(a,b) \in \Z_q \times \Z_q} || (X(a)Z(b) - \omega_p^{ab} Z(b)X(a)) \sigma \ket{\tau}||^2 \leq  O(\kappa(\omega^{co}(G, \strategy) + \poly(\delta) ) + O(\delta^{1/8})).
	\end{equation}
\end{corollary}

We give a proof for \Cref{cor:coralgrelation} in \Cref{sec:roundingapp} and an application in \Cref{sec:DistPBT}.

\section{State-dependent Gower-Hatami theorem}

In this section, we generalize the variant of the Gowers-Hatami theorem~\cite{gowersInverseStabilityTheorems2017} by Vidick~\cite[Theorem 4]{vidickExpositoryNoteQuantum2018} to the tracial von Neumann algebra setting. Recall from the preliminary that the state-dependent norm is defined as $|| A ||_{\rho} = \sqrt{ \tau( A^* A \rho) }$. From this definition, we have the following theorem.

\begin{theorem}[Gowers-Hatami for tracial von Neumann algebras] \label{thm:SGHinTracialVNA}
	Let $G$ be a finite group, $\alicealg \subseteq \BofH$ a tracial von Neumann algebra in standard form, $\eps \geq 0$ and $\phi: G \rightarrow \cU(\alicealg)$ be an $(\eps, \rho)$-representation for some $\rho \in \alicealg^{+}$ such that $\tau(\rho)=1$. Then there exists some isometry $V: \cH \rightarrow  \cH \tensor \bC^{| G |}  $ such that for all $B \in \alicealg^{'}$,  $VB = (B \tensor I_{|\cG|})V$ and a unitary representation $\phi^{'}: G \rightarrow  \alicealg \tensor \cB( \bC^{| G |}) $ such that
	\begin{equation*}
		\Ex_{g \in G} \| \phi(g) - V^* \phi^{'}(g) V ||_{\rho}^2 \leq \eps.
	\end{equation*}
\end{theorem}

Note that since each irreducible representation $\phi$ for a finite group appears $dim(\phi)$ times in the left regular representation, the proof below follows the same structure as \cite[Theorem 4]{vidickExpositoryNoteQuantum2018}.

\begin{proof}
	Define the isometry $V: \cH \rightarrow  \cH \tensor \bC^{| G |} $ as $V \ket{\psi} = \frac{1}{\sqrt{|G |}}\bigoplus_{g \in G} \left(\phi (g^{-1}) \ket{\psi}\right)$, and note $V^* \left( \frac{1}{\sqrt{|G|}} \bigoplus_{g \in G} \ket{\psi_g} \right)  = \frac{1}{|G|}\sum_{g \in G} \phi(g^{-1})^* \ket{\psi_g}$. We see that for all $\ket{\psi} \in \cH$
	\begin{equation*}
		V^* V \ket{\psi} = V^*   \left(\frac{1}{\sqrt{|G|}} \bigoplus_{g \in G} \phi (g^{-1}) \ket{\psi}\right) = \frac{1}{|G|}\sum_g  \phi(g^{-1})^* \phi(g^{-1}) \ket{\psi} = \ket{\psi}, 
	\end{equation*}
	showing that $V^* V = \cI_{\cH}$ and hence V is an isometry. Furthermore for all $B \in \alicealg^{'}$
	\begin{equation}
		VB \ket{\psi} =  \frac{1}{\sqrt{|G |}} \bigoplus_{g \in G} \left(\phi (g^{-1}) B \ket{\psi}\right) = (B \tensor I_{|\cG|}) V \ket{\psi}, 
	\end{equation}
	implying $VB =  (B \tensor I_{|\cG|}) V$. Define $\phi^{'} (g) =  \cI_{\cH} \tensor  L_g $, where $L_g$ is the left regular representation for the group G. Note for all $a \in G$
	\begin{align*}
		V^* \phi^{'}(a) V &= \frac{1}{\sqrt{|G |}} V^* (  \cI_{\cH}  \tensor L_a)\left(\bigoplus_{g \in G} \phi(g^{-1})\right) 
		&= \frac{1}{\sqrt{|G |}} V^* \bigoplus_{ag \in G}  \phi(g^{-1}) 
		&= \frac{1}{|G|} \sum_{g} \phi(ag) \phi^*(g).
	\end{align*}
	Hence since $\phi$ is an $(\eps, \rho)-$representation
	\begin{align*}
		\Ex_{g \in G} \| \phi(g) - V^* \phi^{'}(g) V ||_{\rho}^2 &= \Ex_{g \in G} \| \phi(g) - \sum_{h} \phi(gh) \phi^*(h) ||_{\rho}^2 \leq\Ex_{g,h\in G} \| \phi(g) \phi(h) - \phi(gh) ||_{\rho}^2  \leq \eps
	\end{align*}
	completing the claim.
\end{proof}

We remark that the condition where $V B = (I_{|\cG|} \tensor B)V$ for all $B \in \alicealg^{'}$ intuitively corresponds to $V$ only acting on Alice's algebra in the context of non-local games. Using \Cref{exam:StandardformFD}, the above theorem can be seen as \cite[Theorem 4]{vidickExpositoryNoteQuantum2018} applied on the first register of $\bC^n \tensor \bC^n$. For most rigidity proofs, if the tracially embeddable strategy in question is defined over the vector state $\ket{\psi} = \sigma \ket{\tau}$, then $\rho = \sigma^2$ (recall from \Cref{sec:Tracialemd} that this corresponds to the reduced density matrix). 

%(in this case, the above condition can be rewritten as $V (\cI_n \tensor B) = (\cI_n \tensor B)V $ which trivially implies that $V$ only acts on the first register)
 \section{Distribution-dependent Pauli Basis test} \label{sec:DistPBT}
In this section, we introduce the $\mu$-dependent Pauli Basis test for some probability distribution $\mu$ over $\{0,1\}^n$ (by \Cref{sec:PrestringandPauli}, $\mu$ can also be viewed as a distribution over $\bF_{2}^n$, and we will only treat it as such in the context of spectral gap for the distribution). The goal is to force two honest players to prepare $ n$ copies of EPR pairs and force them to measure either $\rho_X^{\tensor n}$ or $\rho_Z^{\tensor n}$ on their half of the EPR pair. This test is inspired by~\cite[Section 7.3]{jiMIPRE2022} and was implicitly studied in \cite{dechiffreOperatorAlgebraicApproach2019} for synchronous correlations. As a result, many ideas from this section originate from the two works above. On a very high level, the honesty of the player is done by periodically forcing one of the players to make the same $\rho_X$ or $\rho_Z$ measurements on some subset of the $n$ copies and perform with another player who is asked to measure \textit{all} of the qubits. The frequency of specific subsets used for cross-checking is specified by the probability distribution $\mu$. Interestingly, the robustness of this test depends on the \textit{spectral gap} of the probability measure $\mu$, a quantity we define later in this section. 

The remainder of the section is organized as follows: We first introduce the spectral gap for a probability distribution over a finite group $G$, as well as some of the key lemma from \cite{delasalleSpectralGapStability2022}, which we use in this section. We also introduce the well-known \textit{Magic Square game} in \Cref{sec:MSgame}, as it will serve as a key subroutine to the  $\mu$-dependent Pauli Basis test. Finally, we introduce the test in \Cref{sec:mudepPB}, along with some rigidity statements. Finally, in \Cref{sec:PBremark}, we make some final remarks about the complexity of $\MIPco$ and compare this test to other games with similar structures in the literature. 

%Move to Preliminary, maybe?
\subsection{Spectral gap} \label{sec:spectralgap}
In this section, we recall the definition of the spectral gap for a probability distribution over a group and introduce a key lemma used to analyze the $\mu$-dependent Pauli Basis test. Many of these definitions and theorems are from \cite{delasalleSpectralGapStability2022} and refer to the original reference for a more comprehensive introduction to this topic. Let $G$ be a finite group and $\mu$ be a probability measure on $G$; we use $\supp(\mu)$ to denote the support of the distribution, or $\supp(\mu) = \{g \in G | \mu(g) > 0 \}$. We say $\mu$ is symmetric if $\mu(g^{-1}) = \mu(g)$ for all $g \in G$, and generating if $\supp(\mu)$ generates the group G. We define $\kappa(\mu)$ as the smallest number such that for all unitary representation $\pi$, we have 
\begin{equation*}
	sp(\sum_g \mu(g) \pi(g) ) \subseteq  [-1, 1 - \frac{1}{\kappa(\mu)}] \cup \{ 1 \}, 
\end{equation*}
and we note that $\kappa(\mu) = \infty$ in certain cases. We note that $\kappa(\mu)$ is the inverse of the \textit{spectral gap} within the literature and could potentially be infinite given the probability distribution. For the convenience of notation, we let $\kappa(\mu) = \infty$ if the probability distribution is not generating nor symmetric, and we justify this in \Cref{sec:mudepPB}. We recall the following theorem about $\kappa(\mu)$ and note that the below theorem is based on the well-known Poincar\'e inequality. 

\begin{theorem}[\cite{delasalleSpectralGapStability2022}, Corollary 2.5] \label{thm:Ponicarethm}
	Let $q = p^n$ for some prime power p, and let $\omega_p = e^{\frac{2 \pi i}{p}}$ be the root of unity of order p. Let $\mu$ and $\upsilon$ be two symmetric and generating probability measure on $\bF_{q}$. Then, for any unitary representation $\pi_1: \bF_q \rightarrow \alicealg$ and $\pi_2: \bF_q \rightarrow \alicealg$, we have
	\begin{equation*}
		\Ex_{(g,h)\in \bF_{q}^2 } || \pi_1(g) \pi_2(h) - \omega_{p}^{gh} \pi_2 (h)\pi_1(g)  ||_2^2 \leq \kappa(\mu) \kappa(\upsilon) \sum_{(g,h)\in \bF_{q}^2} \mu(g) \upsilon(h)  || \pi_1(g) \pi_2 (h) - \pi_2 (h)\pi_1(g)  ||_2^2
	\end{equation*}
\end{theorem}
Note that the above theorem only applies when $\pi_1$ and $\pi_2$ are both \textit{exact} unitary representations for the group $\bF_{q}$. We also recall the following fact about the spectral gap of $\mathbb{Z}_q$.
\begin{theorem}[\cite{alonRandomCayleyGraphs1994}, Proposition 4.6] \label{thm:Specgapconstant}
	For prime $p$, there exists a subset $S \subseteq \mathbb{Z}_{p^n}$ with $|S| = O(n)$ such that the uniform distribution $\mu$ on S has the property that $\kappa(\mu) \leq 2$.
\end{theorem}

\begin{comment}
	
	\jqnote{Should throw a proof in the appendix, maybe just directly cite this inequality, NEED TO DO THE PROOF FOR CHARACTER FUNCTION}
	\begin{theorem}[\cite{delasalleSpectralGapStability2022}, Theorem 2.1] \label{thm:Ponicarethm}
		Let $G$ and $H$ be two finite groups, and let $\mu$ and $\upsilon$ be two symmetric and generating probability measures on $G$ and $H$, respectively. Let $$

		Then, for any unitary representation $\pi^{G}: G \rightarrow \alicealg$ and $\pi^{H}: H \rightarrow \alicealg$, we have
		\begin{equation*}
			\Ex_{g \in G, h \in H} || \pi^G(g) \pi^H (h) - \pi^H (h)\pi^G(g)  ||_2^2 \leq \kappa(\mu) \kappa(\upsilon) \sum_{g \in G, h \in H} \mu(g) \upsilon(h)  || \pi^G(g) \pi^H (h) - \pi^H (h)\pi^G(g)  ||_2^2
		\end{equation*}
	\end{theorem}
\end{comment}

\subsection{The magic square game} \label{sec:MSgame}
\begin{figure}[!t]
	\begin{center}
		\begin{tabular}{|c|c|c|}
			\hline
			$x_1$ & $x_2$ & $x_3$ \\
			\hline 
			$x_4$ & $x_5$ & $x_6$ \\	
			\hline
			$x_7$ & $x_8$ & $x_9$ \\
			\hline
		\end{tabular}
		
	\end{center}
	\caption{The magic square game, where each row and column corresponds to an equation.}
	\label{fig:magicsquare}
\end{figure}
We recall the well-known magic square game of Mermin and Peres~\cite{merminSimpleUnifiedForm1990, peresIncompatibleResultsQuantum1990, p.k.aravindSimpleDemonstrationBells2002} in this section. This game has served as a key subroutine to many Bell tests as it forces the verifier to measure in anti-commutative observables. We use the binary constraint formulation of the magic square game~\cite{cleveCharacterizationBinaryConstraint2014} in this section, whereby the game is defined using six equations and nine variables over $\bF_2$, which is displayed over a three by three grid (see \Cref{fig:magicsquare}). Each row and column corresponds to an equation that sums up to zero, except the last column, where the equation sums up to one. In this game, the referee will randomly pick one of the six constraints and send it to Alice, and randomly pick one of the three variables that appear on Alice's constraint and send it to Bob. Alice must respond with an assignment for the three variables consistent with the constraint, and Bob must respond with an assignment to the variable, which must be consistent with Alice's assignment. In this paper, we also assume the magic square game is synchronous (similar to the discussion on the tensor code test at the end of \Cref{sec:soundnesstensorcode}), meaning with probability $\frac{1}{2}$, the referee will instead send the same constraint or variable label to Alice and Bob and expect the same answers from both players. The magic square game has the ``anti-commutation property" under the tensor product model (see, for example, \cite[Theorem 6.9]{coladangeloRobustSelftestingLinear2019a}), where for any near-optimal strategy, the binary observables used for variables $1$ and $5$ must also be approximately anti-commuting. 

Combining \cite{mousaviNonlocalGamesCompression2022}, Theorem 3.1 part iii) (rigidity for the magic square game under synchronous correlations) and \Cref{cor:coralgrelation}, we recover the same anti-commutation result for the commuting operator model.  

\begin{theorem}[Approximate anti-commutation test for the magic square game in the commuting operator model] \label{thm:RigMS}
	Let $\strategy = (\cL^2(\alicealg, \tau),\sigma \ket{\tau}, \{A_a^x \}, \{ B_{a}^x \})$ be a tracially embeddable strategy for the magic square game such that $\omega^{co}(\text{Magic square}, \strategy) \geq 1- \eps$. 
	
	For $i \in \{1 \cdots 9\}$, let $\{A_{0}^{\text{(variable, i)}}, A_{1}^{\text{(variable, i)}}\} \subseteq \alicealg$ be Alice's measurement operator for variable i, and likewise with $\{B_{0}^{\text{variable, i}}, B_{1}^{\text{variable, i}}\} \subseteq \alicealg^{'}$ with Bob's measurement operator. For $P \in \{ A, B\}$
	\begin{align*}
		P^{\text{(variable, i)}} &=  P_{0}^{\text{(variable, i)}} - P_{1}^{\text{(variable, i)}}.
	\end{align*}
	Then
	\begin{align*}
		||(P^{\text{(variable, 1)}}P^{\text{(variable, 5)}} + P^{\text{(variable, 5)}} P^{\text{(variable, 1)}})||_{\sigma^2}^2 &\leq O(\poly(\delta)). 
	\end{align*}
\end{theorem}
Note that it is possible to obtain a better bound using a similar analysis as \cite{coladangeloRobustSelftestingLinear2019a}, theorem 6.9 and replacing the state-dependent Gowers Hatami with \Cref{thm:SGHinTracialVNA}, and we leave this as an exercise to the readers. 

\subsection{The test} \label{sec:mudepPB}

\begin{figure}[!b]
	\centering
	\begin{tikzpicture}[scale=.8]
		
		\tikzset{type/.style args={[#1]#2}{
				draw,circle,fill,scale=0.25,
				label={[font=\scriptsize, label distance=1pt]#1:#2}
		}}
		
		\foreach \i in {1,...,6} \draw (0,8-9/8*\i)
		coordinate (Constraint-\i)
		node[type={[180]$\text{Constraint}_\i$}] {};
		
		\foreach \i in {2,...,4} \draw (2.5,9-\i)
		coordinate (Variable-\i)
		node[type={[330]$\text{Variable}_\i$}] {};
		
		\foreach \i in {6,...,9} \draw (2.5,9-\i)
		coordinate (Variable-\i)
		node[type={[330]$\text{Variable}_\i$}] {};

		\draw (3.5,8) coordinate (Variable-1) node[type={[45]$(\text{Coordinate},X)$}] {};
		\draw (3.5,4) coordinate (Variable-5) node[type={[315]$(\text{Coordinate},Z)$}] {};
		
		\foreach \i in {1,...,3} \foreach \j in {1,...,3}
		\pgfmathsetmacro{\k}{(\i-1)*3+\j}
		\draw[blue] (Constraint-\i) -- (Variable-\k);
		
		\foreach \i in {4,...,6} \foreach \j in {1,...,3}
		\pgfmathsetmacro{\k}{\i-3+(\j-1)*3}
		\draw[blue] (Constraint-\i) -- (Variable-\k);

		\draw (6,8) coordinate (Pauli-X) node[type={[0]$(\text{Pauli}, X)$}] {};
		\draw (6,4) coordinate (Pauli-Z) node[type={[0]$(\text{Pauli},Z)$}] {};
		\draw (6,6) coordinate (Commutation) node[type={[0]$\text{Commutation}$}] {};
		
		\foreach \from \to in {Variable-1/Pauli-X, Variable-5/Pauli-Z,Variable-1/Commutation, Variable-5/Commutation}
		\draw[red] (\from) -- (\to);
		\foreach \from \to in {Variable-1/Pauli-X, Variable-5/Pauli-Z}
		\draw (\from) -- (\to);
	\end{tikzpicture}
	
	\caption{Graph $G$ for the probability Pauli basis test, where each vertex corresponds to the type of question the players could potentially receive, and each edge corresponds to the potential question pair the players can receive. The colour of the vertex is used in the sampling procedure, which will be specified in more detail in \Cref{fig:PauliBasis}. Each vertex contains a self-loop (which counts as black edges) due to the synchronous condition, but not drawn for clarity.}
	\label{fig:PauliBasisgraph}
\end{figure}

\begin{figure}[!htbp]
	\centering
	\begin{gamespec}
		\setlength{\tabcolsep}{1em}
		\begin{tabularx}{\textwidth}{ l   l   X   }
			\toprule
			Question label & Question content & Answer format \\
			\midrule
			(Pauli,W) &  & $ t_{W} \in \{0,1\}^n$ \\
			(Coordinate, X)& $u \in \supp(\mu)$& $ t_{\text{(Pauli, X)}} \in \{0,1\}^{|u|}$ \\
			(Coordinate, Z)& $v \in \supp(\mu)$& $t_{\text{(Pauli, Z)}}  \in \{0,1\}^{|v|}$ \\
			Commutation& $(u,v) \in \supp(\mu)^2$& $(t_{X}, t_{Z})\{0,1\}^{u+v}$ \\
			$\text{Variable}_i$& $(u,v) \in \supp(\mu)^2$& $t_{\text{var}} \in \{0,1\}$ \\
			$\text{Constraint}_i$& $(u,v) \in \supp(\mu)^2$& $t_{\text{cons}} \in \{0,1\}^{3}$ \\
			\bottomrule
			\multicolumn{3}{c}{Figure: Question and answer format for the $\mu \sim \bF_2^n$ dependent Pauli Basis test}
		\end{tabularx}
		\begin{center}
			\textbf{Sampling procedure}
		\end{center}
		\begin{enumerate}
			\item Sample $(u,v) \in \{0,1\}^{2n}$ according to the distribution $\mu \times \mu$. 
			\item If $u \cdot v = 0 \mod 2$, sample, uniformly, a red or black edge $e$ from the graph in \Cref{fig:PauliBasisgraph}. Otherwise, uniformly sample a blue or black edge $e$ from the same graph.
			\item Uniformly pick one of the vertex of $e$, send the question label and question content corresponding to that vertex to Alice, and send the question corresponding to the other vertex to Bob.
		\end{enumerate}
		\begin{center}
			\textbf{Decision procedure}
		\end{center}
		\begin{itemize}
			\item (Self-loop): If Alice and Bob are given the same question, they win if and only if they output the same answer. 
			\item (Pauli, W) \textbf{---}  (Coordinate, W): Alice and Bob win iff $t_{W}|_u = t_{\text{(Pauli, W)}}$. 
			\item (Coordinate, W) \color{red} \textbf{---}  \color{black}  Commutation: Alice and Bob wins iff $t_{W} = t_{\text{(Pauli, W)}}$.
			\item (Constraint)  \color{blue} \textbf{---}   \color{black}  (Coordinate, X): Alice and Bob win iff $\sum_{i : u_i =1} (t_{\text{(Pauli, X)}})_i$ is consistent with the variable assignment for variable 1 in the magic square game.
			\item (Constraint) \color{blue} \textbf{---}  \color{black}  (Coordinate, Z): Alice and Bob win iff $\sum_{i : v_i =1} (t_{\text{(Pauli, Z)}})_i$ is consistent with the variable assignment for variable 5 in the magic square game.
			\item (Constraint)  \color{blue} \textbf{---}   \color{black}  (Variable): Alice and Bob win if and only if the answer is consistent with the magic square game. 
		\end{itemize}
		\vspace{1em}
	\end{gamespec}
	\caption{The description for the $\mu \sim \bF_2^n$ dependent Pauli Basis test. We assume that $W \in \{X, Z\}$ in this figure. }
	\label{fig:PauliBasis}
\end{figure}

In this section, we give the formal definition for the $\mu$-dependent Pauli Basis test. We describe the sampling and decision procedure in \Cref{fig:PauliBasis} and include a diagram for the sampling procedure in \Cref{fig:PauliBasisgraph} for clarity. 

First, we will give an informal overview of the test. As mentioned at the beginning of this section, the goal of this test is for two players to certify that they share a state of the form $\ket{EPR}^{\tensor n}$. For $W \in \{X, Z\}$, once the players receive the question label $(\text{Pauli}, W)$, they are expected to perform the measurement $\{ \rho_{W}^{\tensor n}\}$ on their half of their respective states. The referee achieves this by first randomly sampling some $u \sim \mu$, sending the label $(\text{Coordinates, W})$ along with the coordinates $u$, and expecting the player receiving this label to measure according to $\rho_W(u)$. In other words, expect the player to perform a $W$ measurement on the coordinate $i \in [n]$ if $u_i = 1$. The other player will receive $(\text{Pauli}, W)$. Since the entangled state expected from the players are EPR pairs, the player would return the same measurement output on the sub-coordinate where they both made the measurement.

The commutation/anti-commutation test is used to enforce consistency on the question label $(\text{Coordinates, W})$. To be more precise, the referee first sampling two coordinates $(u,v)$ and determines whether the observable $\rho_X(u)$ and $\rho_Z(v)$ commutes or anti-commutes based on the value of $u \cdot v \text{ mod } 2$. In the scenario where they do commute $u \cdot v \text{ mod } 2 = 0$, the ``commutation test" or the red edges of \Cref{fig:PauliBasis} is performed. In this case, the referee asked one of the players to simultaneously measure both the $\rho_X(u)$ and $\rho_Z(v)$ (which would be consistent since they commute) and compare the measurement output to the other player performing one of the two measurements. In the case where $u \cdot v \text{ mod } 2 = 1$, the ``anti-commutation test", or the blue edges of \Cref{fig:PauliBasis} is performed, where the referee forces both players to play the magic square game from \Cref{sec:MSgame} with the observable $\rho_X(u)$ and $\rho_Z(v)$ for question 1 and 5 of the magic square game respectively. In this case, \Cref{thm:RigMS} guarantees that any perfect or near-perfect strategy ensures the two aforementioned observable to anti-commute. We remark that unlike the Pauli Basis test \cite[Section 7.3]{jiMIPRE2022}, the commutation and the anti-commutation can be applied directly onto the question label $(\text{Coordinates, W})$ since the answers given are no longer compressed through the low-degree test.

In order to state the main rigidity theorem, let $\strategy = (\cL^2(\alicealg, \tau),\sigma \ket{\tau}, \{A_a^x \}, \{ B_{a}^x \})$ be a projective, tracially embeddable strategy for the $\mu$-dependent Pauli Basis test. For $W \in \{X,Z\}$ and $P \in \{A,B\}$ let $\{P_a^{(\text{Pauli}, W)}\}$ be the measurement operator for the question label $(\text{Pauli}, W)$ for player $P$, define $W^{P}(u) $ to be the unitary observable to be the following set of observables:
\begin{equation} \label{eq:obserpaulibasis}
	W^{P}(u) = \sum_{q \in \bF_{2}^n} (-1)^{q \cdot u}P_q^{(\text{Pauli}, W)}, 
\end{equation}
where, recall $W^{A}(u) \subseteq \alicealg$ and $W^{B} \subseteq \alicealg^{'}$. We state the following rigidity theorem about the $\mu$-dependent Pauli Basis test. 

\begin{theorem}[Robustness of the $\mu$-dependent Pauli basis test under the commuting operator model] \label{thm:RigPaulibasis}
	Let $\cG(\mu)$ be the $\mu$-dependent Pauli Basis test and let $\strategy = (\cL^2(\alicealg, \tau),\sigma \ket{\tau}, \{A_a^x \}, \{ B_{a}^x \})$ be a tracially embeddable strategy for $\cG(\mu)$ with $\omega^{co}(\cG(\mu), \strategy) \geq 1 - \eps$. There exist two isometries $V_{A}: \cH \rightarrow  \cH \tensor \bC^{2^{2n}}$ and $V_{B}: \cH \rightarrow \cH \tensor \bC^{2^{2n}}$ with $(V_B \tensor \cI_{2^{2n}}) V_{A} = ( V_{A} \tensor \cI_{2^{2n}} )V_B$, and a state $\ket{\text{Aux}} \in \cH  \tensor \cI_{2^{2n}}$ such that
	\begin{equation} \label{eq:robustnessthmstat1}
		\left\| \left(  V_B \tensor \cI_{2^{2n}}  \right) V_{A} \ket{\psi} -\ket{\text{Aux}}  \ket{\text{EPR}}^{\tensor n} \right\|^2 \leq O\left(\poly(\kappa(\mu), \eps )\right),
	\end{equation}
	and for all $W \in \{X, Z\}$ and $u \in \{0,1\}^n$
	\begin{align} \label{eq:robustnessthmstat2}
		\left | \left( V_{A} W^{A}(u) V_{A}^{*} \tensor \cI_{2^{2n}} - \cI_{\cH} \tensor \cI_{2^{2n}} \otimes \rho^{W}(u) \tensor \cI_{2^{n}} \right) \ket{\text{Aux}}  \ket{\text{EPR}}^{\tensor n} \right |^2 \leq O\left(\poly(\kappa(\mu), \eps )\right), \\
		\left | \left( V_{B} W^{B}(u) V_{B}^{*} \tensor \cI_{2^{2n}} - \cI_{\cH} \tensor \cI_{2^{2n}} \tensor \cI_{2^{n}} \tensor \rho^{W}(u) \right) \ket{\text{Aux}} \ket{\text{EPR}}^{\tensor n} \right |^2  \leq O\left(\poly(\kappa(\mu), \eps )\right).
	\end{align}
\end{theorem}

We provide a proof for \Cref{thm:RigPaulibasis} in \Cref{sec:anticommPauliBasis}. Intuitively, if $\mu$ is a uniform distribution on some subset $S \subseteq \{0,1\}^n$. The larger $S$ is, the harder it is for dishonest players to cheat, but this also gives a larger question set (or sampling complexity) for the $\mu$-dependent Pauli Basis test. This naturally gives a trade-off between the sampling complexity and the accuracy of the test. Hence, $\kappa(\mu)$ can intuitively be interpreted as ``how good the distribution is" for the theorem above. We also remark that, when considering $\mu$ as a distribution over $\bF_{2}^n$, the distribution is automatically symmetric since $\bF_{2}^n$ is an abelian group. If $\mu$ is not generating, then this implies there exists some generator $i \in [n]$ such that $s_i = 0$ for all $s \in \supp(\mu)$. Translating this in the context of the $\mu$-dependent Pauli basis test would imply that there exists an EPR pair which would not get cross-checked in any of the subsets sampled under $\mu$; hence, the robustness theorem should not hold under this instance. This justifies the assumption we made for $\kappa(\mu) = \infty$ with distribution, which is not symmetric nor symmetric in \Cref{sec:spectralgap}. We further remark that the condition $(V_B \tensor \cI_{2^{2n}}) V_{A} = ( V_{A} \tensor \cI_{2^{2n}} )V_B$ is akin to $V_A$ and $V_B$ acting on different Hilbert space under the tensor product model, and the extra $\bC^{2^n}$ on the axillary space results from using the left regular representation of the Weyl-Heisenberg group. 

Combining \Cref{thm:RigPaulibasis} with \Cref{thm:Specgapconstant}, we have the following corollary. Note that the game below is the same as the test in \cite[Section 3.4]{delasalleSpectralGapStability2022}; the corollary below is a generalization from synchronous strategies to general tracially embeddable strategies.
\begin{corollary} \label{cor:PauliBasis}
	There exist a protocol $\cG$ with $|\cX| = O(\poly(n))$ and $|\cA| = O(exp(n))$, with $\{ X, Z \} \subseteq \cX$, $\{0,1\}^n \subseteq \cA$.  Such that with all tracially embeddable strategy $\strategy = (\cL^2(\alicealg, \tau),\sigma \ket{\tau}, \{A_a^x \}, \\ \{ B_{a}^x \})$ with $\omega^{co}(\cG(\mu), \strategy) \geq 1 - \eps$,  there exist two isometries $V_{A}: \cH \rightarrow  \cH \tensor \bC^{2^{2n}}$ and $V_{B}: \cH \rightarrow \cH \tensor \bC^{2^{2n}}$ with $(V_B \tensor \cI_{2^{2n}}) V_{A} = ( V_{A} \tensor \cI_{2^{2n}} )V_B$, and a state $\ket{\text{Aux}} \in \cH$ such that
	\begin{equation} 
		\left\| \left(  V_B \tensor \cI_{2^{2n}}  \right) V_{A} \ket{\psi} -\ket{\text{Aux}}  \ket{\text{EPR}}^{\tensor n} \right\|^2 \leq O\left(\poly(\eps)\right),
	\end{equation}
	and for all $W \in \{X, Z\}$ and $u \in \{0,1\}^n$
	\begin{align} 
	\label{eq:Robustalice}\left | \left( V_{A} W^{A}(u) V_{A}^{*} \tensor \cI_{2^{2n}} - \cI_{\cH} \tensor \cI_{2^{2n}} \otimes \rho^{W}(u) \tensor \cI_{2^{n}} \right) \ket{\text{Aux}}  \ket{\text{EPR}}^{\tensor n} \right |^2 \leq O\left(\poly(\eps )\right), \\
	\label{eq:Robustbob} \left | \left( V_{B} W^{B}(u) V_{B}^{*} \tensor \cI_{2^{2n}} - \cI_{\cH} \tensor \cI_{2^{2n}} \tensor \cI_{2^{n}} \tensor \rho^{W}(u) \right) \ket{\text{Aux}} \ket{\text{EPR}}^{\tensor n} \right |^2 \leq O\left(\poly(\eps )\right).
\end{align}
\end{corollary}

We remark although the above corollary is written for observables. The above corollary also holds for the player's measurement observables. To be more precise, ~\eqref{eq:Robustalice} and ~\eqref{eq:Robustbob} in~\Cref{cor:PauliBasis} can be replace with the following:
	\begin{align} 
	\label{eq:Robustalicemeasure} \sum_{a \in \{0,1\}^n}| \left( V_{A} \tilde{P}^{A}_a V_{A}^{*} \tensor \cI_{2^{2n}} - \cI_{\cH} \tensor \cI_{2^{2n}} \otimes \rho^{W}_a \tensor \cI_{2^{n}} \right) \ket{\text{Aux}} \ket{\text{EPR}}^{\tensor n} |^2 , \\
	\label{eq:Robustbobmeasure} \sum_{a \in \{0,1\}^n}| \left( V_{B} \tilde{P}^{B}_a V_{B}^{*} \tensor \cI_{2^{2n}} - \cI_{\cH} \tensor \cI_{2^{2n}} \otimes \rho^{W}_a \tensor \cI_{2^{n}} \right) \ket{\text{Aux}} \ket{\text{EPR}}^{\tensor n} |^2 .
\end{align}
We show that~\eqref{eq:Robustalicemeasure} is the same as~\eqref{eq:Robustalice} below; the proof of~\eqref{eq:Robustbobmeasure} being equivalent to~\eqref{eq:Robustbob} follows a similar proof. Let $\tilde{P}_a = V_A P^{(\text{Pauli}, W)}  V_A^*$, let $u = (1, \cdots ,1)$ be the all 1 string in $\bF_2^n$. We see that
{\allowdisplaybreaks
\begin{align*}
	&\left | \left( V_{A} W^{A}(u) V_{A}^{*} \tensor \cI_{2^{2n}} - \cI_{\cH} \tensor \cI_{2^{2n}} \otimes \rho^{W}(u) \tensor \cI_{2^{n}} \right) \ket{\text{Aux}}  \ket{\text{EPR}}^{\tensor n} \right |^2 \\
	&=  \bra{\text{Aux}}\bra{\text{EPR}}^{\tensor n} \left( \sum_{a \in \{0,1\}} (-1)^{\Tr(a)} (\tilde{P}_a \tensor \cI_{2^{2n}}) - \sum_{b \in \{0,1\}} (-1)^{\Tr(b)} \cI_{\cH} \tensor \cI_{2^{2n}} \otimes \rho^{W}_b \tensor \cI_{2^{n}} \right)^2 \ket{\text{Aux}}  \ket{\text{EPR}}^{\tensor n} \\
	&= 2 -  \sum_{a \in \{0,1\}}  \bra{\text{Aux}}\bra{\text{EPR}}^{\tensor n}  (\tilde{P}_a \tensor \cI_{2^{2n}}) (\cI_{\cH} \tensor \cI_{2^{2n}} \otimes \rho^{W}_a \tensor \cI_{2^{n}} ) \ket{\text{Aux}}  \ket{\text{EPR}}^{\tensor n} \\
	&\quad -  \sum_{a \in \{0,1\}} \sum_{b \in \{0,1\}} (-1)^{\Tr(a+b)} (\bra{\text{Aux}}\bra{\text{EPR}}^{\tensor n}  (\tilde{P}_a \tensor \cI_{2^{2n}}) (\cI_{\cH} \tensor \cI_{2^{2n}} \otimes \rho^{W}_b \tensor \cI_{2^{n}} ) \\
	&\qquad-  (\cI_{\cH} \tensor \cI_{2^{2n}} \otimes \rho^{W}_b \tensor \cI_{2^{n}} ) (\tilde{P}_a \tensor \cI_{2^{2n}}) \ket{\text{Aux}}  \ket{\text{EPR}}^{\tensor n}) \\
	&= 2 -  \sum_{a \in \{0,1\}}  \bra{\text{Aux}}\bra{\text{EPR}}^{\tensor n}  (\tilde{P}_a \tensor \cI_{2^{2n}}) (\cI_{\cH} \tensor \cI_{2^{2n}} \otimes \rho^{W}_a \tensor \cI_{2^{n}} ) \ket{\text{Aux}}  \ket{\text{EPR}}^{\tensor n}  \\
	&\quad -  \sum_{a \in \{0,1\}} \sum_{b \in \{0,1\}} (-1)^{\Tr(a+b)} (\bra{\text{Aux}}\bra{\text{EPR}}^{\tensor n}  (\tilde{P}_a \tensor \cI_{2^{2n}}) (\cI_{\cH} \tensor \cI_{2^{2n}} \otimes \rho^{W}_b \tensor \cI_{2^{n}} ) 
	&\qquad-  (\tilde{P}_a \tensor \cI_{2^{2n}}) (\cI_{\cH} \tensor \cI_{2^{2n}} \otimes \rho^{W}_b \tensor \cI_{2^{n}} ) \ket{\text{Aux}}  \ket{\text{EPR}}^{\tensor n}) \\
	&= 2 -  \sum_{a \in \{0,1\}}  \bra{\text{Aux}}\bra{\text{EPR}}^{\tensor n}  (\tilde{P}_a \tensor \cI_{2^{2n}}) (\cI_{\cH} \tensor \cI_{2^{2n}} \otimes \rho^{W}_a \tensor \cI_{2^{n}} ) \ket{\text{Aux}}  \ket{\text{EPR}}^{\tensor n} \\
	&= \sum_{a \in \{0,1\}^n}| \left( V_{A} \tilde{P}^{A}_a V_{A}^{*} \tensor \cI_{2^{2n}} - \cI_{\cH} \tensor \cI_{2^{2n}} \otimes \rho^{W}_a \tensor \cI_{2^{n}} \right) \ket{\text{Aux}} \ket{\text{EPR}}^{\tensor n}  |^2 , 
\end{align*}
}
where the second to last inequality follows from $\rho_{b}^{W} \otimes \cI \ket{EPR} = \cI \otimes\rho_{b}^{W} \ket{EPR}$. 

\subsection{Remark and other self-test}  \label{sec:PBremark}

We first remark that the technique used for showing \Cref{thm:RigPaulibasis} is inspired mainly by rigidity proofs for the tensor product model, with the exception of \Cref{thm:Ponicarethm}, which will also hold for finite-dimensional representations. Hence, by a similar analysis, \Cref{thm:RigPaulibasis} can be translated into a finite-dimensional proof. It is possible to design a test with two different distributions for the X coordinate and the Z coordinate (instead of having both depending on $\mu$), but choose to define it as \Cref{fig:PauliBasis} for clarity. Furthermore, we note that an explicit subset $S$ for $\bF_{2}^n$ in which the uniform distribution on $S$ can be used instead in order to show \Cref{cor:PauliBasis} (with a slightly worse bound), and we refer the reader to \cite[Section 1.4]{delasalleSpectralGapStability2022} for more details. 

This test is partially inspired by the Pauli Basis test \cite[Section 7.3]{jiMIPRE2022}. Informally, the Pauli Basis test asks the players to send back a compressed version of their measurement result through the low-individual degree test and check the commutation/anti-commutation property for \Cref{prop:anticommPauliBasis} through the compressed output. By using \Cref{cor:soundnesstensorcode}, the same analysis can be carried out to recover a similar result for the commuting operator model. As remarked by \cite[Remark 3.5]{delasalleSpectralGapStability2022}, the analysis from \cite[Theorem 7.14]{jiMIPRE2022} gives robustness of $O(a(nd)^a)(\eps^b + n^{-b} + 2^{-bnd})$ for some constant a,b and c, which is worse than the $O(\eps)$ obtain in \Cref{cor:PauliBasis}. 

We also remark that, in order to use the $\mu$-dependent Pauli Basis test as an introspection protocol for $\MIP^{co}$ following the technique from \cite{jiMIPRE2022}, the distribution $\mu$ needs also be a $l$-level conditionally linear sampler (over the field $\cF_{2^n}$ and dimensional n). We refer to \cite[definition 4.14]{jiMIPRE2022} for more details. In this case, the $\mu$-dependent Pauli Basis test would be a $2+ l$-level conditionally linear sampler. Notably, if $S \subseteq \bF_{2}^n$, and $\mu$ is a uniform distribution on $S$, then $\mu$ can be sampled using level $0$ linear sampler using $\lceil \log(|S|) \rceil$ random bits and requires $O(|S|)$ runtime to match the bit sampled to the corresponding element within $\supp(\mu)$. Since the game $\cG_{\text{Pauli}}$ from \Cref{cor:PauliBasis} uses $\mu$ to be a uniform distribution of a particular subset of size $O(n)$ for $\bF_{2}^n$. $\cG_{\text{Pauli}}$ has a sample runtime of $O(\poly(n))$ and uses $O(\polylog(n))$ bits to sample., which makes $\cG_{\text{Pauli}}$ potentially suitable for designing an introspection protocol for $\MIP^{co}$. 

Finally, we would like to remark on a similar self-test known as the low-weight Pauli Braiding test introduced recently at \cite[Section 4]{broadbentQuantumDelegationOfftheshelf2024}, where informally, the goal is to force honest players to perform a small amount of X or Z measurement within a large amount of entanglement, if $S$ is the subset of $\bF_{2}^n$ consisting of ``low weight bases", the low-weight Pauli Braiding can be seen as a $\mu$-dependent Pauli basis test, but without having the Pauli question (where instead, the honest measurement is made during the coordinate question). Currently, the robustness of the Low-weight Pauli Brading test is shown to be $O(\eps n^c)$, where n is the number of qubits and c is the weight for the amount of Pauli measurement. We remark that although it is tempting to use spectral gap-related techniques to improve the analysis from \cite{broadbentQuantumDelegationOfftheshelf2024}, we currently do not see a way to apply \Cref{thm:Ponicarethm} to the Low-weight Pauli Braiding test due to the requirement that the representation on $G$ needs to be exact (which only occurs if the players are forced to perform an all X or all Z measurement over all their EPR pairs). We leave this as an open problem to the readers.  
\begin{comment}
	
	\begin{theorem}[\cite{delasalleSpectralGapStability2022}, Lemma 2.2]
		Let $G$ be a finite group and let $\mu$ be a symmetric and generating probability measure on $G$. Then for all unitary representation $\pi: G \rightarrow \alicealg$ to a tracial von Neumann algebra $(\alicealg, \tau)$ and $V \in \alicealg$,
		\begin{equation*}
			\cE_{g \in G} || V - \pi(G)^* V \pi(G) ||_2^2 \leq \kappa(\mu) \sum_{g \in G} \mu(g)  || V - \pi(G)^* V \pi(G) ||_2^2\\
		\end{equation*}
	\end{theorem}
	
\end{comment}

\printbibliography
\appendix
\section{Proof for the Rounding Theorem} \label{App:Roundingproof}
In this section, we prove the rounding theorem. We start by introducing a few lemmas about von Neumann algebras. The first lemma we recall is the orthogonalization lemma \cite[Theorem 1.2]{delasalleOrthogonalizationPositiveOperator2022}, which is used to approximate a set of POVMs on a von Neumann algebra by a PVM being defined on the same algebra. 

\begin{lemma}[Orthogonalization lemma] \label{lem:orthogonalizationlemma}
	Let $\alicealg \subseteq \BofH$ be a von Neumann algebra and let $\ket{\psi} \in \cH$ be a unit vector. For any POVM $\{ A_a \} \subseteq \alicealg$ such that $\sum_a \braket{\psi|A_a^2|\psi} >1-  \epsilon$, there exists a PVM $\{P_a\} \subseteq \alicealg$ such that
	\begin{equation*}
		\sum_{a} \braket{\psi| (A_a - P_a)^2 |\psi}< 9 \eps.
	\end{equation*}
\end{lemma} 
If $(\alicealg, \tau)$ is a tracial von Neumann algebra in standard form, we can replace $\ket{\psi}$ by $\sigma \ket{\tau}$ for some $\sigma \in \alicealg$ for the lemma above in order to obtain a Hilbert-Schmidt norm approximation of the original POVM. The second lemma originates from the seminal paper of Connes \cite{connesClassificationInjectiveFactors1976}, and it is an important tool used in both \cite{vidickAlmostSynchronousQuantum2022} and \cite{paul-paddockRoundingNearoptimalQuantum2022}. For $\lambda \in \mathbb{R}$, we define the characteristic function $\chi_{\geq \lambda}: \bR \rightarrow \bR$ as
\begin{equation*}
	\chi_{\geq \lambda}(x) := \begin{cases}
		1 & \text{if } x \geq \lambda \\
		0 & \text{otherwise}.
	\end{cases} 
\end{equation*}
We extend, via functional calculus, the above function acting on $\alicealg^{+}$. Since $\int_{0}^{\infty} \chi_{\geq \sqrt{\lambda}}(x) d \lambda  = x^2$ for all $x \in \bR$, by extension, for all $\sigma \in \alicealg^{+}$
\begin{equation} \label{eq:projectordecomp}
	\int_{0}^{\infty} \chi_{\geq \sqrt{\lambda}}(\sigma) d \lambda  = \sigma^2.
\end{equation}
We recall the following lemma due to Connes \cite[Lemma 1.2.6 ]{connesClassificationInjectiveFactors1976}. 
\begin{lemma}[Connes' joint distribution lemma]\label{lem:ConnesEmbedding}
	Let $(\alicealg, \tau)$ be a tracial von Neumann algebra and $\rho, \sigma$ be two positive elements in $\alicealg$. Then
	\begin{equation*}
		\int_{0}^{\infty} || \chi_{\geq \sqrt{\lambda}}(\rho) -  \chi_{\geq \sqrt{\lambda}}(\sigma) ||_2^{2} d \lambda \leq  || \rho - \sigma ||_2 \cdot || \rho  + \sigma ||_2.
	\end{equation*}
\end{lemma}

The last lemma shows the existence of an analogue for inverses in some scenarios for non-invertible elements. We remark that this statement is true even if $\alicealg$ is a $C$*-algebra, as the proof below only uses the norm closure property of $C$*-algebras. 
\begin{lemma}[Existence of a left and right inverse]\ \label{lem:inversetrick1}
	Let $A, B \in \alicealg^{+}$ such that $B \leq A$, then there exists some element $R \in \alicealg$ with $R^* R \leq \cI$ and $A^{\frac{1}{2}} R  = B$. Furthermore, there exists some element $L \in \alicealg$ with $L^* L \leq \cI$ and $L A^{\frac{1}{2}} = B$.
\end{lemma}
\begin{proof}
	Fix $n \in \bN$, we see that $0<\frac{\cI_{\alicealg}}{n} \leq A + \frac{\cI_{\alicealg}}{n} \leq (\|A\|_{op} + \frac{1}{n}) \cI_{\alicealg}$, which implies that $A + \frac{\cI_{\alicealg}}{n}$ is positive and invertible in $\alicealg$. Let $R_n = (A + \frac{\cI_{\alicealg}}{n})^{-\frac{1}{2}} B^{\frac{1}{2}}$. We note that $L_n$ forms a Cauchy sequence in the operator norm, and since $B\leq A \leq A + \frac{\cI_{\alicealg}}{n}$, we have $0 \leq R_n^* R_n \leq \cI_{\alicealg}$. Hence, since $\alicealg$ is closed in the operator norm, the sequence $R_n$ converges to some positive element $R \in \alicealg^{+}$ with $R^* R \leq \cI$, and by continuity, $A^{\frac{1}{2}}R = B^{\frac{1}{2}}$. The second part of the proof follows from considering the sequence $L_n =  B^{\frac{1}{2}} (A + \frac{\cI_{\alicealg}}{n})^{-\frac{1}{2}}$.
\end{proof}

%Cauchy Schwarz and trama
\subsection{Lemmas about symmetric strategies} \label{sec:symstragproof}

In this subsection, we show some lemma about symmetric strategies. The first lemma regarding symmetric strategy is bounding the synchronicity of any tracially embeddable strategy by the synchronicity of the two symmetric strategies defined by their respective measurement operators. 

\begin{lemma} \label{lem:theviennalemma}
	Let $\cG = (\cX^2, \cA^2, \mu, D)$ be a synchronous game, $\strategy = ( \cL^2(\alicealg, \tau), \sigma \ket{\tau}, \{ A_{a}^x \},  \\ \{ (B_{b}^y)^{op} \} )$ be a tracially embeddable strategy for the game $\cG$, then
	\begin{equation*}
		1 - \deltasync\left(\cG, \strategy\right) \leq \sqrt{1 - \deltasync(\cG,  \strategy_{(A_a^x, \sigma \ket{\tau})})} \sqrt{ 1 - \deltasync(\cG, \strategy_{(B_b^y, \sigma \ket{\tau})} },
	\end{equation*}
	where $\strategy_{(A_a^x, \sigma \ket{\tau})}$ is the symmetric strategy defined by $( \cL^2(\alicealg, \tau), \sigma \ket{\tau},\{ A_a^x \})$, and likewise with $\strategy_{(A_a^x, \sigma \ket{\tau})}$. 
\end{lemma}
\begin{proof}
	By the Cauchy-Schwarz's inequality and Jensen's inequality,
	\begin{align*}
		1 - \deltasync(\cG, \strategy) &= \Ex_{x \sim \mu} \sum_{a} \braket{\tau| \sigma A_x^a \sigma B_x^a \sigma|\tau } \leq\sqrt{\Ex_{x \sim \mu} \sum_{a} \braket{\tau| \sigma A_a^x \sigma A_a^x|\tau }  } \sqrt{\Ex_{x \sim \mu} \sum_{a} \braket{\tau|  \sigma B_a^x \sigma B_a^x |\tau }  } \\
		&= \sqrt{1 - \deltasync(\cG,  \strategy_{(A_a^x, \sigma \ket{\tau})})} \sqrt{ 1 - \deltasync(\cG, \strategy_{(B_b^y, \sigma \ket{\tau})} }. \qedhere
	\end{align*}
\end{proof}
The following lemma is an analogue of \cite[Lemma 2.10]{vidickAlmostSynchronousQuantum2022} with tracially embeddable strategies, and the proof follows a similar structure. 

\begin{lemma} \label{lem:measurementtocoorlation}
	Let $\cG = (\cX^2, \cA^2, \mu, D)$ be a synchronous game, and let $\strategy = ( \cL^2(\alicealg, \tau), \sigma \ket{\tau},\{ A_{a}^x \},  \{ (B_{b}^y)^{op} \} )$ be a tracially embeddable strategies for the game $\cG$. Let $\{C_x^a\} \subseteq \alicealg$ be another set of POVMs, with
	\begin{equation*}
		\gamma = \Ex_{x \sim \mu} \sum_a \braket{\tau| \sigma(A_a^x - C_a^x)^2  \sigma |\tau},
	\end{equation*}
	for some $\gamma \geq 0$. Then for the symmetric strategy $\strategy_{(A_a^x, \sigma \ket{\tau})} = ( \cL^2(\alicealg, \tau), \sigma \ket{\tau},\{ A_{a}^x \})$, we have
	\begin{equation*}
		\Ex_{(x,y) \sim \mu} \sum_{a,b} |\braket{\tau| \sigma A_a^x (B_y^b)^{op} \sigma |\tau} -\braket{\tau| \sigma C_a^x  (B_b^y)^{op} \sigma |\tau} | \leq 6 \left( \deltasync(\cG,  \strategy_{(A_a^x, \sigma \ket{\tau})}) +\sqrt{\gamma} \right). 
	\end{equation*}
\end{lemma}

\begin{proof}
	For simplicity of notation, let $\delta_{A} = \deltasync(\cG,  \strategy_{(A_a^x, \sigma \ket{\tau})})$. We show this lemma using the following three inequalities, followed by the triangle inequality:
	\begin{align}
		&\Ex_{x,y \sim \mu} \sum_{a,b}|\braket{\tau| \sigma A_{a}^x (B_{b}^{y})^{op} \sigma |\tau} -\braket{\tau| \sigma (A_{a}^x)^2 (B_{b}^{y})^{op} \sigma |\tau} | \leq 2\delta_{A}, \label{eq:distanceboundeq1} \\
		&\Ex_{x,y \sim \mu} \sum_{a,b}|\braket{\tau| \sigma (A_{a}^x)^2 (B_{b}^{y})^{op} \sigma |\tau} - \braket{\tau| \sigma (C_{a}^x)^2 (B_{b}^{y})^{op} \sigma |\tau}| \leq 2 \sqrt{\gamma}, \label{eq:distanceboundeq3} \\
		&\Ex_{x,y \sim \mu} \sum_{a,b}|\braket{\tau| \sigma C_{a}^x (B_{b}^{y})^{op} \sigma |\tau} -\braket{\tau| \sigma (C_{a}^x)^2 (B_{b}^{y})^{op} \sigma |\tau} | \leq  4\left(\delta_A + \sqrt{\gamma} \right). \label{eq:distanceboundeq2} 
	\end{align}
	For \eqref{eq:distanceboundeq1},  
	\begin{align}
		\nonumber \Ex_{x,y \sim \mu} \sum_{a,b}|\braket{\tau| \sigma A_{a}^x (B_{b}^{y})^{op} \sigma |\tau} -\braket{\tau| \sigma (A_{a}^x)^2 (B_{b}^{y})^{op} \sigma |\tau} |   &= \Ex_{(x,y) \sim \mu} \sum_{a,b} \braket{\tau| \sigma( A_{a}^x -  (A_{a}^x)^2 ) (B_{b}^{y})^{op} \sigma |\tau} \\
		&= 1 - \Ex_{x \sim \mu} \sum_{a}\braket{\tau| \sigma (A_{a}^x)^2  \sigma |\tau}. \label{eq:measuretocoor1}
	\end{align}
	By the Cauchy-Schwarz's inequality and Jensen's inequality,
	\begin{align*}
		1 - \delta_{A} =  \Ex_{x \sim \mu} \sum_{a}\braket{\tau| \sigma A_{a}^x (A_{a}^x)^{op}  \sigma |\tau} &\leq \sqrt{\Ex_{x \sim \mu} \sum_{a}\braket{\tau| \sigma (A_{a}^x)^2  \sigma |\tau} }  \sqrt{\Ex_{x \sim \mu} \sum_{a}\braket{\tau| \sigma ((A_{a}^x)^{op})^2  \sigma |\tau} }\\
		&\leq  \sqrt{\Ex_{x \sim \mu} \sum_{a}\braket{\tau| \sigma (A_{a}^x)^2  \sigma |\tau} }.
	\end{align*}
	Solving for $\Ex_{x \sim \mu} \sum_{a}\braket{\tau| \sigma (A_{a}^x)^2  \sigma |\tau}$, we have
	\begin{equation} \label{eq:measuretocoor2}
		1 - 2 \delta_{A} \leq 1 -2 \delta_{A} + \delta_{A}^2 \leq \Ex_{x \sim \mu} \sum_{a}\braket{\tau| \sigma (A_{a}^x)^2  \sigma |\tau},
	\end{equation}
	and hence \eqref{eq:distanceboundeq1} follows accordingly. For \eqref{eq:distanceboundeq3}, we bound the inequality through the following two inequalities:
	\begin{align}
		\Ex_{x,y \sim \mu} \sum_{a,b} \braket{\tau| \sigma ((A_a^x)^2-A_a^x C_a^x) (B_b^y)^{op} \sigma| \tau} &\leq \sqrt{\gamma}, \label{eq:distanceboundeq10}\\
		\nonumber \Ex_{x,y \sim \mu} \sum_{a,b} \braket{\tau| \sigma ((C_a^x)^2-A_a^x C_a^x) (B_b^y)^{op}\sigma| \tau} &\leq \sqrt{\gamma}.
	\end{align}
	Since the inequalities are proven in a similar manner, we will only show the proof for \eqref{eq:distanceboundeq10}: by Cauchy-Schwarz's inequality and Jensen's inequality,
	\begin{align*}
		&\Ex_{x,y \sim \mu} \sum_{a,b} \braket{\tau| \sigma ((A_a^x)^2-A_a^x C_a^x) (B_b^y)^{op} \sigma| \tau} \\
		&\quad \leq \sqrt{\Ex_{x,y \sim \mu} \sum_{a,b} \braket{\tau| \sigma(A_a^x)^2 ((B_b^y)^2)^{op} \sigma| \tau}}\sqrt{\Ex_{x,y \sim \mu} \sum_{a,b} \braket{\tau| \sigma(A_a^x- C_a^x)^2 \sigma| \tau}} \\
		&\quad \leq 1 \cdot \sqrt{\Ex_{x,y \sim \mu} \sum_{a} \braket{\tau| \sigma (A_a^x- C_a^x)^2\sigma| \tau}} \leq \sqrt{\gamma}. 
	\end{align*}
	For \eqref{eq:distanceboundeq2}, this is analogous to \eqref{eq:distanceboundeq1}, except we are estimating $\delta_C = \deltasync(\cG,  \strategy_{(C_a^x, \sigma \ket{\tau})})$ instead, where $\strategy_{(C_a^x, \sigma \ket{\tau})}$ is the symmetric strategy $( \cL^2(\alicealg, \tau), \sigma \ket{\tau},\{ C_{a}^x \} )$. To bound $\delta_C$, let $\eta$ denote the synchronicity for the tracially embeddable strategy $( \cL^2(\alicealg, \tau), \sigma \ket{\tau},\{ C_{a}^x \},  \{ (A_{b}^y)^{op} \} )$. By Cauchy-Schwarz's inequality and Jensen's inequality,
	\begin{align*}
		\eta - \delta_A  \leq  &\Ex_{x \sim \mu} \sum_{a}  |\braket{\tau| \sigma (C_a^x) (A_a^x)^{op} \sigma| \tau} - \braket{\tau| \sigma (A_a^x) (A_a^x)^{op} \sigma| \tau}| \\
		&= \Ex_{x \sim \mu} \sum_{a}  |\braket{\tau| \sigma (C_a^x - A_a^x) (A_a^x)^{op} \sigma| \tau} | \\
		&\leq \sqrt{\Ex_{x \sim \mu} \sum_{a}  \braket{\tau| \sigma (C_a^x - A_a^x)^2  \sigma| \tau}}\sqrt{\Ex_{x \sim \mu} \sum_{a} \braket{\tau|  \left((A_a^x)^{2}\right)^{op} \sigma| \tau}} \leq \sqrt{\gamma}.
	\end{align*}
	By solving for $\eta$ and by \Cref{lem:theviennalemma}, $\delta_A \leq 2\eta  \leq 2\left(\delta_A + \sqrt{\gamma}\right)$. Proceeding the same way as \eqref{eq:distanceboundeq1}
	\begin{align*}
		\Ex_{x,y \sim \mu} \sum_{a,b}|\braket{\tau| \sigma C_{a}^x (B_{b}^{y})^{op} \sigma |\tau} -\braket{\tau| \sigma (C_{a}^x)^2 (B_{b}^{y})^{op} \sigma |\tau} | &\leq 1 - \Ex_{x \sim \mu} \sum_{a}\braket{\tau| \sigma (C_{a}^x)^2  \sigma |\tau} \\
		&\leq 2 \delta_C \leq 4\left(\delta_A + \sqrt{\gamma}\right). \qedhere
	\end{align*}
\end{proof}
Combining \Cref{lem:theviennalemma} and \Cref{lem:measurementtocoorlation}, we have the following lemma. 
\begin{lemma} \label{lem:corrtosympro}
	Let $\cG = (\cX^2, \cA^2, \mu, D)$ be a synchronous game, and let $\strategy = (\cL^2(\alicealg, \tau),\sigma \ket{\tau}, \{A_a^x \} , \{B_b^y\})$ be a $\delta$-synchronous strategy for $\cG$. Then the symmetric strategy $\strategy' = (\cL^2(\alicealg, \tau),\sigma \ket{\tau}, \{A_a^x \})$ is a $O(\delta)$-synchronous strategy with 
	\begin{equation*}
		\Ex_{(x,y) \sim \mu} \sum_{a,b} |\braket{\tau| \sigma A_a^x (B_b^y)^{op} \sigma |\tau} -\braket{\tau| \sigma A_a^x  (A_b^y)^{op} \sigma |\tau} | \leq O(\sqrt{\delta}),
	\end{equation*}
\end{lemma}
\begin{proof}
	Let $\delta_A$ be the synchronicity of the strategy $\strategy'$, by \Cref{lem:theviennalemma}, $\delta_A \leq 2 \delta = O(\delta)$. For the second part of the lemma, we see that by triangle inequality, we have 
	\begin{equation} \label{eq:corrtosympro4}
		\Ex_{x \sim \mu} \sum_a \braket{\tau| \sigma ((A_a^x - B_a^x)^{op})^2 \sigma |\tau} \leq \Ex_{x \sim \mu} \sum_a (||(A_a^x - (B_a^x)^{op}) \sigma \ket{\tau}||^2 +  ||(A_a^x - (A_a^x)^{op}) \sigma \ket{\tau}||^2).
	\end{equation}
	For the first part of \eqref{eq:corrtosympro4}, since $(A_a^x)^2 \leq A_a^x \leq \cI$, 
	\begin{align*} 
		\Ex_{x \sim \mu} \sum_a \|(A_a^x - (B_a^x)^{op}) \sigma \ket{\tau}\|^2 &= \Ex_{x \sim \mu} \sum_a \left( \braket{\tau|\sigma (A_a^x)^2 \sigma|\tau} +  \braket{\tau|\sigma ((B_a^x)^2)^{op} \sigma|\tau} - 2  \braket{\tau|\sigma A_a^x  (B_a^x)^{op} \sigma|\tau} \right) \\
		&\leq \Ex_{x \sim \mu} \sum_a \left( 2 - 2  \braket{\tau|\sigma A_a^x  (B_a^x)^{op} \sigma|\tau} \right) = 2 \delta.
	\end{align*}
	For the second part of \eqref{eq:corrtosympro4}, by a similar calculation
	\begin{align} \label{eq:corrtosympro3}
		\Ex_{x \sim \mu} \sum_a \|(A_a^x - (A_a^x)^{op}) \sigma \ket{\tau}\|^2 \leq  \Ex_{x \sim \mu} \sum_{a} (2 - 2\braket{\tau| \sigma A_x^a (A_x^a)^{op} \sigma |\tau}) \leq 2\delta_A\leq O( \delta). 
	\end{align}
	Hence, by combining the above two equations, we have
	\begin{equation*}
		\Ex_{x \sim \mu} \sum_a \braket{\tau| \sigma ((A_a^x - B_a^x)^{op})^2 \sigma |\tau}  \leq O(\delta).
	\end{equation*}
	By \Cref{lem:measurementtocoorlation} and as well as \Cref{lem:theviennalemma} apply on the symmetric strategy $\strategy^B = (\cL^2(\alicealg, \tau),\sigma \ket{\tau}, \{B_b^y \}_{b \in \cA})$, we have
	\begin{equation*}
		\Ex_{(x,y) \sim \mu} \sum_{a,b}|\braket{\tau| \sigma |A_a^x (B_b^y)^{op} \sigma |\tau}  -\braket{\tau| \sigma A_a^x  (A_b^y)^{op} \sigma |\tau}| \leq O(\sqrt{\delta}) + O(\delta_B) \leq  O(\sqrt{\delta})
	\end{equation*}
	where $\delta_B$ denotes the synchronicity of $\strategy^B$ this concludes the second claim of the lemma. 
\end{proof}

The following corollary shows that every $\delta$-synchronous correlation can be approximated by some correlations generated by a symmetric projective strategy defined on the standard form of a tracial von Neumann algebra $(\alicealg, \tau)$. 
\begin{corollary} \label{cor:orthogonalizationlemmasync}
	Let $\cG = (\cX^2, \cA^2, \mu, D)$ be a synchronous game, and let $(\cL^2(\alicealg, \tau),\sigma \ket{\tau}, \{A_a^x \} , \{B_b^y\})$ be a $\delta$-synchronous strategy for $\cG$. There exists a symmetric, projective and $\delta^{\frac{1}{4}}$-synchronous strategy $( \cL^2(\alicealg, \tau), \sigma \ket{\tau},\{ P_{a}^x \})$ with 
	\begin{equation*}
		\Ex_{(x,y) \sim \mu} \| \left( A_a^x - P^x_a\right) \sigma \ket{\tau}\|^2 \leq O(\delta),
	\end{equation*}
	such that
	\begin{equation}\label{eq:orthogonalizationlemmasynceq}
		\Ex_{x \sim \mu} \sum_{a \in A} |\braket{\tau| \sigma A_a^x (B_b^y)^{op} \sigma |\tau} -\braket{\tau| \sigma P_a^x (P_b^y)^{op} \sigma |\tau}| \leq O(\delta^{\frac{1}{4}}),
	\end{equation} 
\end{corollary}

\begin{proof} 
	Let $\strategy^{sym} = (\cL^2(\alicealg, \tau),\sigma \ket{\tau}, \{A_a^x \}_{a \in \cA})$ be a symmetric strategy. By \Cref{lem:corrtosympro}, $ \deltasync(\cG, \strategy^{sym}) = O(\delta)$ and 
	\begin{equation} \label{eq:ortho0}
		\Ex_{(x,y) \sim \mu} \sum_{a,b} |\braket{\tau| \sigma A_a^x (B_b^y)^{op} \sigma |\tau} -\braket{\tau| \sigma A_a^x  (A_b^y)^{op} \sigma |\tau} | \leq O(\sqrt{\delta}).
	\end{equation}
	By the Cauchy-Schwarz's inequality, 
	\begin{equation*}
		1-O(\delta) = \Ex_{x \sim \mu} \sum_{a}\braket{\tau| \sigma A_a^x (A_a^x)^{op}  \sigma |\tau} \leq \sqrt{ \Ex_{x \sim \mu} \sum_{a}\braket{\tau| \sigma (A_a^x)^2 \sigma |\tau}},
	\end{equation*}
	or $1-O(\delta) \leq \Ex_{x \sim \mu} \sum_{a}\braket{\psi|(A_a^x)^2|\psi}$. Hence, for all $x \in \cX$, by applying \Cref{lem:orthogonalizationlemma} on $\{A_a^x\}$, there exists a set of projective measurements $\{P_a^x\}$ such that
	\begin{equation} \label{eq:ortho1}
		\Ex_{x \sim \mu} \|  \left( A_a^x - P^x_a\right) \sigma \ket{\tau}\|^2 = \Ex_{x \sim \mu} \braket{\tau|\sigma (A_a^x - P_a^x)^2\sigma |\tau}  \leq O(\delta).
	\end{equation}
	By applying \Cref{lem:measurementtocoorlation} on $\{P_a^x\}$ for $\strategy^{sym}$, we have
	\begin{equation*}
		\Ex_{(x,y) \sim \mu} \sum_{a,b} |\braket{\tau| \sigma A_a^x (A_b^y)^{op} \sigma |\tau} -\braket{\tau| \sigma P_a^x  (A_b^y)^{op} \sigma |\tau} | \leq O(\sqrt{\delta}).
	\end{equation*}
	Let $\strategy^{2} = (\cL^2(\alicealg, \tau),\sigma \ket{\tau}, \{ A_a^x \}_{a \in \cA}, \{ (P_a^x)^{op} \}_{a \in \cA})$. By Cauchy-Schwarz's inequality, and $\ket{\tau}$ being a tracial state, we have 
	\begin{align*}
		\deltasync(\cG, \strategy^{2}) -  \deltasync(\cG, \strategy^{sym}) &\leq \Ex_{x} \sum_{a} \braket{\tau| \sigma A_a^x (P_a^x)^{op} \sigma |\tau} -\braket{\tau| \sigma A_a^x (A_a^x)^{op}\sigma |\tau} \\
		&\leq \sqrt{\Ex_{x} \sum_{a} \braket{\tau| \sigma (A_a^x)^2  \sigma  |\tau}}\sqrt{\Ex_{x} \sum_{a} \braket{\tau| \sigma  ((A^a_x - P_a^x)^2)^{op} \sigma  |\tau}}  \\
		&\leq 1 \cdot \sqrt{(\Ex_{x} \sum_{a} \braket{\tau| (\sigma)^2  ((A^x_a - P_a^x)^2)^{op} |\tau}} \\
		&= \sqrt{\Ex_{x} \sum_{a} \braket{\tau| \sigma  (A^x_a - P_a^x)^2 \sigma |\tau}} = O(\sqrt{\delta}),
	\end{align*}
	and hence $\delta_2 \leq O(\sqrt{\delta})$. By repeating this argument again, on $(A_a^x)^{op}$ with $\strategy_2$ and together with \eqref{eq:ortho0} and \eqref{eq:ortho1} gives \eqref{eq:orthogonalizationlemmasynceq}. 
\end{proof}

\subsection{Proof of \Cref{thm:MainRounding}}

To show \Cref{thm:MainRounding}, we first show the following proposition. This proposition is analogous to \cite[Theorem 3.1]{vidickAlmostSynchronousQuantum2022} in the infinite-dimensional case. It states that any projective, symmetric, almost synchronous strategy defined on the standard form of some tracial von Neumann algebra $(\alicealg, \tau)$ can be approximated by a set of synchronous strategies, each on some subalgebra of $\alicealg$. 
\begin{proposition}\label{prop:roundinginVNA}
	Let $\cG = (\cX^2, \cA^2, \mu, D)$ be a synchronous game and let $(\alicealg, \tau)$ be a tracial von Neumann algebra in standard form. Furthermore, let $( \cL^2(\alicealg, \tau), \sigma \ket{\tau},\{ A_{a}^x \})$ be a projective, symmetric and $\delta$-synchronous strategy. There exists a set of projectors $\{P_{\lambda} \}_{[0,\infty]} \subseteq \alicealg$ such that
	\begin{equation*}
		\int_{0}^{\infty}	P_{\lambda} d\lambda= \sigma^2,
	\end{equation*}
	and for each $\lambda$, there exists a set of PVMs $\{ A_a^{\lambda,x}  \}_{\lambda \in [0,\infty]} \subseteq P_{\lambda} \alicealg P_{\lambda}$ such that
	\begin{equation*} 
		\int_{0}^{\infty} \Ex_{x \sim \mu} \sum_{a}  || (A_a^x - A_a^{\lambda,x} ) P_{\lambda}||_2^2 d\lambda \leq O(\delta^{\frac{1}{4}}). 
	\end{equation*} 
\end{proposition}
\begin{proof} 
	Let $P_{\lambda} := \chi_{\geq \sqrt{\lambda}}(\sigma)$, by \eqref{eq:projectordecomp} $\int_{0}^{\infty} P_{\lambda} = \sigma^2$. We first show 
	\begin{equation}  \label{eq:roundingprojector4}
		\int_{0}^{\infty} \Ex_{x \sim \mu} \sum_{a}  || (A_a^x - P_{\lambda} A_a^x P_{\lambda} ) P_{\lambda}||_2^2  d\lambda \leq 2 \sqrt{\delta}. 
	\end{equation}
	By H\"older's inequality
	\begin{align*}
		\int_{0}^{\infty} \Ex_{x \sim \mu} \sum_{a}  || (A_a^x - P_{\lambda} A_a^x P_{\lambda} ) P_{\lambda}||_2^2  d\lambda \leq 	\Ex_{x \sim \mu}  \int_{0}^{\infty} \sum_{a} || (A_a^x P_{\lambda} - P_{\lambda} A_a^x)||_2^2   d\lambda.
	\end{align*}
	Let $m = |\cA|$ and associate the answer set with $\mathbb{Z}_m$ and define $U_b^x := \sum_a^m e^{\frac{2i \pi a b}{m}} A_{a}^x$. Since $A_{a}^x$ is a PVM, $U_b^x$ defines a unitary in $\alicealg$ for each $a \in \cA$ and $b \in \mathbb{Z}_m$. By substituting $A_a^x$ with $U_b^x$
	\begin{align}
		\Ex_{x \sim \mu} \int_{0}^{\infty} \sum_{a \in \mathbb{Z}_m} ||A_a^x P_{\lambda} - P_{\lambda} A_a^x||_2^2  d\lambda &= \Ex_{x \sim \mu} \int_{0}^{\infty} \left(  \sum_{a \in \mathbb{Z}_m}  2 ( \braket{\tau |P_{\lambda}  A_a^x | \tau} - \braket{\tau |P_{\lambda}   A_a^x P_{\lambda}  A_a^x| \tau}) \right) d\lambda \label{eq:roundingprojector1}\\
		\nonumber&= \Ex_{x \sim \mu}  \int_{0}^{\infty} \left( 2 \braket{\tau |P_{\lambda} | \tau} - 2\Ex_{b}\braket{\tau |P_{\lambda}  (U^x_b)^* P_{\lambda} U^x_b| \tau}  \right) d\lambda \\
		\nonumber &= \Ex_{x \sim \mu} \Ex_{b} \int_{0}^{\infty}  ||U^x_b P_{\lambda} - P_{\lambda} U^x_b||_2^2 d\lambda.
	\end{align}	\begingroup
	\allowdisplaybreaks
	By \Cref{lem:ConnesEmbedding}, since $ \chi_{\geq \sqrt{\lambda}} \left( (U_{b}^x)^* \sigma U_{b}^x \right) = (U_{b}^x)^* \left(\chi_{\geq \sqrt{\lambda}} ( \sigma)  \right) U_{b}^x $ for all $\lambda$, it follows that
	\begin{align*}
		\Ex_{x \sim \mu} \Ex_{b} \int_{0}^{\infty}  ||U^x_b P_{\lambda} - P_{\lambda} U^x_b||_2^2 d\lambda &=\Ex_{x \sim \mu} \Ex_b \int_{0}^{\infty}|| P_{\lambda} -   (U_b^x)^{*}P_{\lambda} U_b^x||_{2}^2 d\lambda \\
		&\leq \Ex_{x \sim \mu} \Ex_{b} || \sigma -  (U_b^x)^{*} \sigma U_b^x||_{2} || \sigma +  (U_b^x)^{*} \sigma U_b^x||_{2}. \\
		&\leq \sqrt{\Ex_{x \sim \mu} \Ex_{b} || \sigma -  (U_b^x)^{*} \sigma U_b^x||_{2}^2 } \cdot \sqrt{ \Ex_{x \sim \mu} \Ex_{b} || \sigma +  (U_b^x)^{*} \sigma U_b^x||_{2}^2 }\\
		&\leq \sqrt{2} \cdot \sqrt{\Ex_{x \sim \mu} \Ex_{b} || \sigma -  (U_b^x)^{*} \sigma U_b^x||_{2}^2 }.
	\end{align*}
	\endgroup
	By a similar calculation as \eqref{eq:roundingprojector1}
	\begin{align*}
		\sqrt{2} \cdot \sqrt{\Ex_{x \sim \mu} \Ex_{b} || \sigma -  (U_b^x)^{*} \sigma U_b^x||_{2}^2  } &=  \sqrt{2} \cdot \sqrt{2 - 2\Ex_{x \sim \mu} \sum_{a} \braket{\tau| \sigma A_a^x \sigma A_a^x |\tau}} =  \sqrt{2} \cdot \sqrt{2 \delta} = 2 \sqrt{\delta},
	\end{align*}
	showing \eqref{eq:roundingprojector4}. Now, for each $\lambda$, the algebra $P_{\lambda} \alicealg P_{\lambda}$ is a tracial von Neumann algebra with the corresponding tracial state being $\frac{1}{\sqrt{\tau(P_{\lambda})}}P_{\lambda} \ket{\tau}$. Furthermore, $\{P_{\lambda} A_a^x P_{\lambda}  \}$ defines a set of POVMs within the algebra $P_{\lambda} \alicealg P_{\lambda}$, and we can use \Cref{lem:orthogonalizationlemma} to construct a set of PVMs to approximate the above POVM. To this end, observe that
	\begin{align}
		\nonumber 1- 2\delta &\leq  \Ex_{x \sim \mu} \sum_{a} \braket{\tau| \sigma (A^x_a)^2 \sigma|\tau} = \Ex_{x \sim \mu} \sum_{a} \int_{0}^{\infty} \braket{\tau| P_{\lambda} (A^x_a)^2 P_{\lambda}|\tau} d \lambda\\
		&= \Ex_{x \sim \mu} \sum_{a} \int_{0}^{\infty} \left(\braket{\tau| P_{\lambda} A^x_a P_{\lambda} A^x_a |\tau} +  \braket{\tau| P_{\lambda} A^x_a ( A^x_a P_{\lambda} - P_{\lambda} A^x_a  ) |\tau} \right) d \lambda. \label{eq:roundingprojector2}
	\end{align}
	To bound the second term of \eqref{eq:roundingprojector2}, by Cauchy-Schwarz's inequality, then by \eqref{eq:roundingprojector1}
	\begin{align}
		\nonumber &\Ex_{x \sim \mu} \sum_{a} \int_{0}^{\infty} \braket{\tau| P_{\lambda} A^x_a ( A^x_a P_{\lambda} - P_{\lambda} A^x_a  ) |\tau}  d \lambda \\
		&\quad \leq  \sqrt{\Ex_{x \sim \mu} \sum_{a} \int_{0}^{\infty} \braket{\tau| P_{\lambda} (A^x_a)^2 P_{\lambda} |\tau}  d \lambda } \sqrt{\Ex_{x \sim \mu} \sum_{a} \int_{0}^{\infty}   ||A_a^x P_{\lambda} - P_{\lambda} A_a^x||_2^2  d \lambda}\leq  1 \cdot \sqrt{2} \delta^{\frac{1}{4}}. \label{eq:roundingprojector3} 
	\end{align}
	By rewriting the first term of \eqref{eq:roundingprojector2} and bounding the second term using \eqref{eq:roundingprojector3},
	\begin{align*}
		&\Ex_{x \sim \mu} \sum_{a} \int_{0}^{\infty} (\braket{\tau| P_{\lambda} A^x_a P_{\lambda} A^x_a |\tau} +  \braket{\tau| P_{\lambda} A^x_a ( A^x_a P_{\lambda} - P_{\lambda} A^x_a  ) |\tau}) \\
		&\quad \leq 	\Ex_{x \sim \mu} \sum_{a} \int_{0}^{\infty} \braket{\tau| P_{\lambda} (P_{\lambda} A^x_a P_{\lambda})^2P_{\lambda}|\tau}  +  \sqrt{2} \delta^{\frac{1}{4}}. 
	\end{align*}
	Solving for $\Ex_{x \sim \mu} \sum_{a} \int_{0}^{\infty} \braket{\tau| P_{\lambda} (P_{\lambda} A^x_a P_{\lambda})^2P_{\lambda}|\tau}$, we have
	\begin{equation}
		1 - 2 \delta -\sqrt{2} \delta^{\frac{1}{4}} \leq \Ex_{x \sim \mu} \sum_{a} \int_{0}^{\infty} \braket{\tau| P_{\lambda} (P_{\lambda} A^x_a P_{\lambda})^2P_{\lambda}|\tau}.
	\end{equation}
	Hence, by \Cref{lem:orthogonalizationlemma}, for each $\lambda \in \bR$ and $x \in \cX$, there exists a set of PVMs $\{ A_a^{\lambda, x} \} \subseteq P_{\lambda} \alicealg P_{\lambda}$ such that
	\begin{equation}  
		\int_{0}^{\infty} \Ex_{x \sim \mu} \sum_{a}  \braket{\tau| P_{\lambda} ( P_{\lambda} A_{a}^x P_{\lambda} - A_a^{\lambda, x})^2 P_{\lambda}|\tau}  d\lambda \leq 2 \delta + \sqrt{2} \delta^{\frac{1}{4}} = O(\delta^{\frac{1}{4}}). \label{eq:roundingprojector5}
	\end{equation}
	Combining \eqref{eq:roundingprojector4} and \eqref{eq:roundingprojector5} via the triangle inequality
	\begin{equation*} 
		\int_{0}^{\infty} \Ex_{x \sim \mu} \sum_{a}  || (A_a^x - A_a^{\lambda,x} ) P_{\lambda}||_2^2 d\lambda \leq O(\delta^{\frac{1}{4}}). \qedhere
	\end{equation*} 
\end{proof}
We are now ready to give the proof of \Cref{thm:MainRounding}. 
\begin{proof}
	Let $\{C_{x,y,a,b}\} \in \cC_{qc}(\cX, \cA)$ be a $\delta$-synchronous correlation for the game $\cG$. By \Cref{thm:tracialeEmbedding}, we can, by perturbing $C_{x,y,a,b}$ up to arbitrary $\eps >0$ accuracy, assume the existence of some tracially embeddable strategy $(\cL^2(\alicealg, \tau),\sigma \ket{\tau}, \{A_a^x \}_{a \in \cA} , \{B_b^y\}_{b \in \cA})$ on a tracial von Neumann algebra $(\alicealg, \tau)$ such that
	\begin{equation*}
		C_{x,y,a,b} = \braket{\tau|\sigma A_{a}^x (B_{b}^y)^{op}\sigma|\tau}.
	\end{equation*}
	By \Cref{cor:orthogonalizationlemmasync}, there exist a set of PVMs $\{ Q_a^x \}$ and a $O(\delta^{\frac{1}{4}})$-synchronous, symmetric and projective strategy $\strategy^{'} = (\cL^2(\alicealg, \tau),\sigma\ket{\tau}, \{Q_a^x \}_{a \in \cA})$ such that
	\begin{equation} \label{eq:roundingmainthmeq1}
		\Ex_{(x,y) \sim \mu} \sum_{a,b} |\braket{\tau| \sigma A_a^x (B_y^b)^{op} \sigma |\tau} -\braket{\tau| \sigma Q_a^x  (Q_b^y)^{op} \sigma |\tau} | \leq O(\delta^{\frac{1}{4}}),
	\end{equation}
	with 
	\begin{equation} \label{eq:roundingmainprojectordistance}
		\Ex_{x \sim \mu} \| \left( A_a^x - Q^x_a\right) \sigma \ket{\tau}\|^2 \leq O(\delta).
	\end{equation}
	Applying \Cref{prop:roundinginVNA} to $\strategy^{'}$, we get a set of projectors $\{P_{\lambda} \}_{[0,\infty]} \subseteq \alicealg$ with $\int_{0}^{\infty}P_{\lambda} d\lambda= (\sigma^+)^{2}$, and a set of PVMs $\{ A_a^{\lambda,x}  \}_{\lambda \in [0,\infty], x \in \cX}$, with $\{ A_a^{\lambda,x}  \} \subseteq P_{\lambda} \alicealg P_{\lambda}$ for all $\lambda \in \mathbb{R}^{+}$, such that
	\begin{equation} \label{eq:roundingmainthmeq2}
		\int_{0}^{\infty} \Ex_{x \sim \mu} \sum_{a}  || (Q_a^x - A_a^{\lambda,x} ) P_{\lambda}||_2^2 d\lambda \leq O(\delta^{\frac{1}{4}}). 
	\end{equation} 
	For $\tau(P_{\lambda}) \neq 0$, the symmetric strategy $\strategy^{\lambda} = ( \cL^2( P_{\lambda} \alicealg P_{\lambda}, P_{\lambda} \ket{\tau}),  \frac{1}{\sqrt{\tau(P_\lambda)}} P_{\lambda} \ket{\tau},\{  A_a^{\lambda,x}  \} )$ is synchronous since $ \frac{1}{\sqrt{\tau(P_\lambda)}} P_{\lambda} \ket{\tau}$ is a tracial state for $P_{\lambda} \alicealg P_{\lambda}$ (see the remark after \Cref{def:symstrategy}). Let $B_a^{\lambda,x}$ be the corresponding opposite algebra element for $A_a^{\lambda,x}$ define on the tracial von Neumann algebra $(P_{\lambda} \alicealg P_{\lambda}, \frac{1}{\sqrt{\tau{P_{\lambda}}}} P_{\lambda})$. In other words, we have $ (A_a^{\lambda,x}) P_{\lambda} \ket{\tau} =  (B_a^{\lambda,x}) P_{\lambda} \ket{\tau}$ for all $x \in \cX$ and $a \in \cA$. We claim that for all $x \in \cX$ and $a \in \cA$
	\begin{equation*}
		\braket{\tau|P_{\lambda}  A_a^{\lambda,x}   B_a^{\lambda,x} P_{\lambda}|\tau} = \braket{\tau|P_{\lambda} ( A_a^{\lambda,x} ) ( A_a^{\lambda,x} )^{op} P_{\lambda}|\tau} 
	\end{equation*}
	where the $op$ map above corresponds to the opposite algebra map for $(\alicealg, \ket{\tau})$. In other words, we have the synchronous strategy $\strategy^{\lambda}$ can be expressed as a symmetric strategy within the tracial von Neumann algebra $(\alicealg, \tau)$. Since $A_a^{\lambda,x} \in P_{\lambda} \alicealg P_{\lambda}$, we can express $A_a^{\lambda,x}$ as $P_{\lambda} A_a^{\lambda,x}  P_{\lambda}$ in $\alicealg$, and hence 
	\begin{equation*}
		\braket{\tau|P_{\lambda} ( A_a^{\lambda,x} ) ( A_a^{\lambda,x} )^{op} P_{\lambda}|\tau}  = 	\braket{\tau|P_{\lambda} ( A_a^{\lambda,x} )  P_{\lambda}  (B_a^{\lambda,x}   P_{\lambda})  |\tau}= 	\braket{\tau|P_{\lambda} A_a^{\lambda,x} B_a^{\lambda,x}   P_{\lambda}   |\tau},
	\end{equation*}
	where the last equation follows from $(P_{\lambda} \alicealg P_{\lambda})' = \alicealg' P_{\lambda}$ for any von Neumann algebra $\alicealg$ and projector $P_{\lambda} \in \alicealg$. Hence, showing the above claim. 
	
	Since $\int_0^{\infty} \tau(P_\lambda) d \lambda = \tau\left(\int_0^{\infty} P_\lambda d \lambda\right) = \tau(\sigma^2) = 1$, we can define the probability distribution $P$ over $[0, \infty)$ by $P(\lambda) = \tau(P_{\lambda})$. Let $C^{\lambda}_{x,y,a,b}$ denote the correlation set generated by the strategy $\strategy^{\lambda}$ for $\lambda$ with $\tau(P_{\lambda}) \neq 0$ and any arbitrary synchronous correlation otherwise. We claim that
	\begin{equation} \label{eq:roundingmainthmeq7}
		\Ex_{(x,y) \sim \mu} \sum_{a,b}|\braket{\tau| \sigma Q_a^x  (Q_b^y)^{op} \sigma |\tau} - \int_{0}^{\infty} P(\lambda) \cdot C_{x,y,a,b}^{\lambda} d \lambda| \leq O(\delta^{\frac{1}{8}}).
	\end{equation} 
	\begingroup
	\allowdisplaybreaks
	To show this, we prove the following three inequalities,
	\begin{align}
		\Ex_{(x,y) \sim \mu} \sum_{a,b}|\braket{\tau| \sigma Q_a^x  (Q_b^y)^{op} \sigma |\tau} - \int_{0}^{\infty} \braket{ \tau |(Q_a^x)^{op} (Q_b^y)^{op} P_{\lambda} |\tau} d \lambda| &\leq  O(\delta^{\frac{1}{4}}),  \label{eq:roundingmainthmeq3} \\
		\Ex_{(x,y) \sim \mu} \sum_{a,b}|\int_{0}^{\infty} \braket{\tau| P_{\lambda} Q_a^x  (Q_b^y)^{op} P_{\lambda} |\tau} d \lambda - \int_{0}^{\infty} \braket{ \tau |(Q_a^x)^{op} (Q_b^y)^{op} P_{\lambda} |\tau} d \lambda | &\leq O(\delta^{\frac{1}{8}}),  \label{eq:roundingmainthmeq4} \\		
		\Ex_{(x,y) \sim \mu} \sum_{a,b}|\int_{0}^{\infty} \braket{\tau| P_{\lambda} Q_a^x  (Q_b^y)^{op} P_{\lambda} |\tau} d \lambda - \int_{0}^{\infty} P(\lambda) \cdot C_{x,y,a,b}^{\lambda} d \lambda| d\lambda &\leq O(\delta^{\frac{1}{8}}) .  \label{eq:roundingmainthmeq5}
	\end{align}
	For \eqref{eq:roundingmainthmeq3}, since  $\strategy^{'}$ is $O(\delta^{\frac{1}{4}})$-synchronous,
	
	\begin{align*}
		&\Ex_{(x,y) \sim \mu} \sum_{a,b}|\braket{\tau| \sigma Q_a^x  (Q_b^y)^{op} \sigma |\tau} - \int_{0}^{\infty} \braket{ \tau |(Q_a^x)^{op} (Q_b^y)^{op} P_{\lambda} |\tau} | d \lambda \\
		&= \Ex_{(x,y) \sim \mu} \sum_{a,b}|\braket{\tau| \sigma Q_a^x  (Q_b^y)^{op} \sigma |\tau} -  \braket{ \tau | \sigma(Q_b^y)^{op} (Q_b^y)^{op} \sigma |\tau}| d \lambda \\
		&\leq \sqrt{\Ex_{(x,y) \sim \mu} \sum_{a,b}|\braket{\tau| \sigma (Q_a^x - (Q_a^x)^{op})^2 \sigma |\tau} } \cdot \sqrt{\Ex_{(x,y) \sim \mu} \sum_{a,b}|\braket{\tau| \sigma (Q_b^y)^{op})^2 \sigma |\tau} } \leq O(\delta^{\frac{1}{4}}). 
	\end{align*}
	\endgroup
	For \eqref{eq:roundingmainthmeq4}, by a similar calculation as above and Jensen's inequality
	\begin{align*}
		&\Ex_{(x,y) \sim \mu} \sum_{a,b}|\int_{0}^{\infty} \braket{\tau| P_{\lambda} Q_a^x  (Q_b^y)^{op} P_{\lambda} |\tau} d \lambda - \int_{0}^{\infty} \braket{ \tau |(Q_a^x)^{op} (Q_b^y)^{op} P_{\lambda} |\tau} d \lambda \\
		&\leq \int_{0}^{\infty} \Ex_{(x,y) \sim \mu} \sum_{a,b}|\braket{\tau| P_{\lambda} (Q_a^x - (Q_a^x)^{op})  (Q_b^y)^{op} P_{\lambda} |\tau} | d \lambda \\
		&\leq\sqrt{ \int_{0}^{\infty} \Ex_{(x,y) \sim \mu} \sum_{a,b}|\braket{\tau| P_{\lambda} (Q_a^x - (Q_a^x)^{op})^2  P_{\lambda} |\tau} | d \lambda  } \sqrt{ \int_{0}^{\infty} \Ex_{(x,y) \sim \mu} \sum_{a,b}|\braket{\tau| P_{\lambda} ( (Q_b^y)^{op})^2 P_{\lambda} |\tau} | d \lambda  } \\
		&\leq\sqrt{ \int_{0}^{\infty} \Ex_{(x,y) \sim \mu} \sum_{a,b}\| P_{\lambda}Q_a^x  - Q_a^x P_{\lambda} \|_2^2 d \lambda  }  \cdot 1 \leq O(\delta^{\frac{1}{8}}),
	\end{align*}
	where the last inequality follows from \eqref{eq:roundingprojector1}. For \eqref{eq:roundingmainthmeq5}, by replacing $C^{\lambda}_{x,y,a,b}$ with $\braket{\tau| P_{\lambda} A_a^{\lambda} (A_a^{\lambda})^{op} P_{\lambda} |\tau}$ for $P(\lambda) = 0$, 
	\begin{align*}
		&\Ex_{(x,y) \sim \mu} \sum_{a,b}|\int_{0}^{\infty} \braket{\tau| P_{\lambda} Q_a^x  (Q_b^y)^{op} P_{\lambda} |\tau} d \lambda - \int_{0}^{\infty} P(\lambda) \cdot C_{x,y,a,b}^{\lambda} d \lambda| d\lambda \\
		&\leq \int_{0}^{\infty}  \Ex_{(x,y) \sim \mu} \sum_{a,b}|\braket{\tau| P_{\lambda} Q_a^x  (Q_b^y)^{op} P_{\lambda} |\tau} d \lambda - \braket{\tau| P_{\lambda}  A_a^{\lambda} (A_a^{\lambda})^{op} P_{\lambda} |\tau} | d\lambda.
	\end{align*}
	Hence, by \eqref{eq:roundingmainthmeq2} and $\ket{\tau}$ being tracial, $\int_{0}^{\infty} \Ex_{x \sim \mu} \sum_{a}  || ((Q_a^x)^{op} - (A_a^{\lambda,x})^{op} ) P_{\lambda}||_2^2 d\lambda \leq O(\delta^{\frac{1}{4}})$. By applying \Cref{lem:measurementtocoorlation} for each $\lambda$ (along with \Cref{lem:theviennalemma} to bound the synchronicity), \eqref{eq:roundingmainthmeq5} follows. By applying triangle inequality on \eqref{eq:roundingmainthmeq3}, \eqref{eq:roundingmainthmeq4} and \eqref{eq:roundingmainthmeq4}, the inequality \eqref{eq:roundingmainthmeq7} follows. By combining \eqref{eq:roundingmainthmeq1} and \eqref{eq:roundingmainthmeq7}, we have the equation below, which concludes the proof for \Cref{thm:MainRounding}, 
	\begin{equation*}
		\Ex_{(x,y) \sim \mu} \sum_{a,b}|\braket{\tau| \sigma A_a^x (B_y^b)^{op} \sigma |\tau} - \int_{0}^{\infty} P(\lambda) \cdot \braket{\tau| P_{\lambda} A_a^{\lambda, x}  ( A_b^{\lambda, y})^{op} P_{\lambda} |\tau}  d \lambda|  \leq O(\delta^{\frac{1}{8}}). 
	\end{equation*}
	For the ``furthermore" part of the theorem, by combining~\eqref{eq:roundingmainprojectordistance} with the definition of $P_{\lambda}$
	\begin{align*}
		&\int_{0}^{\infty} \Ex_{x \sim \mu} \sum_{a}  || (A_a^x - A_a^{\lambda,x} ) P_{\lambda}||_2^2 d \lambda \\
		&\leq \int_{0}^{\infty} \Ex_{x \sim \mu} \sum_{a}  || (Q_a^x - A_a^{\lambda,x} ) P_{\lambda}||_2^2 d \lambda + \int_{0}^{\infty} \Ex_{x \sim \mu} \sum_{a} \braket{\tau| (Q_a^x - A_a^x)^2 P_{\lambda}^2|\tau}d \lambda \\
		&\leq O(\delta^{\frac{1}{4}}) + \Ex_{x \sim \mu} \sum_{a} \braket{\tau|  (Q_a^x - A_a^x )^2 \left(\int_{0}^{\infty} P_{\lambda} d \lambda \right)|\tau} \\
		&= O(\delta^{\frac{1}{4}}) +  \Ex_{x \sim \mu} \sum_{a}  || (Q_a^x - A_a^x ) \sigma\ket{\tau}||_2^2  \\
		&= O(\delta^{\frac{1}{4}} + \delta), 
	\end{align*}
	where the second line follows from triangle inequality, the third line follows from \eqref{eq:roundingmainthmeq2} and the last line follows from \eqref{eq:roundingmainprojectordistance}. 
\end{proof}
\subsection{Proof for the applications of \Cref{thm:MainRounding}} \label{sec:roundingapp}

\subsubsection{Proof of \Cref{cor:corsoundness}}

We provide a proof to \Cref{cor:corsoundness}, and we remark that the proof follows a similar structure as \cite[Corollary 4.1]{vidickAlmostSynchronousQuantum2022}.

\begin{proof}
	We first assume the strategy $\strategy = (\cL^2(\alicealg, \tau),\sigma \ket{\tau}, \{A_a^x \})$ is symmetric, projective with $\sigma \in \alicealg^{+}$. Let $\{P_{\lambda}\}$ and $\{ A_a^{\lambda,x} \}$ be the set of projectors and PVMs promised by \Cref{prop:roundinginVNA}.	Let $\strategy^{\lambda} = ( \cL^2( P_{\lambda} \alicealg P_{\lambda}, P_{\lambda} \tau),  \frac{1}{\sqrt{\tau(P_\lambda)}} P_{\lambda} \ket{\tau},\{  A_a^{\lambda,x}  \})$ be the synchronous strategy from $\{ A_a^{\lambda,x} \}$. Given $\lambda \in [0,\infty)$ with $\tau(P_{\lambda}) \neq 0$. By the ``soundness property'' from the assumption, for each $\lambda \in [0, \infty)$, since each $\strategy^{\lambda}$ is synchronous, and $(P_{\lambda} \alicealg P_{\lambda})' = \alicealg' P_{\lambda}$,  there exists a family of POVM $\{ G_b^{\lambda, y} \}_{(y,b) \in \cY \times \cB} \subseteq P_{\lambda} \alicealg P_{\lambda}$ such that
	\begin{equation*}
		\frac{1}{\tau(P_{\lambda})} \Ex_{(x,y) \sim \rho} \sum_{a \in \cA}  \braket{ \tau| P_{\lambda}  A_a^{x, \lambda}  (G_{g_{xy}(a)}^{\lambda, y})^{op} P_{\lambda}| \tau} \geq \kappa(\omega^{co}(\cG, \strategy^{\lambda})),
	\end{equation*}
	where $op$ from above is defined on the tracial von Neumann algebra $(\alicealg, \tau)$. For each $(y, b) \in \cB \times \cY$, since $\int_{0}^{\infty} P_{\lambda} G_b^{\lambda, y} P_{\lambda} d \lambda \leq \sigma^2$, by \Cref{lem:inversetrick1}, there exist $H_{b}^y  \leq \cI$ such that $\sigma H_b^y \sigma = \int_{0}^{\infty} P_{\lambda} G_b^{\lambda, y} P_{\lambda} d \lambda$. We see that
	\begin{equation*}
		\sum_b \int_{0}^{\infty}  P_{\lambda} G_b^{\lambda, y} P_{\lambda} d \lambda = \int_{0}^{\infty} P_{\lambda}\left(
		\sum_b  G_b^{\lambda, y}  \right)P_{\lambda} d \lambda = \sigma^2, 
	\end{equation*}
	and hence, by a similar continuity argument as in \Cref{lem:inversetrick1}, $\{H_b^y\}$ forms a valid POVM in $\alicealg$. We argue that $\{H_b^y\}$ is the desired family of measurements for the corollary. We see that
	\begin{align}
		\nonumber \Ex_{(x,y) \sim \rho} \sum_{a \in \cA} \braket{ \tau| \sigma   A_a^{x}  (H_{g_{xy}(a)}^y)^{op} \sigma| \tau} &= 	\Ex_{(x,y) \sim \rho} \sum_{a \in \cA} \braket{ \tau| A_a^{x}   \sigma   H_{g_{xy}(a)}^y \sigma| \tau} \\
		&= \int_{0}^{\infty} P(\lambda)  \left(\frac{1}{\tau(P_{\lambda})} \Ex_{(x,y) \sim \rho} \sum_{a \in \cA} \braket{ \tau| P_{\lambda} A_a^{x}   (G_{g_{xy}(a)}^{\lambda,y})^{op} P_{\lambda} | \tau}\right)  d\lambda \label{eq:soundnesseq5}
	\end{align}
	where $P = \tau(P_{\lambda})$. We claim that
	\begin{equation} \label{eq:soundnesseq4}
		\int_{0}^{\infty} P(\lambda)  \left(\frac{1}{\tau(P_{\lambda})} \Ex_{(x,y) \sim \rho} \sum_{a \in \cA} \braket{ \tau| P_{\lambda} A_a^{x}   (G_{g_{xy}(a)}^{\lambda,y})^{op} P_{\lambda} | \tau}\right)  d\lambda \geq O(\kappa(\eps - \poly(\delta)) - \delta^{\frac{1}{8}}).
	\end{equation}
	To show the above inequality, we prove the following two inequalities:
	\begin{align}
		&\int_{0}^{\infty}P(\lambda) \cdot \left| \frac{1}{\tau(P_{\lambda})} \Ex_{(x,y) \sim \rho} \sum_{a \in \cA} \braket{ \tau| P_{\lambda} (A_a^{\lambda, x} - A_a^{x})  (G_{g_{xy}(a)}^{\lambda, y})^{op} P_{\lambda}| \tau} \right| d \lambda \leq O(\delta^{\frac{1}{8}}), \label{eq:soundnesseq1}\\
		\label{eq:soundnesseq2} &\int_{0}^{\infty} P(\lambda) \cdot \left(\frac{1}{\tau(P_{\lambda})} \Ex_{(x,y) \sim \rho} \sum_{a \in \cA} \braket{ \tau| P_{\lambda}  A_a^{\lambda, x}   (G_{g_{xy}(a)}^{\lambda, y})^{op} P_{\lambda}| \tau} \right) d \lambda \geq  \kappa\left( \omega^{co}(\cG, \strategy) - \poly(\delta)\right).  
	\end{align}
	For \eqref{eq:soundnesseq1}
	\begingroup
	\allowdisplaybreaks
	\begin{align*}
		&\int_{0}^{\infty}P(\lambda) \cdot \left| \frac{1}{\tau(P_{\lambda})} \Ex_{(x,y) \sim \rho} \sum_{a \in \cA} \braket{ \tau| P_{\lambda} (A_a^{\lambda, x} - A_a^{x})  (G_{g_{xy}(a)}^{\lambda, y})^{op} P_{\lambda}| \tau} \right| d \lambda \\
		&\leq \int_{0}^{\infty} \Ex_{(x,y) \sim \rho} \sum_{a \in \cA} \left| \braket{ \tau| P_{\lambda} (A_a^{\lambda, x} - A_a^{x})  (G_{g_{xy}(a)}^{\lambda, y})^{op} P_{\lambda}| \tau} \right| d \lambda \\
		&\leq \sqrt{\int_{0}^{\infty} \Ex_{(x,y) \sim \rho} \sum_{a \in \cA} \left| \braket{ \tau| P_{\lambda} (A_a^{\lambda, x} - A_a^{x})^2 P_{\lambda}| \tau} \right| d \lambda} \cdot \sqrt{\int_{0}^{\infty} \Ex_{(x,y) \sim \rho} \sum_{a \in \cA} \left| \braket{ \tau| P_{\lambda}  (G_{g_{xy}(a)}^{\lambda, y})^{op} P_{\lambda}| \tau} \right| d \lambda} \\
		&\leq \sqrt{\int_{0}^{\infty} \Ex_{x \sim \mu} \sum_{a \in \cA} \| (A_a^{\lambda, x} - A_a^{x}) P_{\lambda} \|_2^2 d \lambda} \cdot 1 \leq O(\delta^{\frac{1}{8}}).
	\end{align*}
	\endgroup
	For \eqref{eq:soundnesseq2}, by \Cref{thm:MainRounding}, $\int_{0}^{\infty} P(\lambda) \cdot \omega^{co}(\cG, \strategy^{\lambda}) d(\lambda) \geq \omega^{co}(\cG, \strategy) - \poly(\delta)$, since $\kappa$ is convex
	\begin{align*}
		\int_{0}^{\infty} P(\lambda) \cdot \left( \frac{1}{\tau(P_{\lambda})} \Ex_{(x,y) \sim \rho} \sum_{a \in \cA} \braket{ \tau| P_{\lambda}  A_a^{\lambda, x}  (H_{g_{xy}(a)}^y)^{op} P_{\lambda}| \tau} \right) d \lambda &\geq \kappa\left( \int_{0}^{\infty}  \omega^{co}(\cG, \strategy^{\lambda}) \right) \\
		&\geq \kappa\left( \omega^{co}(\cG, \strategy) - \poly(\delta)\right).
	\end{align*}
	Hence, by the triangle inequality, \eqref{eq:soundnesseq4} follows from \eqref{eq:soundnesseq1} and \eqref{eq:soundnesseq2}. Combining \eqref{eq:soundnesseq5} and \eqref{eq:soundnesseq4},
	\begin{equation*}
		\Ex_{(x,y) \sim \rho} \sum_{a \in \cA} \braket{ \tau| \sigma   A_a^{x}  (H_{g_{xy}(a)}^y)^{op} \sigma| \tau}  \geq  O \left( \kappa(\eps - \poly(\delta))  - \delta^{\frac{1}{8}} \right),
	\end{equation*}
	which concludes the corollary for the special case claimed above. 
	
	For the general case where $\strategy = (\cL^2(\alicealg, \tau),\sigma \ket{\tau}, \{A_a^x \}, \{(B_a^x)^{op} \})$ is a tracially embeddable strategy. By \Cref{lem:theviennalemma}, the symmetric strategy $\strategy^{'} = (\cL^2(\alicealg, \tau),\sigma \ket{\tau}, \{A_a^x \})$ is an $O(\delta)$-synchronous strategy, with
	\begin{equation*}
		\Ex_{(x,y) \sim \mu} \sum_{a,b} |\braket{\tau|\sigma A_a^x (B_b^y)^{op}\sigma|\tau} - \braket{\tau|\sigma A_a^x (A_b^y)^{op} \sigma|\tau} | \leq O(\sqrt{\delta}),
	\end{equation*}
	which implies $\omega^{co}(\cG, \strategy^{'}) \geq  \omega^{co}(\cG, \strategy) - O(\poly(\delta))$. By the special case proven above, there exists a family of POVMs $\{H_{g_{xy}^y (a) }\}$ such that 
	\begin{equation} \label{eq:soundnesseq3}
		\Ex_{(x,y) \sim \rho} \sum_{a \in \cA} \braket{ \tau|  \sigma A_a^{x}  (H_{g_{xy}(a)}^y)^{op} \sigma| \tau} \geq  O(\kappa(\omega^{co}(\cG, \strategy)  - \poly(\delta)) - \delta^{\frac{1}{8}}). 
	\end{equation}
	which concludes the proof for the general case. 
\end{proof}

\subsection{Proof of \Cref{cor:coralgrelation}}

We provide a proof to \Cref{cor:coralgrelation}, and we remark that the proof follows a similar structure as \cite[Corollary 4.4]{vidickAlmostSynchronousQuantum2022}.

\begin{proof}	Define $X(a)$ and $Z(b)$ the same way as the ''algebraic relations" property on the measurement $A_a^x$. Let $\{P_{\lambda}\}$ and $\{ A_a^{\lambda,x} \}$ be the set of projectors and PVMs promised by \Cref{prop:roundinginVNA}, and let $\strategy^{\lambda} = ( \cL^2( P_{\lambda} \alicealg P_{\lambda}, P_{\lambda} \tau),  \frac{1}{\sqrt{\tau(P_\lambda)}} P_{\lambda} \ket{\tau},\{  A_a^{\lambda,x}  \})$ be the synchronous strategy from $\{ A_a^{\lambda,x} \}$. Let $X^{\lambda}(a) = \sum_{j \in \Z_q} \omega_p^{bj} A_{j}^{\lambda, X}$ and similarly $Z^{\lambda}(a) = \sum_{j \in \Z_q} \omega_p^{bj} A_{j}^{\lambda, Z}$. For each $a \in \Z_q$
	\begin{equation}
		\int^{\infty}_{0}  P(\lambda) || (X(a) - X^{\lambda}(a)) P_{\lambda}||_2^2 d \lambda \leq  \int^{\infty}_{0}  P(\lambda) \sum_{j \in \Z_q}|| \omega_p^{bj}(A_x^X- A_x^{\lambda,X}) P_{\lambda} ||_2^2 d \lambda  \leq O(\delta^{\frac{1}{4}}),
	\end{equation}
	where the last inequality follows from \Cref{prop:roundinginVNA} and the assumption that $\mu(X, X) = O(1)$. Similarly,
	\begin{equation} \label{eq:coralgrelationshipZcomm}
		\int^{\infty}_{0}  P(\lambda) || (Z(a) - Z^{\lambda}(a)) P_{\lambda} ||_2^2 \leq O(\delta^{\frac{1}{4}}). 
	\end{equation}
	To show \Cref{eq:coralgrelation2}, we first note that via the triangle inequality and $\ket{\tau}$ is a tracial state
	\begin{align*}
		\Ex_{(a,b)} ||(X(a)Z(b) - \omega_p^{ab} Z(b)X(a))\sigma||_2^2 &=\Ex_{(a,b)} 
		\bra{\psi}\sigma (X(a)Z(b) - \omega_p^{ab} Z(b)X(a))\sigma \ket{\psi} \\
		&\leq \int_{0}^{\infty}  P(\lambda) \Ex_{(a,b) \in \Z_q \times \Z_q} ||(X(a)Z(b) - \omega_p^{ab} Z(b)X(a)) P_{\lambda}||_2^2 d \lambda.
	\end{align*}
	Also, by a similar argument as the proof of \Cref{cor:corsoundness}
	\begin{align*}
		\int_{0}^{\infty}  P(\lambda) \Ex_{(a,b)} ||(X^{\lambda}(a)Z^{^{\lambda}}(b) - \omega_p^{ab} Z^{\lambda}(b)X^{\lambda}(a))P_{\lambda}||_2^2 \leq \kappa(\int_{0}^{\infty}  \omega^{co}(\cG, \strategy^{\lambda}) \leq  \kappa(\omega^{co}(\cG, \strategy) + \poly(\delta)).
	\end{align*}
	Hence, it is sufficient to show the following to conclude the lemma:
	\begin{equation}\label{eq:coralgrelationmaineq}
		\int_{0}^{\infty} P(\lambda) \Ex_{(a,b)} ||\left ( (X(a)Z(b) - \omega_p^{ab}Z(b)X(a)) - (X^{\lambda}(a)Z^{^{\lambda}}(b) - \omega_p^{ab} Z^{\lambda}(b)X^{\lambda}(a)) \right)P_{\lambda} ||_2^2 d \lambda \leq O(\delta^{1/8}).
	\end{equation}
	Rearranging \Cref{eq:coralgrelationmaineq} via the triangle inequality
	\begin{align*}
		&\int_{0}^{\infty} P(\lambda) \Ex_{(a,b)} ||\left( (X(a)Z(b) - \omega_p^{ab}Z(b)X(a)) - (X^{\lambda}(a)Z^{^{\lambda}}(b) - \omega_p^{ab} Z^{\lambda}(b)X^{\lambda}(a)) \right)P_{\lambda} ||_2^2 \\
		&\leq \int_{0}^{\infty} P(\lambda) \Ex_{(a,b)} \left( || X(a)Z(b) -X^{\lambda}(a)Z^{^{\lambda}}(b)||_2^2 + || \omega_p^{ab} (Z^{\lambda}(b)X^{\lambda}(a)) - Z(b)X(a))||_2^2 \right) d \lambda \\
		&= \int_{0}^{\infty} P(\lambda) \Ex_{(a,b)}  || X(a)Z(b) -X(a)Z^{^{\lambda}}(b)||_2^2  + || X(a)Z^{^{\lambda}}(b) -X^{\lambda}(a)Z^{^{\lambda}}(b)||_2^2 \\
		&\qquad + ||(Z^{\lambda}(b)X^{\lambda}(a)) - Z^{\lambda}(b)X(a) ||_2^2 + ||( Z^{\lambda}(b)X(a) - Z(b)X(a)) ||_2^2  d \lambda. 
	\end{align*}
	Since the proof of the above four components is similar, we will only include the proof for the first term below:
	\begin{equation*}
		\int_{0}^{\infty} P(\lambda) \Ex_{(a,b)}  || X(a)Z(b) -X(a)Z^{^{\lambda}}(b)||_2^2 d \lambda \leq 1 \cdot \sqrt{\int_{0}^{\infty} P(\lambda) \Ex_{b}  ||Z(b) -Z^{^{\lambda}}(b)||_2^2 d\lambda } \leq O(\delta^{1/8}),
	\end{equation*}
	where the last inequality follows from \Cref{eq:coralgrelationshipZcomm}, completing the proof of the corollary. 
\end{proof}

\section{Proof of \Cref{thm:RigPaulibasis}} \label{sec:anticommPauliBasis}

In order to show \Cref{thm:RigPaulibasis}, we first show the following proposition about the observables defined within \eqref{eq:obserpaulibasis}.
\begin{proposition} \label{prop:anticommPauliBasis}
	Let $\cG(\mu)$ be the $\mu$-dependent Pauli Basis test and let $\strategy = (\cL^2(\alicealg, \tau),\sigma \ket{\tau}, \{A_a^x \}, \{ B_{a}^x \})$ be a projective, tracially embeddable strategy for $\cG(\mu)$ with $\omega^{co}(\cG(\mu), \strategy) \geq 1 - \eps$. Let $W^{A}$ to be the unitary observable define within \eqref{eq:obserpaulibasis}, then
	\begin{align*}
		&\Ex_{(u,v) \in (\mathbb{Z}_{2^n})^2}|| X^{A}(u) Z^{A}(v) - (-1)^{u \cdot v} Z^{A}(v) X^{A}(u)||_{\sigma^2}^2 \leq O\left( \poly(\kappa(\mu), \eps) \right), \\
		&\Ex_{(u,v) \in (\mathbb{Z}_{2^n})^2}|| X^{B}(u) Z^{B}(v) - (-1)^{u \cdot v} Z^{B}(v) X^{B}(u)||_{\sigma^2}^2 \leq   O\left( \poly(\kappa(\mu), \eps) \right).
	\end{align*}
\end{proposition}

\begin{proof}
	We first assume that the strategy $( \cL^2(\alicealg, \tau), \ket{\tau},\{ A_{a}^x \})$ is synchronous (and use $W(a)$ instead to represent the observables define on \eqref{eq:obserpaulibasis}). This implies that each $A_a^x$ are projective measurements. We wish first to show that
	\begin{equation} \label{eq:lem_Paulibasis}
		\Ex_{(u,v) \sim \mu \times \mu}\left \| X(u) Z(v) - (-1)^{u \cdot v} Z(u) X(v)\right \|_{2}^2 \leq O(\poly(\eps)).
	\end{equation}
	We remark that \eqref{eq:lem_Paulibasis} is proven implicitly in \cite[Proposition 3.8]{delasalleSpectralGapStability2022}. For $u \in \{0,1\}$ and $r \in \{0,1\}^{|r|}$, we define $A_{t|_{u} =  r}^{(\text{Pauli}, W)} = \sum_{t \in \{0,1\}^n, t|_{u} =  r} A_t^{(\text{Pauli}, W)}$. Hence, we can rewrite  \eqref{eq:obserpaulibasis} as
	\begin{equation*}
		W(u) = \sum_{\substack{r \in \{0,1\}^{|u|}\\|r| \text{ mod } 2 = 0}} A_{t|_{u} =  r}^{(\text{Pauli}, W)} - \sum_{\substack{r \in \{0,1\}^{|u|}\\|r| \text{ mod } 2 = 1}} A_{t|_{u} =  r}^{(\text{Pauli}, W)} 
	\end{equation*}
	Let $\{A_{t_{\text{(Pauli, X)}}}^{(\text{Coordinate, X}), u}\}_{(\text{(Pauli, X)})\in \{0,1\}^{|u|}}$ and $\{A_{t_{\text{(Pauli, Z)}}}^{(\text{Coordinate, Z}), v}\}_{(\text{(Pauli, Z)})\in \{0,1\}^{|v|}}$ be the measurement operator for the question label $(\text{Coordinate}, X)$ and $(\text{Coordinate}, Z)$ respectively.  Furthermore, for $i \in \{0,1\}$ and $W = \{X,Z\}$, define $A_{|t_W| = i}^{(\text{Coordinate, W}), u} = \sum_{\substack{t_W \in \{0,1\}^{|u|} \\ |t_W| \text{ mod } 2 = i}} A_{t_W}^{(\text{Coordinate, W}), u}$, and define the observable $W^{(\text{Coordinate}), u} = A_{|t_W| = 0}^{(\text{Coordinate, W}), u} - A_{|t_W| = 1}^{(\text{Coordinate, W}), u}$. Since in $\cG(\mu)$, the question pair on the black edge gets picked with a constant probability, the players can fail these questions with at most $O(\eps)$ probability over these two question pairs with strategy $\strategy$. This implies for $W \in \{X, Z\}$
	\begin{equation*}
		\Ex_{u \sim \mu} \sum_{t_W \in \{0,1\}^{|u|}} \bra{\tau}  A_{t|_{u} =  t_W}^{(\text{Pauli}, W)} A_{t_W}^{\text{(Coordinate, W)}, u} \ket{\tau} \leq 1 - O(\eps).
	\end{equation*}
	Rewriting the above equation
	\begin{align*}
		&\Ex_{u \sim \mu} \sum_{t_W \in \{0,1\}^{|u|}} || A_{t|_{u} = t_W}^{(\text{Pauli}, W)} - A_{t_W}^{\text{(Coordinate, W)}, u} ||_2^2 \\
		&= 2 - 2 \left(\Ex_{u \sim \mu} \sum_{t_W \in \{0,1\}^{|u|}} \bra{\tau}  A_{t|_{u} =  t_W}^{(\text{Pauli}, W)} A_{t_W}^{\text{(Coordinate, W)}, u} \right) \leq O(\eps), 
	\end{align*}
	and hence, by the triangle inequality
	\begin{equation} 
		\Ex_{u \sim \mu} || W(u) -W^{(\text{Coordinate}), u} ||_2^2 =  \Ex_{u \sim \mu} || \sum_{t_X \in \{0,1\}^{|u|}}(-1)^{|t_X|} \left(A_{t|_{u} =  t_W}^{(\text{Pauli}, W)} - A_{t_W}^{\text{(Coordinate, W)}, u}\right)  ||_2^2 \leq O(\eps). \label{eq:Paulibasisanticommlemc}
	\end{equation}
	In order to show the lemma, we first show
	\begin{equation} \label{eq:lem_PaulibasisCoor}
		\Ex_{(u,v) \sim \mu \times \mu}|| X^{(\text{Coordinate}), u} Z^{(\text{Coordinate}), v} - (-1)^{u \cdot v} Z^{(\text{Coordinate}), v} X^{(\text{Coordinate}), u}||_{2}^2 \leq O(\poly(\eps)),
	\end{equation}
	then use \Cref{eq:Paulibasisanticommlemc} to conclude the lemma. We consider the case where $u \cdot v = 0$ and $u \cdot v = 1$ for \eqref{eq:lem_PaulibasisCoor} separately. First, assume that $u \cdot v = 0$ in step 1 of the sampling procedure for the $\mu$ dependent Pauli Basis test. Let $\mu^2_0$ denote the distribution $\mu \times \mu$ conditioning on the output $u \cdot v = 0$, and let $\{A_{(t_X,t_Z)}^{(\text{Commutation}), u,v }\}_{(t_X, t_Z)\in \{0,1\}^{|u|+ |v|}}$ be the measurement operator corresponds to the ``commutation" question label with answer pair $(t_X, t_Z)$. Since $\strategy$ succeed with probability $1 - \eps$, the players can fail the ``commutation test" with probability at most $O(\eps)$, for $W \in \{X,Z\}$, this implies
	\begin{align}
		\Ex_{(u,v) \sim \mu^2_0} \sum_{t_X \in \{0,1\}^{|u|}}\bra{\tau}A_{(t_X,t_Z)}^{(\text{Commutation}), u,v } A_{t_W}^{(\text{Coordinate, W}), u} \ket{\tau} \leq 1- O(\eps). 
	\end{align}
	Rewriting the above inequality the same way as \eqref{eq:Paulibasisanticommlemc}, we have 
	\begin{align*}
		\Ex_{(u,v) \sim \mu^2_0} 	\sum_{t_X \in \{0,1\}^{|u|}} || A_{(t_X,t_Z)}^{(\text{Commutation}), u,v} - A_{t_W}^{(\text{Coordinate,W}),u} ||_2^2 &\leq O(\eps).
	\end{align*}
	Hence, by the triangle inequality, Jensen's inequality and Cauchy-Schwartz inequality, 
	\begin{align}
		\nonumber &\Ex_{(u,v) \sim \mu^2_0} \sum_{\substack{ t_X \in \{0,1\}^{|u|} \\ t_Z \in \{0,1\}^{|v|}}} ||  A_{t_X}^{(\text{Coordinate,X}), u} A_{t_Z}^{(\text{Coordinate Z}), v} -A_{t_Z}^{(\text{Coordinate Z}), v}  A_{t_X}^{(\text{Coordinate,X}), u} ||_2^2 \\
		\nonumber &\leq\Ex_{(u,v) \sim \mu^2_0}  \sum_{\substack{ t_X \in \{0,1\}^{|u|} \\ t_Z \in \{0,1\}^{|v|}}} ||   (A_{t_X}^{(\text{Coordinate,X}), u} - A_{(t_X,t_Z)}^{(\text{Commutation}), u,v}) A_{t_Z}^{(\text{Coordinate Z}) v} \\
		\nonumber &\quad +  A_{(t_X,t_Z)}^{(\text{Commutation}), u,v} (A_{t_Z}^{(\text{Coordinate Z}), v} - A_{(t_X,t_Z)}^{(\text{Commutation}), u,v}) \\
		\nonumber &\quad+ A_{(t_X,t_Z)}^{(\text{Commutation}), u,v}(A_{(t_X,t_Z)}^{(\text{Commutation}), u,v} - A_{t_X}^{(\text{Coordinate,X}), u})\\
		&\quad + (A_{(t_X,t_Z)}^{(\text{Commutation}), u,v} - A_{t_Z}^{(\text{Coordinate Z}), v} ) A_{t_X}^{(\text{Coordinate,X}), u} ||_2^2 \leq O(\sqrt{\eps}).\label{eq:Paulibasisanticommlemd}
	\end{align}
	Returning to \eqref{eq:lem_PaulibasisCoor}, for $W \in \{X,Z\}$ and $u \in\{0,1\}^n$, since $W^{(\text{Coordinate}), u}$ is a binary observable, we can rewrite it as
	\begin{align*}
		W^{(\text{Coordinate}), u} = \cI - 2 \cdot A_{|t_W| = 1}^{(\text{Coordinate, W}), u}. 
	\end{align*}
	Hence, we can rewrite \eqref{eq:lem_PaulibasisCoor} as
	\begin{align}
		\nonumber &\Ex_{(u,v) \sim \mu_0}|| X^{(\text{Coordinate}), u} Z^{(\text{Coordinate}), v} - Z^{(\text{Coordinate}), v} X^{(\text{Coordinate}), u}||_{2}^2 \\
		\nonumber &= 2\Ex_{(u,v) \sim \mu_0} || A_{|t_X| = 1}^{(\text{Coordinate, X}), u} A_{|t_Z| = 1}^{(\text{Coordinate, Z}), v} - A_{|t_Z| = 1}^{(\text{Coordinate, Z}), v} A_{|t_X| = 1}^{(\text{Coordinate, X}), u}||_{2}^2\\
		\nonumber	&\leq 2 \Ex_{(u,v) \sim \mu_0}  \sum_{\substack{ t_X \in \{0,1\}^{|u|} \\ |t_X| \text{ mod } 2 = 1\\ t_Z \in \{0,1\}^{|v|} \\ |t_Z| \text{ mod } 2 = 1}} ||  A_{t_X}^{(\text{Coordinate,X}), u} A_{t_Z}^{(\text{Coordinate Z}), v} -A_{t_Z}^{(\text{Coordinate Z}), v}  A_{t_X}^{(\text{Coordinate,X}), u} ||_2^2 \\
		&\leq 2 \Ex_{(u,v) \sim \mu_0}  \sum_{\substack{ t_X \in \{0,1\}^{|u|}\\ t_Z \in \{0,1\}^{|v|}}} ||  A_{t_X}^{(\text{Coordinate,X}), u} A_{t_Z}^{(\text{Coordinate Z}), v} -A_{t_Z}^{(\text{Coordinate Z}), v}  A_{t_X}^{(\text{Coordinate,X}), u} ||_2^2 \leq O(\sqrt{\eps}), \label{eq:Paulibasisanticommleme}
	\end{align}
	where the third line follows from $|| \cdot||_2 \geq 0$, and the last line follows from \eqref{eq:Paulibasisanticommlemd}. This concludes the case for  $u \cdot v = 0$. 
	Now we consider the case where $u \cdot v =1 $, Let $\mu^2_1$ denote the distribution $\mu \times \mu$ conditioning on the output $u \cdot v = 1$. Since $\strategy$ succeeds with probability $1 - \eps$, the players can fail the ``anti-commutation test" with probability at most $O(\eps)$. Since the anti-commutation test is a synchronous magic square game defined within \Cref{sec:MSgame}, with the question label ``variable 1" and ``variable 5" replaced by the question label $(\text{Coordinate}, X), u$ and $(\text{Coordinate}, X), v$ respectively. By \Cref{thm:RigMS}
	\begin{equation}
		\Ex_{(u,v) \sim \mu^2_1} ||X^{(\text{Coordinate}), u}	Z^{(\text{Coordinate}), v} + Z^{(\text{Coordinate}), v}	X^{(\text{Coordinate}), u}||_2^2 \leq O(\poly(\eps)), \label{eq:Paulibasisanticommlemf}
	\end{equation}
	which concludes the case for $u \cdot v =1$. Combining \eqref{eq:Paulibasisanticommleme} and \eqref{eq:Paulibasisanticommlemf} shows \eqref{eq:lem_PaulibasisCoor}. Combining \eqref{eq:lem_PaulibasisCoor} and \eqref{eq:Paulibasisanticommlemc} via the triangle inequality, we have
	\begin{equation*} 
		\Ex_{(u,v) \sim \mu \times \mu}|| X(u) Z(v) - (-1)^{u \cdot v} Z(u) X(v)||_{2}^2 \leq O(\poly(\eps)),
	\end{equation*}
	showing \eqref{eq:obserpaulibasis}.  Given that each $\{A_{q}^{\text{Pauli}, W}\}$ is a projective measurement for $W \in \{X,Z\}$, the observable $\{ W(u) \}_{u \in \{0,1\}}$ is an (exact) unitary representation for the group $\bZ_{2^n}$. Hence, we apply \Cref{thm:Ponicarethm} on \eqref{eq:lem_Paulibasis} to obtain
	\begin{equation}
		\Ex_{(u,v)}\left \| X(u) Z(v) - (-1)^{u \cdot v} Z(u) X(v) \right \|_{2}^2 \leq \kappa(\mu)^2 \cdot O(\poly(\eps)),
	\end{equation}
	where the expectation is over the uniform distribution over $\{0,1\}^{2n}$. This concludes the lemma for the case where $\strategy$ is synchronous. 
	
	For the general case, note both the label $(\text{Pauli}, X)$ and $(\text{Pauli}, Z)$ appear as a synchronicity question pair with probability $O(1)$ within $\cG(\mu)$. Thus, we can use \Cref{cor:coralgrelation} to extend this result to the general case, concluding the proof for this proposition. 
\end{proof}

Using \Cref{prop:anticommPauliBasis}, we are ready to give a proof for \Cref{thm:RigPaulibasis}. Note that this proof follows a similar structure as \cite[Lemma A.24]{jiMIPRE2022}. 

\begin{proof}
	By using \Cref{lem:orthogonalizationlemma}, we can assume that the measurement operators for $\strategy$ are PVMs with $O(\poly(\eps))$ overhead. With this assumption, for $W \in \{X,Z\}$ and $P \in \{A,B\}$, $\{W^{P}(u) \}$ forms an unitary representation for the group $\bZ_{2^n}$. This implies for all $u,v \in \bZ_{2^n}$
	\begin{equation*}
		W^{P}(u)W^{P}(v) = W^{P}(u + v).
	\end{equation*}
	Combine the above equation with \Cref{prop:anticommPauliBasis}, we see that $\{X^{P}(u)\} \cup \{Z^{P}(u)\} $ forms an $( O\left(\poly(\kappa(\mu), \eps )\right), \sigma^2)$-representation for the group $H^{(n)}$. Hence, by \Cref{thm:SGHinTracialVNA} (and the structure of the left regular representation for the group $H^{(n)}$), there exist two isometries $V_A: \cH \rightarrow  \cH \tensor \bC^{2^{2n}} $ and $V_B: \cH \rightarrow  \cH \tensor \bC^{2^{2n}} $ such that for all $(a,b) \in \bZ_{2^n}$
	\begin{align}
		||X^{A}(a) Z^{A}(b) - V_A^* (  \cI_{\cH} \tensor \rho^X(a) \rho^Z(b) \tensor \cI_{2^{2n}}    ) V_A ||_{\sigma^2} &\leq O(\poly(\kappa(\mu), \eps)),  \\
		||X^{B}(a) Z^{B}(b) - V_B^* ( \cI_{\cH} \tensor  \rho^X(a) \rho^Z(b) \tensor \cI_{2^{2n}} ) V_B ||_{\sigma^2} &\leq O(\poly(\kappa(\mu), \eps)), \label{eq:PaulibasisGH2}
	\end{align}
	We note that on the second item of the above equation, \Cref{thm:SGHinTracialVNA} is applied on the \textit{commutant} $\alicealg^{'}$ rather than $\alicealg$. We wish to first show \eqref{eq:robustnessthmstat1} by using the synchronous condition to connect the two exact representations. By the self-consistency test (self-loop) appears with $O(1)$ probability, for all $W \in \{X,Z\}$, we have 
	\begin{equation*}
		\sum_{t \in\{0,1\}} \braket{\tau|\sigma A_{t}^{(\text{Pauli}, W)}  B_{t}^{(\text{Pauli}, W)}  \sigma|\tau} \geq 1 - O(\eps). 
	\end{equation*}
	Hence, via the triangle inequality, for all $a \in \{0,1\}^n$
	\begin{equation}
		\braket{\tau|\sigma W^{A}(a)  W^{B}(a)  \sigma|\tau} \geq 1 - O(\eps), \label{eq:Paulibasisconsistency}
	\end{equation}
	where we see that $X^B(a) \in \alicealg$ in this formulation. We will first consider \eqref{eq:robustnessthmstat1}, we wish to first show
	\begin{align}
		\nonumber &\braket{\tau | \sigma V_A^* (  V_B^* \tensor \cI_{2^{2n}}^{A})  (\cI_{\cH} \tensor \left(\ketbra{EPR}{EPR}^{\tensor n}\right)^{A_1 B_1} \tensor\cI_{2^{2n}}^{A_2 B_2} ) ( V_B \tensor \cI_{2^{2n}}^{A} ) V_A \sigma| \tau} \\
		&\quad \geq 1 - O(\poly(\kappa(\mu), \eps)). \label{eq:robustnessthmmaineq1}
	\end{align}
	Where, $\cH^A = \bC^{2^{2n}}$ is the Hilbert space created by the isometry $V^A$ and $\cH^B = \bC^{2^{2n}}$ are the space created by the isometry $V^B$. For $W \in \{X,Z\}$, denote $H_W = \Ex_{a \sim \{0,1\}^n}\rho^W(a) \tensor \rho^W(a)$. We see that
	\begin{align*}		 
		\nonumber &1 - \braket{\tau | \sigma V_A^* ( V_B^* \tensor \cI_{2^{2n}}^{A} )  ( \cI_{\cH} \tensor H_W^{A_1 B_1} \tensor\cI_{2^{n}}^{A_2} \tensor \cI_{2^{n}}^{B_2} ) (V_B \tensor \cI_{2^{2n}}^{A} ) V_A \sigma| \tau} \\
		\nonumber &= 1 - \Ex_{a \sim \{0,1\}^n} \braket{\tau | \sigma V_A^* (V_B^* \tensor \cI_{2^{2n}}^{A} ) \cI_{\cH} \tensor \left( \rho^W(a) \tensor \cI_{2^n}\right)^{A} \tensor  \left(\rho^W(a) \tensor\cI_{2^n}\right)^{B}  (V_B \tensor \cI_{2^{2n}} ) V_A \sigma | \tau} \\
		\nonumber &= 1 - \Ex_{a \sim \{0,1\}^n} \braket{\tau |  \sigma V_A^*  \left(\cI_{\cH} \tensor \rho^W(a) \tensor\cI_{2^n}  \right) V_A V_B^*  \left( \cI_{\cH} \tensor \rho^W(a) \tensor\cI_{2^n} \right) V_B \sigma| \tau} \\	
		\nonumber &= \left(1  - \Ex_{a \sim \{0,1\}^n} \braket{\tau|\sigma W^{A}(a)  W^{B}(a)  \sigma|\tau} \right)  \\ 
		\nonumber &\quad + \Ex_{a \sim \{0,1\}^n} \braket{\tau |  \sigma \left(W^A(a) - V_A^*  \left( \cI_{\cH} \tensor \rho^W(a) \tensor\cI_{2^n}  \right) V_A  \right)	V_B^*  \left(  \cI_{\cH} \tensor \rho^W(a) \tensor\cI_{2^n} \right) V_B \sigma| \tau}  \\
		\nonumber &\quad + \Ex_{a \sim \{0,1\}^n} \braket{\tau | \sigma  W^A(a)	\left(  W^B(a) - V_B^*  \left( \cI_{\cH} \tensor \rho^W(a) \tensor\cI_{2^n}  \right) V_B \right) \sigma| \tau} \\ 
		\nonumber&\leq O(\poly(\kappa(\mu), \eps)),
	\end{align*}
	where the last line follows from \eqref{eq:PaulibasisGH2} and \eqref{eq:Paulibasisconsistency}. We note that the braket around the tensor product in the above equation helps keep track of the register of each space. Hence, we have 
	\begin{equation}
		\braket{\tau | \sigma V_A^* ( V_B^* \tensor \cI_{2^{n}}^{A} )  ( \cI_{\cH}\tensor H_W^{A_1 B_1} \tensor\cI_{2^{n}}^{A_2} \tensor \cI_{2^{n}}^{B_2} ) ( V_B \tensor \cI_{2^{2n}}^{A} ) V_A \sigma| \tau} \geq 1- O(\poly(\kappa(\mu), \eps)).  \label{eq:EPRtestpt1cross}
	\end{equation}
	Recall that
	\begin{equation*}
		H_X H_Z = \Ex_{(a,b) \in \bZ_{2^n}} \rho^X(a) \rho^Z(b) \tensor  \rho^X(a) \rho^Z(b) = \left(\ketbra{EPR}{EPR}\right)^{\tensor n}. 
	\end{equation*}
	Hence to show \eqref{eq:robustnessthmmaineq1}, it is sufficient to show that 
	\begin{equation*}
		\braket{\tau | \sigma V_A^* (V_B^* \tensor \cI_{2^{2n}}^{A} )  ( \cI_{\cH}\tensor  \left( H_X H_Z \right)^{A_1 B_1} \tensor\cI_{2^{2n}}^{A_2 B_2} ) ( V_B \tensor\cI_{2^{2n}}^{A} ) V_A \sigma| \tau} \geq 1 - O(\poly(\kappa(\mu), \eps)).
	\end{equation*}
	or 
	\begin{align}
		\nonumber&\|   (\cI_{\cH} \tensor H_Z^{A_1 B_1} \tensor\cI_{2^{2n}}^{A_2 B_2} ) (V_B \tensor \cI_{2^{2n}}^{A} ) V_A \sigma\ket{\tau} - (\cI_{\cH} \tensor H_X^{A_1 B_1} \tensor\cI_{2^{2n}}^{A_2 B_2} ) ( V_B \tensor \cI_{2^{2n}}^{A} ) V_A \sigma\ket{\tau} \| \\ &\quad \leq  O(\poly(\kappa(\mu), \eps)). \label{eq:robustnessthmmaineq2}
	\end{align}
	To show \eqref{eq:robustnessthmmaineq2}, by the triangle inequality, we have
	\begin{align*}
		&\|  \left( \cI_{\cH} \tensor H_Z^{A_1 B_1} \tensor\cI_{2^{2n}}^{A_2 B_2} \right) (V_B \tensor \cI_{2^{2n}}^{A}) V_A \sigma\ket{\tau} -   \left(\cI_{\cH} \tensor H_X^{A_1 B_1} \tensor\cI_{2^{2n}}^{A_2 B_2} \right) (V_B \tensor \cI_{2^{2n}}^{A} ) V_A \sigma\ket{\tau} \| \\
		&\leq  \|  ( \cI_{\cH} \tensor H_Z^{A_1 B_1} \tensor\cI_{2^{2n}}^{A_2 B_2}) (V_B \tensor \cI_{2^{2n}}^{A} ) V_A \sigma\ket{\tau} - ( V_B \tensor \cI_{2^{2n}}^{A} ) V_A \sigma\ket{\tau} \| \\
		&\quad  +  \|  ( V_B \tensor \cI_{2^{2n}}^{A} ) V_A \sigma\ket{\tau} - (H_X^{A_1 B_1} \tensor\cI_{2^{2n}}^{A_2 B_2} \tensor \cI_{\cH})( V_B \tensor \cI_{2^{2n}}^{A} ) V_A \sigma\ket{\tau} \| \\
		&\leq O(\poly(\kappa(\mu), \eps)),
	\end{align*}
	where the last line follows from \eqref{eq:EPRtestpt1cross}. This shows \eqref{eq:robustnessthmmaineq1}. Let
	\begin{equation*}
		\ket{\widehat{aux}} = \left( \cI_{\cH} \tensor (\bra{EPR}^{\tensor n})^{A_1 B_1} \tensor\cI_{2^{2n}}^{A_2 B_2}\right) (V_B \tensor \cI_{2^{2n}}^{A} ) V_A \sigma\ket{\tau},
	\end{equation*}
	and by setting $\ket{aux} = \frac{1}{\| \ket{\widehat{aux}} \| } \ket{\widehat{aux}}$, we conclude the argument for \eqref{eq:robustnessthmstat1}. 
	
	Now we wish to show \eqref{eq:robustnessthmstat2}, since the proof for $P = A$ is the same as $P = B$, we show only the case where $P=A$ below. By labelling the space and rearranging the equation
	\begingroup
	\allowdisplaybreaks
	\begin{align}
		\nonumber &\| \left( (V_{A} W^A(u) V_{A}^{*} \tensor (\cI_{2^{2n}})^{B_1 B_2})^{\cH A_1 A_2} - \cI_{\cH} \tensor (\rho^W(u))^{A_1} \tensor (\cI_{2^{3n}})^{A_2 B_1 B_2} \right) \ket{\text{Aux}} (\ket{\text{EPR}}^{\tensor n})^{A_1 B_1}  \|^2 \\
		\nonumber &\leq 2 -2 \bra{\text{Aux}} \bra{\text{EPR}}^{\tensor n} \left( (V_{A}  W^A(u) V_{A}^{*})(\cI_{\cH} \tensor \rho^W(u) \tensor \cI_{2^n})^{\cH A_1 A_2} \right) \tensor (\cI_{2^{2n}})^{B_1 B_2} \ket{\text{Aux}}  \ket{\text{EPR}}^{\tensor n}  \\
		\nonumber &= 2 -2 \bra{\text{Aux}} \bra{\text{EPR}}^{\tensor n} \left( (V_{A}  W^A(u) V_{A}^{*})\right) \tensor (\rho^W(u) \tensor \cI_{2^{n}})^{B_1 B_2} \ket{\text{Aux}}  \ket{\text{EPR}}^{\tensor n}  \\
		\nonumber&\leq 2 -2 \bra{\tau} \sigma V_A^* (V_B^* \tensor \cI_{2^{2n}}) \left( (V_{A}  W^A(u)  V_{A}^{*})\right) \tensor (\rho^W(u) \tensor \cI_{2^{n}})^{B_1 B_2}  (V_B \tensor \cI_{2^{2n}}) V_A \sigma \ket{\tau} \\
		\nonumber&\qquad   + O(\poly(\kappa(\mu), \eps))\\
		\nonumber&\leq 2 -2 \bra{\tau} \sigma  W^A(u)   V_B^* ( \cI_{\cH} \tensor  \rho^X(a) \rho^Z(b) \tensor \cI_{2^{2n}} )V_B \sigma \ket{\tau} + O(\poly(\kappa(\mu), \eps)) \\
		\nonumber &\leq 2 -2 \bra{\tau} \sigma  W^A(u)  W^B(u) \sigma \ket{\tau} + O(\poly(\kappa(\mu), \eps)) \\
		&\leq  O(\poly(\kappa(\mu), \eps)), \label{eq:robustnessthmmaineq3}
	\end{align}
	\endgroup
	where the second line follows from $\rho^W(u) \tensor \cI_{2^n}\ket{EPR} = \cI_{2^n} \tensor \rho^W(u)  \cI_{2^n}\ket{EPR}$. The fourth line of the above equation follows from Cauchy-Schwartz, \eqref{eq:robustnessthmstat1} and for two arbitrary unit vector $\ket{\psi}, \ket{\phi} \in \cH$ such that $\| \ket{\psi} - \ket{\phi} \|^2 \leq \eps$ for some $\eps > 0$ and $A \leq \cI$
	\begin{align*}
		\braket{\psi|A| \psi} - \braket{\phi|A| \phi} &=  (\bra{\psi} - \bra{\phi}) A \ket{\psi} + \bra{\phi} A (\ket{\psi} - \ket{\phi}) \\
		&\leq \| \ket{\psi} - \ket{\phi} \| \cdot (\| A \ket{\psi} \| + \| A \ket{\phi} \| ) \leq \sqrt{\eps} \cdot 1.
	\end{align*}
	The fifth line on \eqref{eq:robustnessthmmaineq3} follows from the commutation of the isometry, the sixth line follows from \eqref{eq:PaulibasisGH2} and the Cauchy-Schwartz's inequality, and the last line follows from \eqref{eq:Paulibasisconsistency}. This concludes the proof for \Cref{thm:RigPaulibasis}. 
	
	%Hence, the remainder of the proof will follow a similar structure as \cite[Theorem 4.9]{broadbentQuantumDelegationOfftheshelf2023} with the state-dependent Gowers-Hatami's theorem being replaced with \Cref{thm:SGHinTracialVNA}. 
\end{proof}

\end{document}